\documentclass[11pt,english]{article}
\usepackage{geometry}
\usepackage{setspace}
\usepackage[T1]{fontenc}
\usepackage[utf8]{inputenc}
\usepackage{babel}
\usepackage{varioref}
\usepackage{amsmath,nccmath}
\numberwithin{equation}{section}
\usepackage{amssymb}
\usepackage{graphicx}
\usepackage{enumerate}
\usepackage{floatrow}
\usepackage[labelfont=bf]{caption}
\usepackage{appendix}
\usepackage{slashed}
\usepackage{svg}
\usepackage [autostyle, english = american]{csquotes}
\MakeOuterQuote{"}
\makeatletter
\usepackage{hyperref}
\usepackage{xcolor}
\usepackage{authblk}
\usepackage{amsthm}
\usepackage{graphicx}
\usepackage{hyperref}
\hypersetup{
    colorlinks,
    linkcolor={blue!60!black},
    citecolor={green!60!black},
    urlcolor={blue!80!black}
}
\usepackage{mathtools}
\usepackage{comment}
\allowdisplaybreaks
\newtheorem{thm}{Theorem}[section]
\newtheorem{conj}{Conjecture}
\newtheorem{lemma}{Lemma}[section]
\newtheorem{defi}{Definition}[section]

\newtheorem{cor}{Corollary}[section]
\theoremstyle{remark}
\newtheorem{rem}{Remark}[section]
\theoremstyle{plain}
\newtheorem{prop}{Proposition}[section]
\newtheorem{claim}{Claim}
\newenvironment{nalign}{
    \begin{equation}
    \begin{aligned}
}{
    \end{aligned}
    \end{equation}
    \ignorespacesafterend
}
\theoremstyle{remark}
\renewcommand{\labelenumi}{\alph{enumi}.} 
\usepackage[disable]{todonotes}
\newtheorem{subclaim}{Sublemma}
\newcommand{\dd}{\mathop{}\!\mathrm{d}}
\newcommand{\pu}{\partial_u}
\newcommand{\pv}{\partial_v}
\renewcommand{\O}{\mathcal{O}}

\newcommand{\p}{\phi_1}
\newcommand{\cc}{C_{\mathrm{in}}^{(1)}}
\newcommand{\ccc}{C_{\mathrm{in}}^{(2)}}
\newcommand{\pho}{(r\phi_1)}
\newcommand{\IL}[1]{I_{\ell=#1}^{\mathrm{future},f}}
\newcommand{\ILlog}[1]{I_{\ell=#1}^{\mathrm{future},\frac{\log r}{r^3}}}
\newcommand{\ILp}[1]{I_{\ell=#1}^{\mathrm{past},f}}
\newcommand{\ILn}[1]{I_{\ell=#1}^{\mathrm{future}}}
\newcommand{\ILpn}[1]{I_{\ell=#1}^{\mathrm{past}}}
\renewcommand{\div}{\mathrm{div}}

\title{The Case Against Smooth Null Infinity III:\\Early-Time Asymptotics for Higher $\ell$-Modes of Linear Waves on a Schwarzschild Background} 

\author[1]{Lionor M. A. Kehrberger\thanks{kehrberger@mis.mpg.de}} 
\affil[1]{University of Cambridge, Department of Applied Mathematics and Theoretical Physics,

 Wilberforce Road, Cambridge CB3 0WA, United Kingdom}
\date{April 30, 2024} 

\makeatother
\begin{document}
\pagenumbering{roman}

\maketitle 
\begin{abstract}
In this paper, we derive the early-time asymptotics for fixed-frequency solutions $\phi_\ell$ to the wave equation $\Box_g \phi_\ell=0$ on a fixed Schwarzschild background ($M>0$) arising from the no incoming radiation condition on $\mathcal I^-$ and polynomially decaying data, $r\phi_\ell\sim t^{-1}$ as $t\to-\infty$, on either  a timelike boundary of constant area radius $r>2M$ \textbf{(I)} or an ingoing null hypersurface \textbf{(II)}.
In case \textbf{(I)}, we show that the asymptotic expansion of $\pv(r\phi_\ell)$ along outgoing null hypersurfaces near spacelike infinity $i^0$ contains logarithmic terms at order $r^{-3-\ell}\log r$.  
In contrast, in case \textbf{(II)}, we obtain that the asymptotic expansion of $\pv(r\phi_\ell)$ near spacelike infinity $i^0$ contains logarithmic terms already at order $r^{-3}\log r$ (unless $\ell=1$).

These results suggest an alternative approach to the study of late-time asymptotics near future timelike infinity $i^+$ that does not assume conformally smooth or compactly supported Cauchy data:
In case \textbf{(I)}, our results indicate a \textit{logarithmically modified Price's law} for each $\ell$-mode. 
On the other hand, the data of case \textbf{(II)} lead to much stronger deviations from Price's law. 
In particular, we conjecture that compactly supported scattering data on $\mathcal H^-$ and $\mathcal I^-$ lead to solutions that exhibit the same late-time asymptotics on $\mathcal I^+$ for each $\ell$: $r\phi_\ell|_{\mathcal I^+}\sim u^{-2}$ as $u\to\infty$. 
\end{abstract}
\newpage
    \begingroup
\hypersetup{linkcolor=black}
    \tableofcontents{}
    \endgroup
\newpage
\pagenumbering{arabic}
\section{Introduction}
\subsection{Background and Motivation}
In this paper, we study the early-time asymptotics, i.e.\ asymptotics near spatial infinity $i^0$, of solutions, localised on a single angular frequency $\ell=L$, to the wave equation
\begin{equation}\label{waveequation}
\Box_{g_M}\phi_{\ell=L}=0
\end{equation}
on the exterior of a fixed Schwarzschild (or a more general spherically symmetric) background $(\mathcal{M}_M,g_M)$ under certain assumptions on data near past infinity. 
The most important of these assumptions is the \textit{no incoming radiation condition} on $\mathcal{I}^-$, stating that the flux of the radiation field on past null infinity vanishes at late advanced times. In addition, we will assume polynomially decaying (boundary) data on \textit{either} a past-complete timelike hypersurface, \textit{or} a past-complete null hypersurface.
\subsubsection{The spherically symmetric mode}\label{sec:subsec:intro:sphsym}
We initiated the study of such data in~\cite{I}, where we constructed \textit{spherically symmetric} solutions arising from the no incoming radiation condition, as a condition on data on $\mathcal{I}^-$, and polynomially decaying boundary data on a timelike hypersurface $\Gamma$ terminating at $i^-$ (or polynomially decaying characteristic initial data on an ingoing null hypersurface $\mathcal{C}_{\mathrm{in}}$ terminating at $\mathcal{I}^-$). 

The choice for these data, in turn, was motivated by an argument due to D.\ Christodoulou~\cite{Chr02} (based on the monumental proof of the stability of the Minkowski space~\cite{CK93}), which showed that the assumption of \textit{Sachs peeling}~\cite{SeriesVI,SeriesVIII} and, thus, of \textit{(conformally) smooth null infinity}~\cite{Penrose65} is incompatible with the no incoming radiation condition and the prediction of the quadrupole formula for $N$ infalling masses from $i^-$. 
The latter predicts that the rate of change of gravitational energy along $\mathcal I^+$ is given by $\sim -1/|u|^4$ near $i^0$. 
Indeed, modelling gravitational radiation by scalar radiation, we showed in~\cite{I} that the data described above lead to solutions which not only agree with the prediction of the quadrupole approximation (namely that $r^2(\pu\phi)^2|_{\mathcal I^+}\sim |u|^{-4}$ near $i^0$), but also have logarithmic terms in the asymptotic expansion of the \textit{spherically symmetric mode} $\pv(r\phi_0)$ as $\mathcal{I}^+$ is approached, thus contradicting the statement of Sachs peeling that such expansions are analytic in $1/r$.
More precisely, we obtained for the spherically symmetric mode $\phi_0$ that if the limit
\begin{equation}
\lim_{\mathcal{C}_{\mathrm{in}},u\to-\infty}|u|r\phi_0:=\Phi^-
\end{equation}
on initial data is non-zero (or, in the timelike case, if a similar condition on $\Gamma$ holds), then it is, in fact, a conserved quantity along $\mathcal I^-$, and,  for sufficiently large negative values of $u$, one obtains on each outgoing null hypersurface of constant $u$ the asymptotic expansion
\begin{equation}\label{l:eq:intro:pvrphilog}
\pv(r\phi_0)(u,v)=-2M\Phi^- \frac{\log r-\log|u|}{r^3}+\mathcal{O}(r^{-3}).
\end{equation}

In wide parts of the literature, it has been (and still is) assumed that physically relevant spacetimes do possess a smooth null infinity and that, therefore, logarithms as in \eqref{l:eq:intro:pvrphilog} do not appear. 
The result of \cite{I}, in line with \cite{CK93}, thus further puts this assumption in doubt. 
Furthermore, we showed in \cite{II} that the failure of peeling manifested by the \emph{early}-time asymptotics \eqref{l:eq:intro:pvrphilog} translates into logarithmic \emph{late}-time asymptotics near $i^+$, providing evidence for the physical measurability of the failure of null infinity to be smooth. We will return to the discussion of late-time asymptotics in section~\ref{sec:intro:NP}.

For more background on the history and relevance of peeling and smooth null infinity, we refer the reader to the introduction of \cite{I}. 

Finally, we note that the results from~\cite{I} were, in fact, obtained for the non-linear Einstein-Scalar field system ($G_{\mu\nu}[g]=T_{\mu\nu}^{sf}[\phi_0]$) under spherical symmetry and then, \textit{a fortiori,} carried over to the linear case ($G_{\mu\nu}[g]=0$, $\Box_g \phi_0=0 $). 

\subsubsection{Higher \texorpdfstring{$\ell$}{l}-modes}
Ultimately, we would like to develop an understanding of the situation for the \textit{Einstein vacuum equations} without symmetry assumptions (for which the spherically symmetric Einstein-Scalar field system only served as a toy model) in order to close the circle to Christodoulou's original argument~\cite{Chr02}, which was an argument pertaining to gravitational, not scalar, radiation.
In particular, we would like to understand the prediction of the quadrupole approximation, namely that the rate of gravitational energy loss along $\mathcal I^+$ is given by $-1/|u|^4$ as $u\to-\infty$, \textit{dynamically}, i.e.\ arising from suitable scattering data, rather than imposing it on $\mathcal I^+$ as was done in~\cite{Chr02}.
In view of the multipole structure of gravitational radiation, it thus seems to be necessary to first understand the answer to the following question: 

\textit{What are the early-time asymptotics for higher $\ell$-modes of solutions to the wave equation $\Box_g \phi=0$ on  a fixed Schwarzschild background, arising from the no incoming radiation condition, i.e.,\ what is the analogue of \eqref{l:eq:intro:pvrphilog} for $\ell>0$? }

 We shall provide a detailed answer to this question in this paper. Let us already paraphrase two special cases of the main statements (which are summarised in section~\ref{sec:intro:summary}). Statement \textit{\textbf{1)}} below corresponds to Theorems~\ref{thm:intro:gtl},~\ref{thm:intro:gnl}, and Statement \textit{\textbf{2)}} corresponds to Theorem~\ref{thm:intro:moreg}.
\begin{figure}[htbp]
\floatbox[{\capbeside\thisfloatsetup{capbesideposition={right,top},capbesidewidth=4.4cm}}]{figure}[\FBwidth]
{\caption{Schematic depiction of the data setup considered in \textbf{\textit{1)}}: We consider polynomially decaying data on a spherically symmetric timelike hypersurface $\Gamma$, and vanishing data on $\mathcal I^-$. The latter condition is to be thought of as the no incoming radiation condition. }\label{fig:III:-1}}
{\includegraphics[width=200pt]{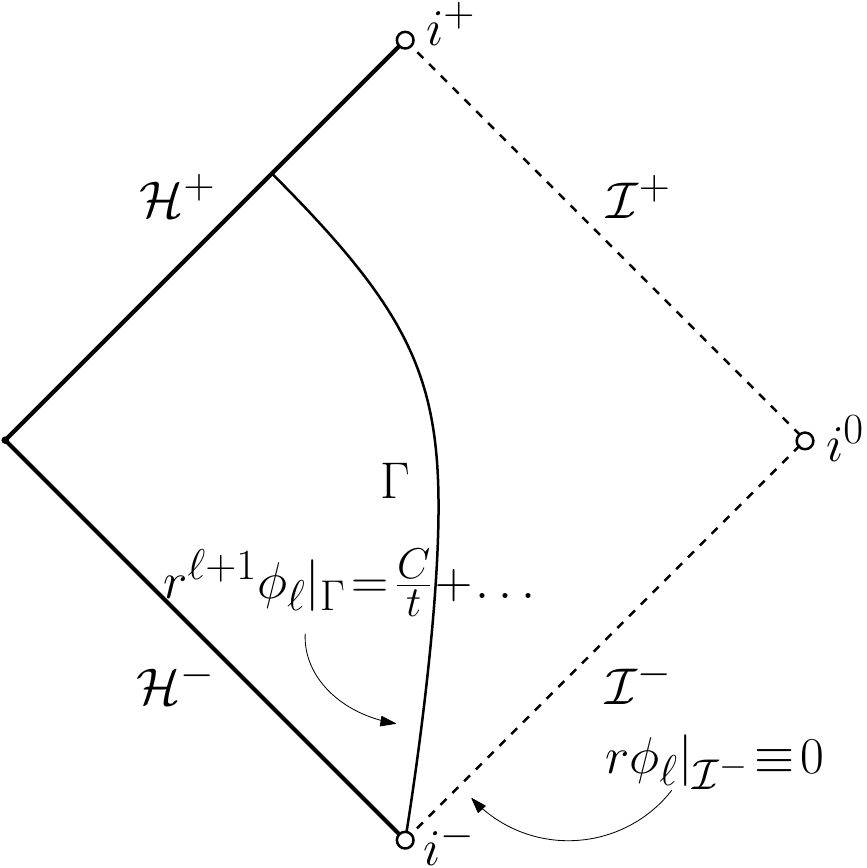}}
\end{figure}
\par\smallskip\noindent
\centerline{\begin{minipage}{0.9\textwidth}
\textit{\textbf{1)} Consider solutions $\phi_\ell$ to \eqref{waveequation} arising from 
 polynomially decaying data ${r^{\ell+1}}\phi_\ell\sim |t|^{-1}$ as $t\to-\infty$ on a spherically symmetric timelike hypersurface $\Gamma$ and the no incoming radiation condition on $\mathcal I^-$. 
 (See Figure~\ref{fig:III:-1}.) 
 Then, schematically, $r\phi_\ell|_{\mathcal I^+}\sim |u|^{-\ell-1}$ along $\mathcal I^+$ as $u\to-\infty$, and the asymptotic expansion of $\pv(r\phi_\ell)$ along outgoing null hypersurfaces of constant $u$ near spacelike infinity $i^0$ reads:
\begin{equation}\label{eq:intro:scheme1}
\pv(r\phi_\ell)=\frac{f_0(u)}{r^2}+\dots+\frac{f_\ell(u)}{r^{2+\ell}}+C\frac{\log r}{r^{3+\ell}}+\dots,
\end{equation}
where $C$ is a non-vanishing constant.
} \end{minipage}}
\par\smallskip
\begin{figure}[htbp]
\floatbox[{\capbeside\thisfloatsetup{capbesideposition={right,top},capbesidewidth=4.4cm}}]{figure}[\FBwidth]
{\caption{Schematic depiction of the data setup considered in \textbf{\textit{2)}}: We consider polynomially decaying data on a spherically symmetric ingoing null hypersurface $\mathcal C_{\mathrm{in}}$, and vanishing data on the part of $\mathcal I^-$ that lies to the future of $\mathcal C_{\mathrm {in}}$. The latter condition is to be thought of as the no incoming radiation condition.}\label{fig:III:0}}
{\includegraphics[width=200pt]{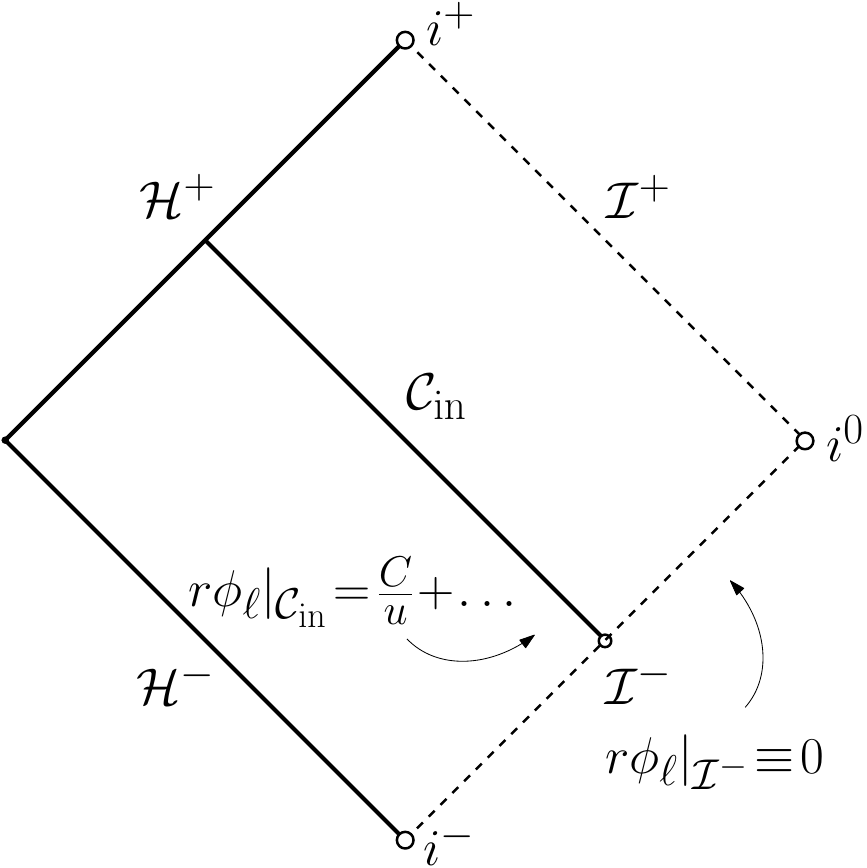}}
\end{figure}
\par\smallskip\noindent
\centerline{\begin{minipage}{0.9\textwidth}
\textit{\textbf{2)}
Alternatively, consider solutions $\phi_\ell$ to \eqref{waveequation} arising from 
 polynomially decaying data $r\phi_\ell\sim |u|^{-1}$ as $u\to-\infty$ on a null hypersurface $\mathcal C_{\mathrm{in}}$ and the no incoming radiation condition. (See Figure~\ref{fig:III:0}.)
  Then, schematically, $r\phi_\ell|_{\mathcal I^+}\sim |u|^{-\min(\ell+1,2)}$ as $u\to-\infty$, and the asymptotic expansion of $\pv(r\phi_\ell)$ along outgoing null hypersurfaces of constant $u$ near spacelike infinity $i^0$ reads:
 \begin{equation}\label{eq:intro:scheme2}
\pv(r\phi_\ell)=\frac{f_0(u)}{r^2}+C\frac{\log r}{r^{3}}+\dots,
\end{equation}
unless $\ell=1$, in which case we instead have that
\begin{equation}\label{eq:intro:scheme3}
\pv(r\phi_\ell)=\frac{f_0(u)}{r^2}+\frac{f_1(u)}{r^3}+C\frac{\log r}{r^{4}}+\dots.
\end{equation}
In both cases, $C$ is a generically non-vanishing constant.
}
\end{minipage}}
\par\smallskip
By incorporating an $r^\ell$-weight into the boundary data assumption (namely  $r^{\ell+1}\phi_\ell|_{\Gamma}\sim |t|^{-1}$), we phrased statement \textit{\textbf{1)}} in such a way as to be independent of the behaviour of the area radius $r$ on $\Gamma$:
Independently of whether $r$ is constant along $\Gamma$ or divergent (e.g.\ $r|_{\Gamma}\sim|t|$),    
the $|t|^{-1}$-decay of $r^{\ell+1}\phi$ on $\Gamma$ translates into $|u|^{-\ell-1}$ decay of $r\phi_\ell$ near $\mathcal I^-$, causing the logarithmic term in \eqref{eq:intro:scheme1} to appear $\ell$ orders later than in \eqref{eq:intro:scheme2}.



The difference between \eqref{eq:intro:scheme2} and \eqref{eq:intro:scheme3}, on the other hand, is a manifestation of certain cancellations that happen if $r\phi_\ell\sim |u|^{-\ell}$ on $\mathcal C_{\mathrm{in}}$. 
Similar cancellations are responsible for $r\phi_\ell$ decaying faster on $\mathcal I^+$ than on $\mathcal C_{\mathrm{in}}$ in case \textbf{\textit{2)}}.
These cancellations, together with the precise and more general versions of the above statements, will be discussed in detail in section~\ref{sec:intro:summary} below, see also Remark~\ref{rem:intro:cancel}. 

Let us finally remark that, even though higher $\ell$-modes thus decay slower than the spherically symmetric mode near spacelike infinity, we still expect the leading-order asymptotics near \textit{future timelike infinity} $i^+$ to be dominated by the spherically symmetric mode in the two data setups described above, see also~\cite{II} and~\cite{AAG21}. 
However, in the case of smooth compactly supported scattering data on $\mathcal I^-$ and the past event horizon $\mathcal H^-$, it turns out that all $\ell$-modes can be expected to have the same decay along $\mathcal I^+$ as $i^+$ is approached.
 We will discuss this in detail in section~\ref{sec:intro:NP}, see already Figures~\ref{fig:III:1}--\ref{fig:III:3}.

\subsection{Summary of the main results}\label{sec:intro:summary}
We now give a summary of the main theorems obtained in this paper. They are all stated with respect to Eddington--Finkelstein double null coordinates $(u,v)$ ($\pv r=1-\frac{2M}{r}=-\pu r$).
Let's first focus on solutions to \eqref{waveequation} supported on a single $\ell=1$-frequency.
\subsubsection{The case \texorpdfstring{$\ell=1$}{L=1}}\label{sec:1.2.1}
Let $\Gamma\subset \mathcal{M}_M$ be a spherically symmetric, past-complete timelike hypersurface of constant area radius function $r=R>2M$.\todo{Necessary to have this restriction?}\footnote{In fact, the theorem below also applies to spherically symmetric hypersurfaces $\Gamma$ on which $r$ is allowed to vary and, in particular, tend to infinity. We will show this in the main body of the paper.}
Let $\ell=1$ and $|m|\leq 1$, and prescribe on $\Gamma$ smooth boundary data for $\phi_{\ell=1}=\phi_1 \cdot Y_{1m}$ that satisfy, as $u\to-\infty$,\todo{Explain O}
\begin{equation}\label{eq:intro:ass1}
\left|r^2\phi_1|_{\Gamma}-\frac{C_\Gamma}{|u|}\right|=\mathcal{O}_5(|u|^{-1-\epsilon})
\end{equation}
for some constant $C_\Gamma$ and for some $\epsilon\in(0,1)$.
Moreover, prescribe in  a limiting sense that
\begin{equation}
\label{eq:intro:nir}
\lim_{u\to-\infty}\pv^n(r\phi_1)(u,v)=0, \quad n=0,\dots, 5
\end{equation}
for all $v\in \mathbb R$. 
We interpret this as the condition of no incoming radiation from $\mathcal I^-$.
We then prove the following theorem, in its rough form (see Theorem~\ref{thm:tl} for the precise version):
	\begin{thm}\label{thm:intro:tl}
	Given smooth boundary data satisfying \eqref{eq:intro:ass1}, there exists a unique smooth (finite-energy) solution to \eqref{waveequation} (restricted to the $(1,m)$-angular frequency) in the domain of dependence of $\Gamma\cup \mathcal{I}^-$ that restricts correctly to these data and satisfies \eqref{eq:intro:nir}. Moreover, this solution satisfies along any spherically symmetric ingoing null hypersurface:
	\begin{align}
	\lim_{u\to-\infty}r^2\pu(r\phi_1)(u,v)&=0,\\
	\lim_{u\to-\infty}r^2\pu(r^2\pu(r\phi_1))(u,v)&\equiv \ILpn{1}[\phi],
	\end{align}
	where $\ILpn{1}[\phi]$ is a constant which is non-vanishing as long as $C_\Gamma$ is non-vanishing and $R/2M$ is sufficiently large, and we further have that
	\begin{align}
	\left|r^2\pu(r^2\pu(r\phi_1))(u,v) - \ILpn{1}[\phi_1]\right|=\mathcal{O}(\max(r^{-1}, |u|^{-\epsilon})).
	\end{align}
	In particular, $r\phi_1$ decays like $u^{-2}$ towards $\mathcal I^-$.
	\end{thm}

The next theorem translates these results into logarithmic asymptotics along outgoing null hypersurfaces in a neighbourhood of spacelike infinity.
Let $\mathcal{C}_{\mathrm{in}}$ be a spherically symmetric, past-complete ingoing null hypersurface (e.g.\ $v=1$). 
Prescribe on $\mathcal{C}_{\mathrm{in}}$ smooth data for $\phi_{\ell=1}=\phi_1 \cdot Y_{1m}$ that satisfy
\begin{align}\label{eq:intro:ass2}
\lim_{u\to-\infty}r^2\pu(r\phi_1)&=\cc,\\
\left|r^2\pu(r^2\pu(r\phi_1))-\ccc\right|&=\mathcal{O}(|u|^{-\epsilon}),\label{eq:intro:ass3}
\end{align}
for some constants $\cc$, $\ccc$ and for some $\epsilon\in(0,1)$. 
Moreover, prescribe equation \eqref{eq:intro:nir} to hold in a limiting sense to the future of $\mathcal C_{\mathrm{in}}$ for $n=0,1,2$. \todo{Requires less regularity here...} Then we have (see Theorem~\ref{thm:nl} for the precise version):
	\begin{thm}\label{thm:intro:nl}
	Given smooth data satisfying \eqref{eq:intro:ass2} and \eqref{eq:intro:ass3}, there exists a unique smooth solution to \eqref{waveequation} (restricted to the $(1,m)$-angular frequency) in the domain of dependence of $\mathcal{C}_{\mathrm{in}}\cup \mathcal{I}^-$ that restricts correctly to these data and satisfies \eqref{eq:intro:nir}.
	 Moreover, this solution satisfies, for sufficiently large negative values of $u$, the following asymptotics as $\mathcal{I}^+$ is approached along any outgoing spherically symmetric null hypersurface:
	\begin{align}\label{eq:1.14}
	\begin{split}
	r^2\pv\pho(u,v)=&-\cc-2\int_{-\infty}^uF(u')\dd u'-\frac{2M\cc-2M\int_{-\infty}^uF(u')\dd u'}{r}\\
	&-2M(\ccc-2M\cc)\frac{\log r-\log|u|}{r^2}+\mathcal{O}(r^{-2}),
	\end{split}
	\end{align}
	where $F(u)$ is given by the limit of the radiation field $r\phi_1$ on $\mathcal{I}^+$,
	\begin{equation}\label{eq:1.15}
	F(u):=\lim_{v\to\infty}r\phi_1(u,v)=\frac{\ccc-2M\cc}{6|u|^2}+\mathcal{O}(|u|^{-2-\epsilon}).
	\end{equation}
	The asymptotics of $r\phi_1$ near $\mathcal I^+$ can be obtained by integrating $\pv(r\phi_1)$ from $\mathcal I^+$ and combining \eqref{eq:1.14} and \eqref{eq:1.15}.
	In particular, if $M(\ccc-2M\cc)\neq 0$, then peeling fails at future null infinity.
	\end{thm}
Theorem~\ref{thm:intro:nl} applies to the solution of Theorem~\ref{thm:intro:tl}, with $\cc=0$ and $\ccc=\ILpn1[\phi]$.
\begin{rem}
Let us already make the following observation: We recall from section~\ref{sec:subsec:intro:sphsym} that, in the spherically symmetric case ($\ell=0$) studied in~\cite{I}, the initial $u$-decay of $r\phi_0$ was transported all the way to $\mathcal{I}^+$, that is, we had that $\lim_{v\to\infty}r\phi_0(u,v)\sim|u|^{-1}$.
 This fact was closely related to the approximate conservation law satisfied by $\pu(r\phi_0)$.\footnote{Recall that $\pv\pu(r\phi_0)= -2M(1-\frac{2M}{r})\cdot\frac{r\phi_0}{r^3}$.} For $\ell=1$, we see that this is no longer the case: The initial $|u|^{-1}$-decay of $r\phi_1$ translates into $|u|^{-2}$-decay on $\mathcal{I}^+$. 
  This improvement in the $u$-decay on $\mathcal I^+$ can be traced back to certain cancellations that happen if the $|u|$-decay of the data comes with a specific power: 
  In fact, notice from \eqref{eq:1.15} that if $C_{\mathrm{in}}^{(1)}=0$, the $u$-decay of $r\phi_1$ on $\mathcal I^+$ sees \textit{no improvement} over its initial decay.
We will understand these cancellations in more generality in the theorems below, see already equation \eqref{eq:intro:thma1} of Theorem~\ref{thm:intro:moreg} and the Remark~\ref{rem:intro:cancel}.
See also \S\ref{sec:4.4.3} for a schematic explanation of these cancellations.
\end{rem}

\subsubsection{The case of general \texorpdfstring{$\ell\geq 0$}{L>0}}
Let $\ell=L\in\mathbb N_0$ and $|m|\leq L$, let $\Gamma$ be as in \S\ref{sec:1.2.1}, and prescribe on $\Gamma$ smooth boundary data for $\phi_{\ell=L}=\phi_L \cdot Y_{Lm}$ that satisfy, as $u$ tends to $-\infty$,
\begin{equation}\label{eq:intro:ass4}
\left|r^{L+1}\phi_L|_{\Gamma}-\frac{C_\Gamma}{|u|}\right|=\mathcal{O}_{L+4}(|u|^{-1-\epsilon})
\end{equation}
for some constant $C_\Gamma$ and some $\epsilon\in(0,1)$, and prescribe again, in a limiting sense, that for all $v\in\mathbb R$: 
\begin{equation}\label{eq:intro:nirg}
\lim_{u\to-\infty}\pv^n(r\phi_L(u,v))=0,\quad n=0,\dots,L+4.
\end{equation}
Then we have (see Theorem~\ref{thm:gtl} for the precise version):
	\begin{thm}\label{thm:intro:gtl}
	Given smooth boundary data satisfying \eqref{eq:intro:ass4}, there exists a unique smooth solution to \eqref{waveequation} (restricted to the $(L,m)$-angular frequency) in the domain of dependence of $\Gamma\cup \mathcal{I}^-$ that restricts correctly to these data and satisfies \eqref{eq:intro:nirg}.
	 Moreover, this solution satisfies along any spherically symmetric ingoing null hypersurface:
	\begin{align}
	\lim_{u\to-\infty}(r^2\pu)^{L-j}(r\phi_L)(u,v)&=0,&&j=0,\dots,L,\\
	\lim_{u\to-\infty}(r^2\pu)^{L+1}(r\phi_L)(u,v)&\equiv \ILpn{L}[\phi],&&\label{eq:intro:thmtimelike2}
	\end{align}
	where $\ILpn{L}[\phi]$ is a constant which is non-vanishing as long as $C_\Gamma$ is non-vanishing and $R/2M$ is sufficiently large, and we further have that
	\begin{align}
	\left|(r^2\pu)^{L+1}(r\phi_1)(u,v) - \ILpn{L}[\phi]\right|=\mathcal{O}(\max(r^{-1},|u|^{-\epsilon})).\label{eq:intro:thmtimelike3}
	\end{align}
	In particular, $r\phi_\ell$ decays like $|u|^{-\ell-1}$ towards $\mathcal I^-$.
	\end{thm}
\begin{rem}
Theorem~\ref{thm:intro:gtl} also applies to boundary data on more general spherically symmetric timelike hypersurfaces on which $r$ is allowed to tend to infinity. See also Theorems~\ref{thm:tlf},~\ref{thm:tlf2}. 

Moreover, the proof can also be applied to any inverse polynomial rate for the $|u|$-decay of the boundary data.
In fact, if $r|_{\Gamma}\to\infty$, one can more generally apply it to growing polynomial rates, $r^{L+1}\phi_L|_{\Gamma}\sim |u|^{-p}$ for some $p<0$, so long as the quantity $r\phi_L|_{\Gamma}$ itself is decaying.
 This leads to some obvious changes in equations \eqref{eq:intro:thmtimelike2}, \eqref{eq:intro:thmtimelike3}. (Schematically, if $r^{L+1}\phi_L|_{\Gamma}\sim |u|^{-p}$ along $\Gamma$, then $r\phi_L\sim |u|^{-L-p}$ along hypersurfaces of constant $v$.)
\end{rem}
\begin{rem}
Notice that the regularity required for the boundary data (eq.\ \eqref{eq:intro:ass4}), when restricted to $L=0$, is higher than that of~\cite{I}. This is because, for general $L\geq 0$, we need to work with certain energy estimates in order to obtain the sharp decay for transversal derivatives on $\Gamma$, which is not necessary for the $\ell=0$-mode.
\end{rem}
As before, the results of Theorem~\ref{thm:intro:gtl} translate into logarithmic asymptotics near spacelike infinity:
 Prescribe on $\mathcal{C}_{\mathrm{in}}$ smooth data for $\phi_{\ell=L}=\phi_L \cdot Y_{Lm}$ that satisfy
\begin{align}\label{eq:intro:ass5}
\lim_{u\to-\infty}(r^2\pu)^{L-j}(r\phi_L)(u,v)&=0,&&j=0,\dots,L,\\
\left|(r^2\pu)^{L+1}(r\phi_1)(u,v) - C^{(L,0)}_{\mathrm{in}}\right|&=\mathcal{O}(|u|^{-\epsilon})&&\label{eq:intro:ass6}
\end{align}
for some constant $C^{(L,0)}_{\mathrm{in}}$ and some $\epsilon\in(0,1)$,
and further prescribe equation \eqref{eq:intro:nirg} to hold in the future of $\mathcal C_{\mathrm{in}}$ for $n=0,\dots,L+1$. We prove the following theorem in its rough form (see Theorem~\ref{thm:gnl} for the precise version):
	\begin{thm}\label{thm:intro:gnl}
	Given smooth data satisfying \eqref{eq:intro:ass5} and \eqref{eq:intro:ass6}, there exists a unique smooth solution to \eqref{waveequation} (restricted to the $(L,m)$-angular frequency) in the domain of dependence of $\mathcal C_{\mathrm{in}}\cup \mathcal{I}^-$ that restricts correctly to these data and satisfies \eqref{eq:intro:nirg}. 
 Moreover, this solution satisfies, for sufficiently large negative values of $u$, the following asymptotics as $\mathcal{I}^+$ is approached along any outgoing spherically symmetric null hypersurface:
	\begin{align}
	\begin{split}
	r^2\pv(r\phi_L)(u,v)=\sum_{i=0}^L \frac{f_i^{(L)}(u)}{r^{i}}+\frac{ \ILlog{L}[\phi](\log r-\log|u|)}{r^{L+1}}+\mathcal{O}(r^{-L-1})
	\end{split}
	\end{align}
	where the $f_i^{(L)}$ are smooth functions of $u$ which satisfy $f_i(u)=\frac{\beta_i^{(L)}}{|u|^{L-i}}+\mathcal{O}(|u|^{-L+i-\epsilon})$ for some explicit numerical constants $\beta_i^{(L)}$, and $\ILlog{L}[\phi]$ is an explicit constant which can be expressed as a non-vanishing numerical multiple of $M$ and $C^{(L,0)}_{\mathrm{in}}$. 
	In addition, we have that
	\begin{align}
\lim_{v\to\infty} r\phi_L(u,v)=\frac{L!C^{(L,0)}_{\mathrm{in}}}{(2L+1)!|u|^{L+1}}+\mathcal{O}(|u|^{-L-1-\epsilon}).
\end{align}
The asymptotics or $r\phi_L$ near $\mathcal I^+$ can again be obtained by integrating $\pv(r\phi_L)$ from $\mathcal I^+$ and combining the above two estimates.
	\end{thm}

Now, while Theorem~\ref{thm:intro:gtl} generalises Theorem~\ref{thm:intro:tl} in every sense, Theorem~\ref{thm:intro:gnl} does not fully generalise Theorem~\ref{thm:intro:nl} since it excludes initial data that satisfy
\begin{equation}
\lim_{u\to-\infty}[r^2\pu]^{L-j}(r\phi_L)=C_{\mathrm{in}}^{(L,j+1)}
\end{equation}
for $j=0,\dots, L$ and \emph{non-vanishing} constants $C_{\mathrm{in}}^{(L,j)}$. If only  $C_{\mathrm{in}}^{(L,1)}$ is non-vanishing, then, in fact, the above theorem remains valid, albeit with some modifications to the $f_i(u)$ and to the constant $\ILlog{L}[\phi]$. 
More generally, however, we have the following:

Instead of \eqref{eq:intro:ass5}, \eqref{eq:intro:ass6}, prescribe on $\mathcal C_{\mathrm{in}}$ that
\begin{align}\label{eq:intro:ass7}
\left|r\phi_L(u,1)-\frac{C_{\mathrm{in}}}{r^p}\right|=\mathcal O_1( r^{-p-\epsilon})
\end{align}
for some $\epsilon\in(0,1]$, a constant $C_{\mathrm{in}}\neq 0$, and for some $p\in \mathbb N_0$ ($p=0$ is permitted). 
Moreover, assume the no incoming radiation condition \eqref{eq:intro:nirg} to hold for $n=1,\dots,L+1$. 
Then we have (see Theorem~\ref{thm:moreg} for the precise version):
	\begin{thm}\label{thm:intro:moreg}
	Given smooth data satisfying \eqref{eq:intro:ass7}, there exists a unique smooth solution to \eqref{waveequation} (restricted to the $(L,m)$-angular frequency) in the domain of dependence of $\mathcal C_{\mathrm{in}}\cup \mathcal{I}^-$ that restricts correctly to these data and satisfies \eqref{eq:intro:nirg}. 
	Define $r_0:=|u|-2M\log|u|$.
	Then the limit of the radiation field satisfies
	\begin{equation}\label{eq:intro:thma1}
	\lim_{v\to\infty}r\phi_L(u,v)=F(u)=\begin{cases}
													\mathcal O(r_0^{-p-\epsilon}),&\text{if }p\leq L \text{ and } p\neq 0,\\
													C(L,p)\cdot C_{\mathrm{in}} r_0^{-p}+\mathcal O(r_0^{-p-\epsilon}), &\text{if }p>L \text{ or } p=0,
										\end{cases}
	\end{equation}
	for some smooth function $F(u)$ and some non-vanishing numerical constant $C(L,p)$.
	
	Moreover, \underline{if $p<L$}, this solution satisfies, for sufficiently large negative values of $u$, the following asymptotics as $\mathcal{I}^+$ is approached along any outgoing spherically symmetric null hypersurface:
	\begin{equation}\label{eq:intro:thmb1}
	r^2\pv(r\phi_L)(u,v)=\sum_{i=0}^{p-1}\frac{f^{(L,p)}_i(u)}{r^i}+\frac{I_{\ell=L}^{\mathrm{future},r^{2+p-L}}[\phi](\log r-\log|u|)}{r^p}+\mathcal{O}\left(\frac{|u|}{r^{p}}\right),
	\end{equation}
	where the $f^{(L,p)}_i$ are smooth functions which satisfy $f^{(L,p)}_i=\mathcal O(r_0^{-p+i+1-\epsilon})$ if $i<p-1$, and $f^{(L,p)}_{i}=\beta_i^{(L,p)}+\mathcal O(r_0^{-\epsilon})$ for some constant $\beta_i^{(L,p)}$  if $i=p-1$. $I_{\ell=L}^{\mathrm{future},r^{2+p-L}}[\phi]$ is a non-vanishing constant which depends on $p,L,C_{\mathrm{in}}$ and $M$.
	
	On the other hand, \underline{if $p\geq L$}, then 
	\begin{equation}\label{eq:mgintro:thmb2}
	r^2\pv(r\phi_L)(u,v)=\sum_{i=0}^{\max(L,p-1)}\frac{f^{(L,p)}_i(u)}{r^i}+\mathcal O\left(\frac{\log r}{r^{\max(L+1,p)}}\right),
	\end{equation}
	where the $f^{(L,p)}_i$ are smooth functions which satisfy $f^{(L,p)}_i=\mathcal O(r_0^{-p+i+1-\epsilon})$ if $p=L$ and $i<L-1$, and which satisfy $f^{(L,p)}_i=\beta_i^{(L,p)}r_0^{-p+i+1}+\mathcal O(r_0^{-p+i-\epsilon})$ for some constants $\beta_i^{(L,p)}$ otherwise (i.e.\ if $p=L=i$, $p=L=i+1,$ or if $p>L$).
	
	The asymptotics or $r\phi_L$ near $\mathcal I^+$ can again be obtained by integrating $\pv(r\phi_L)$ from $\mathcal I^+$ and combining the above estimates.
	\end{thm}
Some remarks are in order.
\begin{rem}\label{rem:intro:cancel}
Notice the different behaviour in the cases $p<L$, $p=L$ and $p>L$ in Theorem~\ref{thm:intro:moreg}. We want to direct the reader's attention to the following points:
\begin{itemize}
	\item Equation \eqref{eq:intro:thma1} shows that if $0\neq p\leq L$, then there is a \textit{cancellation}  and $\lim_{v\to\infty}(r\phi_L)(u,v)$ decays faster in $u$ than $r\phi_L(u,1)$.
	See \S\ref{sec:4.4.3} for a schematic explanation of these cancellations.
	 Such cancellations \textit{do not} happen if $p=0$ of $p>L$. 
	Moreover, they can be viewed as Minkowskian behaviour, i.e., they can already be seen if $M=0$.
	 In fact, in the course of the proof of Theorem~\ref{thm:intro:moreg}, we will derive simple and effective expressions for solutions of $\Box_g\phi_L=0$ on Minkowski arising from the no incoming radiation condition and initial data $r\phi_L(u,1)=C/r^p$ (see Proposition~\ref{prop:mg:prop4}).
	\item In view of \eqref{eq:mgintro:thmb2}, we see that "the first logarithmic term" in the expansion of $r^2\pv(r\phi_L)$ appears at order $r^{-p-1}\log r$ unless $p=L$, in which case it appears one order later. In particular, it never appears at order $r^{-L-1}\log r$.
\end{itemize}
\end{rem}
\begin{rem}
 The proof of Theorem~\ref{thm:intro:moreg} can be generalised to positive non-integer $p$ in \eqref{eq:intro:ass7} (and even to certain negative $p$). However, if $p\notin\{1,\dots,L\}$, we expect no cancellations of the type above to occur. 
On the other hand, if we assume, for instance,  that $r\phi_L(u,1)\sim r^{-p}\log r $ initially, then the same cancellations occur in the range $p\in\{1,\dots L\}$, and one will obtain that $r\phi_L|_{\mathcal I^+}\sim |u|^{-p}$ if $p\in\{1,\dots, L\}$ and $r\phi_L|_{\mathcal I^+}\sim |u|^{-p}\log|u|$ otherwise.
 This observation will be of relevance in future work.
 
 
 \end{rem}
\begin{rem}\label{rem1.6}
All of the above theorems make crucial use of certain approximate conservation laws. These are generalisations of the Minkowskian identities
\begin{align*}
\pu\left(r^{-2\ell-2}(r^2\pv)^{\ell+1}(r\phi_\ell)\right)=0,&&
\pv\left(r^{-2\ell-2}(r^2\pu)^{\ell+1}(r\phi_\ell)\right)=0,
\end{align*}
and have been used in a very similar context in the recent~\cite{AAG21}, see also \cite{MZ21}. 
See already section~\ref{sec:thecommutedequations} and section~\ref{sec:generalNP} for a discussion and derivation of these in the cases $\ell\leq 1$, $\ell \geq 0$, respectively. 
The reason why we stated Theorems~\ref{thm:intro:gnl} and~\ref{thm:intro:moreg} separately is that the former can be proved in a rather simple way using the second conservation law, i.e.\ by propagating the initial decay for $(r^2\pu)^{\ell+1}(r\phi_\ell)$ in $v$, whereas, in order to prove  Theorem~\ref{thm:intro:moreg}, we will need to use the conservation law in the $u$-direction.
\end{rem}
\begin{rem}\label{rem1.7}
The constants $I_{\ell=L}^{\mathrm{future},f}[\phi]$ appearing in the above theorems are modified Newman--Penrose constants. These are closely related to the approximate conservation laws mentioned before. We will discuss this further in the next section.
\end{rem}
\begin{rem}
One can generalise all of the above theorems to hold on more general spherically symmetric spacetimes such as the Reissner--Nordstr\"om spacetimes in the full physical range of charge parameters $|e|\leq M$. 
In the extremal case $|e|=M$, one can moreover apply the well-known conformal "mirror" isometry to obtain results on the asymptotics near the future event horizon $\mathcal H^+$, see section~2.2.2 of~\cite{I}.
\end{rem}

\subsection{Future applications: Late-time asymptotics and the role of the modified Newman--Penrose constants} \label{sec:intro:NP}
The approximate conservation laws mentioned in Remarks~\ref{rem1.6},~\ref{rem1.7} are closely related to the $\ell$-th order Newman--Penrose constants  $I_\ell[\phi]$ defined on future and past null infinity, respectively
 (see also the original~\cite{NPconstants65,NPconstants68}, and, more tailored to our context,~\cite{AAG18b,AAG21} and section~\ref{sec:generalNP} of the present paper). 
In fact, these $\ell$-th order Newman--Penrose constants play an important role in the study of both  \textit{early-time asymptotics (near $i^0$)} and \textit{late-time asymptotics (near $i^+$)} of fixed-$\ell$ solutions to the wave equation on Schwarzschild. 

While the question of early-time asymptotics has not been investigated much elsewhere, the study of late-time asymptotics has been an active field for decades. 
The most prominent result in this line of research is the so-called \emph{Price's law}~\cite{Price72,GPP94}, see also~\cite{Leaver86}. 
Price's law states that smooth, compactly supported data on a Cauchy hypersurface (i.e.\ data with trivial early-time asymptotics) for fixed angular frequency solutions $\phi_{\ell=L}=Y_{Lm}\phi_{Lm}$ to the wave equation \eqref{waveequation}  generically lead to the following asymptotics near future timelike infinity $i^+$ (we suppress the $m$-index in the following):
\begin{align}\label{eq:Price}
r\phi_L|_{\mathcal I^+}\sim u^{-2-L}, &&\phi_L|_{r=\text{constant}}\sim \tau^{-2L-3}, &&\phi_L|_{\mathcal H^+}\sim v^{-2L-3}
\end{align}
along future null infinity, hypersurfaces of constant $r$, and the event horizon $\mathcal H^+$, respectively. 
This statement has been satisfactorily proved in the recent works~\cite{AAG18a,AAG18b, AAG21}, see also~\cite{Hintz22} and~\cite{MZ21}. (For earlier rigorous works on \textit{pointwise upper bounds} (not asymptotics), see \cite{DSS11,DSS12} as well as  \cite{MTT12}.)
We also refer the reader to these papers for more general background and motivation for the study of late-time asymptotics. 

The question of late-time asymptotics for \textit{compactly supported Cauchy data} has thus been completely understood. 
Similar results have been obtained for non-compactly supported data, but in that case, it has  typically been assumed that the data are conformally smooth.
However, if one's motivation for studying late-time asymptotics comes from gravitational wave astronomy (i.e.\ the hope that some devices will eventually be able to measure these asymptotics), then the assumption of smooth compactly supported (or conformally smooth) data on a Cauchy hypersurface becomes questionable -- as long as one accepts the general framework of an \textit{isolated system}. 
For, if one assumes that the gravitational waves emitting system under consideration has existed for all times, then it will certainly have radiated for all times: 
Thus, a spacetime describing this system cannot be expected to contain Cauchy hypersurfaces with compact radiation content. 
On the other hand, the data considered in~\cite{I} and the present paper have a clear physical motivation\footnote{In addition to the remarks at the beginning of the paper, see also sections~1 and~2.1 of~\cite{I} for a more detailed discussion of the physical motivation.} and, thus,  seem like a more reasonable starting point for the question of physically relevant late-time asymptotics.

Motivated by this, we shall now discuss consequences that our results from section~\ref{sec:intro:summary} have on late-time asymptotics.
It turns out that one can gain a simple, intuitive understanding of these in terms of the aforementioned Newman--Penrose constants. 
\subsubsection{The timelike case: A logarithmically modified Price's law for all \texorpdfstring{$\ell$}{L}}\label{intro:subsecNP1}
Let's assume that we have a spherically symmetric timelike hypersurface $\Gamma$ that has \textit{constant} area radius near $i^-$ and terminates at $\mathcal{H}^+$. 
(Note that, if we chose $\Gamma$ to terminate at $i^+$, then we would have to essentially \textit{prescribe} the late-time asymptotics as boundary data on $\Gamma$. 
On the other hand, if we choose $\Gamma$ to terminate at $\mathcal H^+$, then it will turn out that the leading-order late-time asymptotics are completely determined by the data's behaviour near $i^-$. In particular, they do not depend on the extension of the data towards $\mathcal H^+$.)
Consider first the spherically symmetric mode, and prescribe smooth data for it which, near past timelike infinity $i^-$, behave like $r\phi_0=C|u|^{-1}+\O(|u|^{-1-\epsilon})$, and which smoothly extend to the future event horizon $\mathcal{H}^+$; and impose the no incoming radiation condition on $\mathcal I^-$. 
Then the results of~\cite{I} show\footnote{In fact, we only showed in~\cite{I} that $\lim_{u\to-\infty}|u| r\phi_0$ is finite. However, it follows directly from Lemma~5.1 therein that $\lim_{u\to-\infty}|u| r\phi_0=\lim_{u\to-\infty}|u|^2 T(r\phi_0)=\lim_{u\to-\infty}r^2\pu(r\phi_0)$, since $\pv(r\phi_0)\lesssim r^{-2}|u|^{-1}$. Alternatively, one can also refer to Thm.~\ref{thm:intro:gtl} of the present paper.} that the past Newman--Penrose constant exists and is conserved along $\mathcal{I}^-$:
\begin{equation}
\ILpn0[\phi]:=\lim_{u\to -\infty} r^2\pu(r\phi_0)\neq 0.
\end{equation}
Moreover, we showed that the finiteness of the past N--P constant, together with the no incoming radiation condition, implies that, even though the future Newman--Penrose constant vanishes ($\lim_{v\to\infty} r^2 \pv(r\phi_0)=0$), a \textit{logarithmically modified} future Newman--Penrose constant exists and is conserved along $\mathcal{I}^+$:
\begin{equation}
I_{\ell=0}^{ \mathrm{future},\frac{\log r}{r^3}}[\phi]:=\lim_{v\to\infty} \frac{r^3}{\log r} \pv(r\phi_0)=-2MI_{\ell=0}^{\mathrm{past}}[\phi]\neq 0.
\end{equation}
In~\cite{II}, we then applied slight adaptations of the methods of~\cite{AAG18b} to show that this logarithmically modified Newman--Penrose constant completely determines the leading-order late-time asymptotics near $i^+$:
\begin{align}\label{eq:intro:asy1}
r\phi_0(u,\infty)=\frac12\ILlog0[\phi]\frac{\log u}{u^2}+\mathcal{O}(u^{-2}),\\
\phi_0(u,v_R(u))=\frac12\ILlog0[\phi]\frac{\log\tau}{\tau^3}+\mathcal{O}(\tau^{-3}),\\
\phi_0(\infty,v)=\frac12\ILlog0[\phi]\frac{\log v}{v^3}+\mathcal{O}(v^{-3}),\label{eq:intro:asy3}
\end{align}
along $\mathcal{I}^+$, hypersurfaces of constant $R$, and $\mathcal{H}^+$, respectively.\footnote{Notice that the $(u,v)$-gauge used in this paper is related to that of~\cite{II} by $(u,v)=(u'/2,v'/2)$. }  In particular, the leading-order late-time behaviour is independent of the extension of the data towards $\mathcal H^+$ and only depends on the behaviour of the data near $i^-$. We called this \textit{a logarithmically modified Price's law} for the $\ell=0$-mode.

Consider now the $\ell=1$-case. If we assume data as in Theorem~\ref{thm:intro:tl} and smoothly extend them  to $\mathcal H^+$, then we obtain that the past N--P constant of order $\ell=1$ exists and is conserved along $\mathcal{I}^-$:
\begin{equation}\label{eq:intro:NP1}
\ILp1[\phi]:=\lim_{u\to-\infty}r^2\pu(r^2\pu(r\phi_1)+Mr\phi_1)=\lim_{u\to-\infty}r^2\pu(r^2\pu(r\phi_1))\neq 0.
\end{equation}
It then follows from Theorem~\ref{thm:intro:nl} that the decay encoded in \eqref{eq:intro:NP1} (namely, $r\phi_1\sim u^{-2}$), along with the no incoming radiation condition, implies that the future N--P constant vanishes, but that a \textit{logarithmically modified future Newman--Penrose constant of order $\ell=1$} exists and is conserved along $\mathcal{I}^+$ (see also Theorem~\ref{thm:NP} in \S\ref{sec:nl:comments}):
\begin{equation}\label{eq:intro:NP2}
I_{\ell=1}^{ \mathrm{future},\frac{\log r}{r^3}}[\phi]:=\lim_{v\to\infty} \frac{r^3}{\log r} \pv(r^2\pv(r\phi_1)-Mr\phi_1)=4MI_{\ell=1}^{\mathrm{past}}[\phi].
\end{equation}
One should then be able to combine the results above with those of Angelopoulos--Aretakis--Gajic~\cite{AAG21}, with adaptations exactly as in~\cite{II} (which combined the results of~\cite{I} and~\cite{AAG18a,AAG18b}), in order to obtain near $i^+$:
\begin{align}\label{eq:intro:asy4}
r\phi_1(u,\infty)=C_1\frac{\log u}{u^3}+\mathcal{O}(u^{-3}),\\
\phi_1(u,v_R(u))=C_2\frac{\log\tau}{\tau^5}+\mathcal{O}(\tau^{-5}),\\
\phi_1(\infty,v)=C_2\frac{\log v}{v^5}+\mathcal{O}(v^{-5}),\label{eq:intro:asy6}
\end{align}\todo{Could figure out what these constants are?}where $C_1, C_2$ are given by numerical multiples of $I_{\ell=1}^{\log, \mathrm{future}}[\phi]$. In particular, these constants $C_1,C_2$ should be independent of the extension of the data towards $\mathcal H^+$.  
Thus, we would obtain a logarithmically modified Price's law for $\ell=1$ (cf.\ \S\ref{sec:nl:comments:logpricelaw}).

\begin{figure}[htbp]
\floatbox[{\capbeside\thisfloatsetup{capbesideposition={right,top},capbesidewidth=4.4cm}}]{figure}[\FBwidth]
{\caption{Schematic depiction of the situation of \S \ref{intro:subsecNP1}: Given smooth data for $r\phi_\ell$ on $\Gamma$ which decay like $1/t$ near $i^-$, the solution decays like $u^{-\ell-1}$ near $\mathcal I^-$ by Thm.~\ref{thm:intro:gtl} and has finite logarithmically modified N--P constant on $\mathcal I^+$ by Thm.~\ref{thm:intro:gnl}. The depicted late-time behaviour near $i^+$ should follow from the methods of~\cite{AAG21} and should be independent of the data's extension towards $\mathcal H^+$.}\label{fig:III:1}}
{\includegraphics[width=200pt]{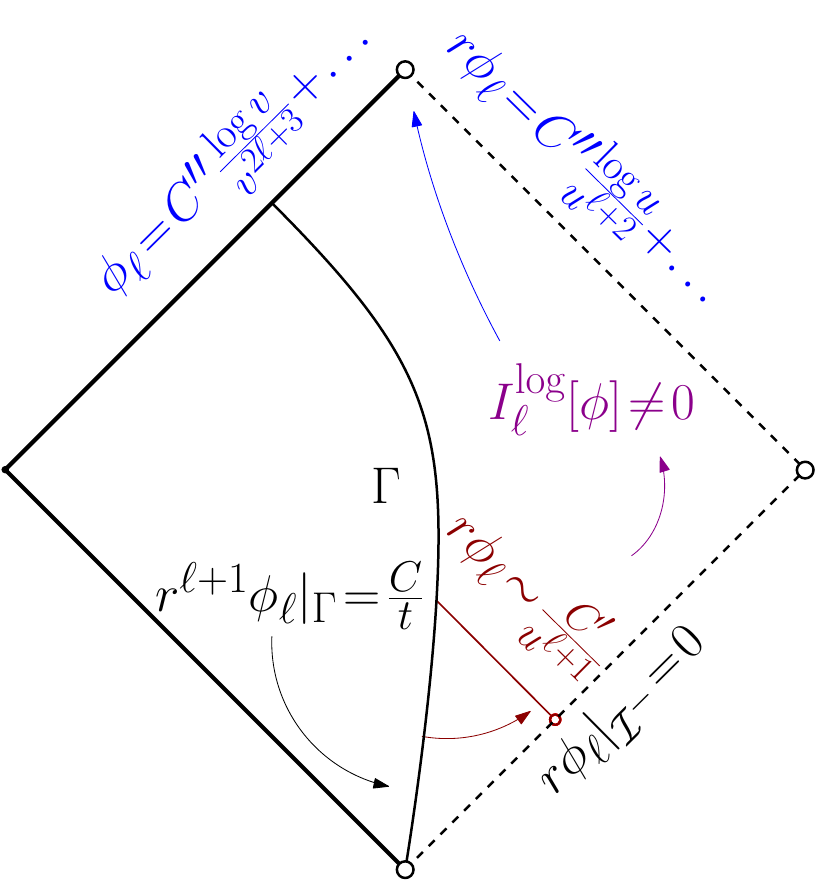}}
\end{figure}

In fact, we expect the same structure to hold in the case of general $\ell=L$. The data of Theorem~\ref{thm:intro:gtl} lead to solutions which have finite $L$-th order past N--P constant and, by Theorem~\ref{thm:intro:gnl} (see also Theorem~\ref{thm:gnl}), have a finite \textit{logarithmically modified} $L$-th order future N--P constant $I_{\ell=L}^{ \mathrm{future},\frac{\log r}{r^3}}[\phi]$, see  \S \ref{sec:generalNP} for the definition of these.
 In other words, Theorems~\ref{thm:intro:gtl} and~\ref{thm:intro:gnl} prove the precise analogues of \eqref{eq:intro:NP1} and \eqref{eq:intro:NP2} for general $\ell=L$. 
In view of the remarks above, one should then be able to recover a \textit{logarithmically modified Price's law {for each $\ell$}} from this. See Figure~\ref{fig:III:1}. 

What would be more difficult, however, is to show such a statement for \textit{fixed, finite regularity} of $\phi_L$ instead of assuming smoothness or regularity that is dependent on $L$. We therefore make the following conjecture:
\begin{conj}\label{conj1}
Prescribe data for $\phi$ on $\Gamma$ that have sufficient but fixed, finite regularity and which satisfy $r^\ell\phi_\ell\sim t^{-1}$ as $t\to-\infty$ for all $\ell$. 
Moreover, prescribe the no incoming radiation condition on $\mathcal I^-$. 
Then there exists an $\ell_0\in \mathbb N$, increasing with the prescribed regularity of the data, such that, for all $\ell\leq \ell_0$, the $\ell$-modes $\phi_\ell$ of the corresponding solution will exhibit the following late-time asymptotics near $i^+$:
\begin{align}\label{eq:dnnt}
r\phi_\ell|_{\mathcal I^+}\sim u^{-2-\ell}\log u, &&\phi_\ell|_{r=\text{constant}}\sim \tau^{-2\ell-3}\log \tau, &&\phi_\ell|_{\mathcal H^+}\sim v^{-2\ell-3}\log v.
\end{align} 
Moreover, the projection onto higher $\ell>\ell_0$-modes $\phi_{\ell>\ell_0}$ satisfies the upper bounds
\begin{align}\label{eq:dnt}
r\phi_{\ell>\ell_0}|_{\mathcal I^+}=\mathcal O( u^{-2-\ell_0-\epsilon}), &&\phi_{\ell>\ell_0}|_{r=\text{constant}}=\mathcal O( \tau^{-2\ell_0-3-\epsilon}), &&\phi_{\ell>\ell_0}|_{\mathcal H^+}=\mathcal O( v^{-2\ell_0-3-\epsilon})
\end{align} 
for some $\epsilon>0$.
 If the data are chosen to be smooth, then $\ell_0$ can be chosen to be $\infty$.
\end{conj}
See also the comments in \S\ref{sec:gnl:comments}.

A proof of the above conjecture would require revisiting the proof of Thm.\ \ref{thm:intro:gtl} since, as stated, Thm.\ \ref{thm:intro:gtl} requires boundary data regularity increasing in angular frequency $L$. However, if one imposes fixed, finite regularity, it should still be possible to extract weaker decay (compared to that of Thm.\ \ref{thm:intro:gtl}) from the methods of the proof that is consistent with \eqref{eq:dnt}.
On the other hand, once these modifications are understood, one should be able to directly apply the methods of~\cite{AAG21}, with modifications as in~\cite{II}, to prove the conjecture.

It would also be interesting to find a definitive answer to the question whether or not the rate \eqref{eq:dnt} can be improved \textit{without assuming additional regularity}.

We finally note that, on the one hand, if the $1/t$-decay on initial data is replaced by any \textit{integrable} decay rate, then the logarithms in \eqref{eq:dnnt} would disappear and we would expect the usual Price's law tails. 
\textit{On the other hand, if one considers a timelike hypersurface $\Gamma$ on which $r|_{\Gamma}\to \infty$ as $u\to-\infty$, say, $r|_{\Gamma}(u)\sim |u|$, and only imposes $r\phi_\ell|_{\Gamma}\sim |t|^{-1}$, then the expected modifications to Price's law are much more severe and exactly as in the null case with $p=1$. We will discuss this latter case now.}


\subsubsection{The null case: More severe deviations from Price's law}\label{intro:subsecNP2}
In contrast to the timelike case, it turns out that the data considered in Theorem~\ref{thm:intro:nl}, which were posed for the $\ell=1$-mode on an ingoing null hypersurface ($r\phi_1|_{\mathcal{C}_{\mathrm{in}}}\sim \cc|u|^{-1}$), generally lead to a non-vanishing future Newman--Penrose constant 
\begin{equation}
I_{\ell=1}^{\mathrm{future}}[\phi]:=\lim_{v\to\infty} r^2 \pv(r^2\pv(r\phi_1)-Mr\phi_1)\neq0
\end{equation} if $\cc\neq 0$, cf.~Theorem~\ref{thm:NP}.
In this case, one recovers the following late-time asymptotics (provided that one smoothly extends the data to $\mathcal H^+$):
\begin{align}
r\phi_1|_{\mathcal{I}^+}\sim u^{-2}, &&\phi_1|_{r=R}\sim \tau^{-4},&& \phi_1|_{\mathcal{H}^+}\sim v^{-4},
\end{align}
with the leading-order asymptotics only depending on $I_{\ell=1}^{\mathrm{future}}[\phi]$. These late-time asymptotics are one power worse than the Price's law decay \eqref{eq:Price} for compactly supported data and have also been derived in~\cite{AAG21}.\footnote{\label{footnote7}
In fact, the authors of~\cite{AAG21} first derived late-time asymptotics for data with $I_{\ell=1}^{\mathrm{future}}[\phi]\neq 0$, and then showed that solutions $\phi$ arising from smooth compactly supported data can generically be written as time derivatives of solutions $\phi^T$ that satisfy $I_{\ell=1}^{\mathrm{future}}[\phi^T]\neq 0$. They then showed that time derivatives decay one power faster, which proved Price's law. 
}

In the case of general $\ell\geq 1$, however, the situation is more subtle: 
The data considered in Theorem~\ref{thm:intro:moreg}, i.e.\ $r\phi_L(u,1)\sim |u|^{-p}$, lead, for $p\leq L\neq 0$, to solutions where the usual Newman--Penrose constant is infinite, $I_{\ell=L}^\mathrm{future}[\phi]:=\lim r^2\pv\Phi_L=\infty$, where $\Phi_L$ is defined in section~\ref{sec:generalNP}, eq.\ \eqref{eq:NPfuture}. 
Instead, the following \textit{$(L-p)$-modified Newman--Penrose constant} remains finite and conserved along null infinity (see also Thm.~\ref{thm:moreg}):
\begin{equation}\label{intro:eq:y}
I_{\ell=L}^{\mathrm{future},r^{2-L+p}}[\phi]:=\lim_{v\to\infty} r^{2-L+p}\pv\Phi_L\neq 0.
\end{equation}
With the decay encoded in \eqref{intro:eq:y}, which is $L-p$ powers worse than in the case of finite unmodified N--P constant, we expect that one can further modify the methods of~\cite{AAG21} to then derive late-time asymptotics near $i^+$ which are $L-p+1$ powers slower than the Price's law decay \eqref{eq:Price} and which do not depend on the data's extension towards $\mathcal H^+$ (see Figure~\ref{fig:III:2}), provided that the solution is smooth. Cf.\ the comments in \S\ref{sec:10comments}.
\begin{figure}[htbp]
\floatbox[{\capbeside\thisfloatsetup{capbesideposition={right,top},capbesidewidth=4.4cm}}]{figure}[\FBwidth]
{\caption{Schematic depiction of the situation of \S \ref{intro:subsecNP2}: Given data for $r\phi_\ell$ on $\mathcal C_{\mathrm{in}
}$ which decay like $1/u^p$ near $\mathcal I^-$, the solution has finite $(\ell-p)$-modified N--P constant (see \eqref{intro:eq:y}) on $\mathcal I^+$ by Thm.~\ref{thm:intro:moreg}, provided that $p\leq \ell$. We also depicted the conjectured late-time behaviour near $i^+$.}\label{fig:III:2}}
{\includegraphics[width=200pt]{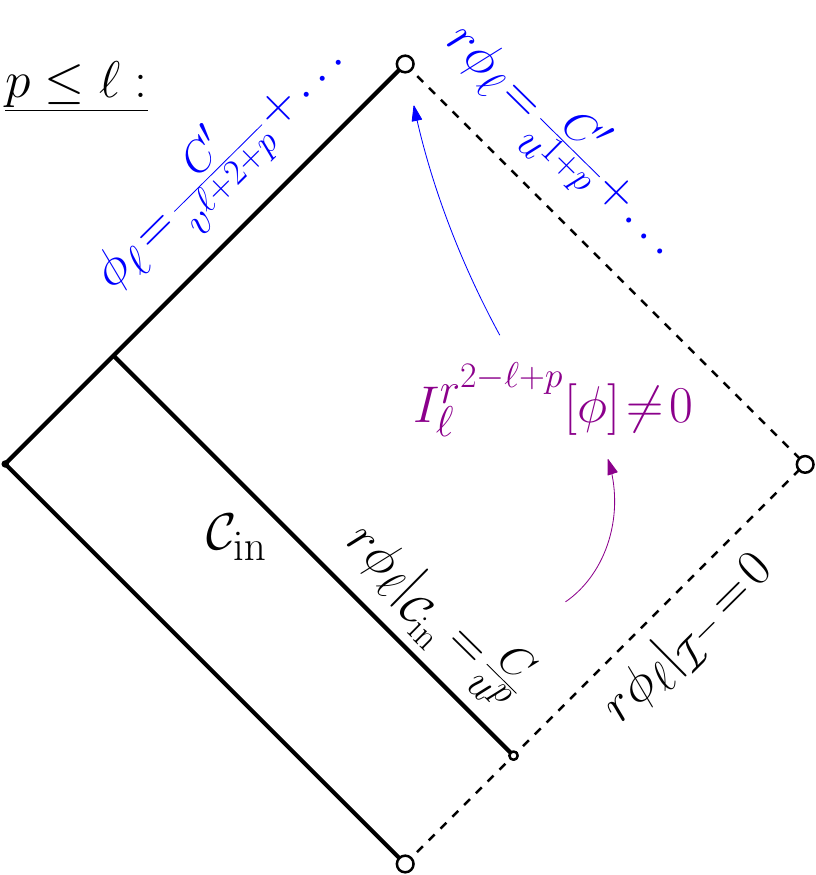}}
\end{figure}

Analogously to Conjecture \ref{conj1}, we also make the following conjecture for the finite regularity problem:
\begin{conj}\label{conj2}
Let $1\leq p \in \mathbb N$, and prescribe  data for $\phi$ on $\mathcal C_{\mathrm{in}}$ that have sufficient but fixed, finite regularity and which satisfy $r\phi_\ell\sim |u|^{-p}$ as $u\to-\infty$ for all $\ell$.
 Moreover, prescribe the no incoming radiation condition on $\mathcal I^-$. 
 Then, for all $\ell\geq p>0$,\footnote{For $\ell=p-1$, one would obtain a logarithmically modified Price's law \eqref{eq:dnnt}, and, for $\ell<p-1$, one would \textit{generically} obtain the usual Price's law behaviour \eqref{eq:Price}. } the $\ell$-modes $\phi_\ell$ of the corresponding solution will exhibit the following late-time asymptotics near $i^+$ along $\mathcal I^+$:
 \begin{equation}\label{intro:eq:yyyy}
 r\phi_\ell|_{\mathcal I^+}\sim u^{-p-1}.
 \end{equation}
 Moreover, there exists an $\ell_0\in \mathbb N$, increasing with the prescribed regularity of the data, such that, away from $\mathcal I^+$, and for some $\epsilon>0$,
\begin{align}\label{intro:eq:yy}
\phi_\ell|_{r=\text{constant}}\sim \tau^{-\ell-p-2},\quad \phi_\ell|_{\mathcal{H}^+}&\sim v^{-\ell-p-2},\quad\text{for all } \ell\in\{p,\dots,\ell_0\},\\
\phi_{\ell>\ell_0}|_{r=\text{constant}}=\mathcal O( \tau^{-\ell_0-p-2-\epsilon}),\quad \phi_{\ell>\ell_0}|_{\mathcal{H}^+}&=\mathcal O( v^{-\ell_0-p-2-\epsilon}).
\end{align}
 If the data are chosen to be smooth, then $\ell_0$ can be chosen to be $\infty$.
\end{conj}
 Remarkably, if one takes $p=1$ in \eqref{eq:intro:ass7}, then the asymptotics \eqref{intro:eq:yyyy}, \eqref{intro:eq:yy} for $\ell\geq 1$ would still be a logarithm faster than the ones for $\ell=0$, \eqref{eq:intro:asy1}--\eqref{eq:intro:asy3}, despite the decay of $\pv(r\phi_\ell)$ towards spatial infinity being slower for $\ell>0$ than for $\ell=0$.

 We note that even if one is willing to assume smoothness, then the modifications to~\cite{AAG21} needed to prove \eqref{intro:eq:yy} are different than those of~\cite{II}, as one now has to deal with a difference in \textit{integer powers} in decay. 
 (In~\cite{II}, we treated non-integer modifications in decay.) 
 We expect that one should be able to use time \textit{derivatives}, rather than time integrals, of our solutions to reduce to the cases treated in~\cite{AAG21}, and then integrate the asymptotics of these time derivatives from $i^+$ to obtain the asymptotics of the original solution.  
 (This would be the opposite procedure of that described in footnote~\ref{footnote7}.) 
 
 In contrast, the fixed, finite regularity problem, i.e.\ a proof of Conjecture~\ref{conj2}, would require much more elaborate modifications.
 In fact, since we conjecture the precise late-time asymptotics \textit{for all} $\ell$ in \eqref{intro:eq:yyyy}, one would now also have to modify the methods of~\cite{AAG21} since the procedure outlined in the previous paragraph would, again, require regularity that is increasing in $\ell$.
 We also want to point out that, since the conjectured asymptotics \eqref{intro:eq:yyyy} along $\mathcal I^+$ are independent of $\ell$, an understanding of the fixed, finite regularity problem would be all the  more important for applications!

\subsubsection{Compactly supported scattering data on \texorpdfstring{$\mathcal H^-$}{H-} and \texorpdfstring{$\mathcal I^-$}{I-}}\label{intro:subsecNP3}
One final natural configuration of data we want to consider is the case of smooth, compactly supported scattering data on $\mathcal I^-$ and the past event horizon $\mathcal H^-$. 
In order to apply our results, we can, without loss of generality, assume that the data on $\mathcal H^-$ are vanishing (this can be achieved by restricting to sufficiently large negative values of $u$). 
Similarly, we can assume, without loss of generality, that the data on $\mathcal I^-$ are supported in $v_1\leq v\leq 1$. If we then integrate the wave equation satisfied by the radiation field $r\phi_L$, namely
\begin{equation}\label{intro:wave:radiation}
\pv\pu(r\phi_L)=\left(1-\frac{2M}{r}\right)\left(\frac{-L(L+1)}{r^2}r\phi_L-\frac{2M}{r^3}r\phi_L\right),
\end{equation}
from $v=v_1$ to $v=1$, we obtain that, generically, $\pu(r\phi_0)(u,1)\sim r^{-3}$ if $L=0$ and  $\pu(r\phi_L)(u,1)\sim r^{-2}$ if $L\geq 1$.  
More precisely, one can derive from \eqref{intro:wave:radiation} that if the data on $\mathcal I^-$ are given by $r\phi_L(-\infty,v)=:G(v)$, then 
\begin{align*}
r^2\cdot \pu(r\phi_L)(u,1)=-L(L+1)\int_{v_1}^1 G(v)\dd v+\mathcal O(r^{-1}), &&\text{if } L>0,\\
r^3\cdot \pu(r\phi_0)(u,1)=-2M\int_{v_1}^1 G(v)\dd v+\mathcal O(r^{-1}), &&\text{if } L=0,
\end{align*}
See also  \S 2.2 and \S 6 of~\cite{I} for a detailed discussion of this restricted to the spherically symmetric mode.
Thus, since the integrals above are non-vanishing for generic scattering data $G$, one can show that Theorem~\ref{thm:intro:moreg} applies, with (generically) $p=2$ if $L=0$ and with $p=1$ if $L\geq 1$.

 The results of Theorem~\ref{thm:intro:moreg} then show that if $L=0$, then $\lim_{v\to\infty}r^3\pv(r\phi_0)<\infty$, whereas if $L\geq 1$, then, generically, $I_{\ell=L}^{\mathrm{future},r^{2-L+1}}[\phi]$ is finite and non-vanishing. 
 Therefore, if $L=0$, one obtains the following late-time asymptotics near $i^+$~\cite{AAG21}:
\begin{align}
r\phi_0|_{\mathcal{I}^+}\sim u^{-2}, &&\phi_0|_{r=R}\sim \tau^{-3},&& \phi_0|_{\mathcal{H}^+}\sim v^{-3}.
\end{align}

On the other hand, if $L>0$, then, since $p=1$, Conjecture~\ref{conj2} would imply that, generically,
\begin{align}
r\phi_L|_{\mathcal{I}^+}\sim u^{-2}, &&\phi_L|_{r=R}\sim \tau^{-L-3},&& \phi_L|_{\mathcal{H}^+}\sim v^{-L-3}.
\end{align}

We are thus led to a third conjecture (see Figure~\ref{fig:III:3}):
\begin{conj}\label{conj3}
Consider compactly supported scattering data on $\mathcal H^-$ and $\mathcal I^-$ for \eqref{waveequation}, {supported on all angular frequencies}, with sufficient but finite regularity. 
Then there exists an $\ell_0\in \mathbb N$, increasing with the prescribed regularity of the data, such that, away from $\mathcal I^+$, and for some $\epsilon>0$,
\begin{align}
\phi_\ell|_{r=\text{constant}}\sim \tau^{-\ell-p-2},\quad \phi_\ell|_{\mathcal{H}^+}&\sim v^{-\ell-p-2},\quad\text{for all } \ell\in\{0,\dots,\ell_0\},\\
\phi_{\ell>\ell_0}|_{r=\text{constant}}=\mathcal O( \tau^{-\ell_0-p-2-\epsilon}),\quad \phi_{\ell>\ell_0}|_{\mathcal{H}^+}&=\mathcal O( v^{-\ell_0-p-2-\epsilon}).
\end{align}
On the other hand, along future null infinity $\mathcal I^+$, we have the asymptotic expression
\begin{equation}\label{super}
r\phi|_{\mathcal I^+}=\sum_{\ell=0}^\infty\sum_{m=-\ell}^\ell C_{\ell m}\cdot Y_{\ell m}(\theta,\phi)u^{-2}+o(u^{-2})
\end{equation}
for some constants $C_{\ell m}$ which can be computed explicitly from the scattering data on $\mathcal I^-$ and which are generically non-zero.
\end{conj}
\begin{figure}[htbp]
\floatbox[{\capbeside\thisfloatsetup{capbesideposition={right,top},capbesidewidth=4.4cm}}]{figure}[\FBwidth]
{\caption{
Given compactly supported scattering data for $r\phi_\ell$ on $\mathcal H^-$ and $\mathcal I^-$, the solution generically decays like $1/u$ near $\mathcal I^-$ (away from the support of the data) unless $\ell=0$. It thus has finite $(\ell-1)$-modified N--P constant (see \eqref{intro:eq:y}) on $\mathcal I^+$ by Thm.~\ref{thm:intro:moreg}. We also depicted the conjectured late-time asymptotics for all $\ell\geq 0$.  }
\label{fig:III:3}}
{\includegraphics[width=200pt]{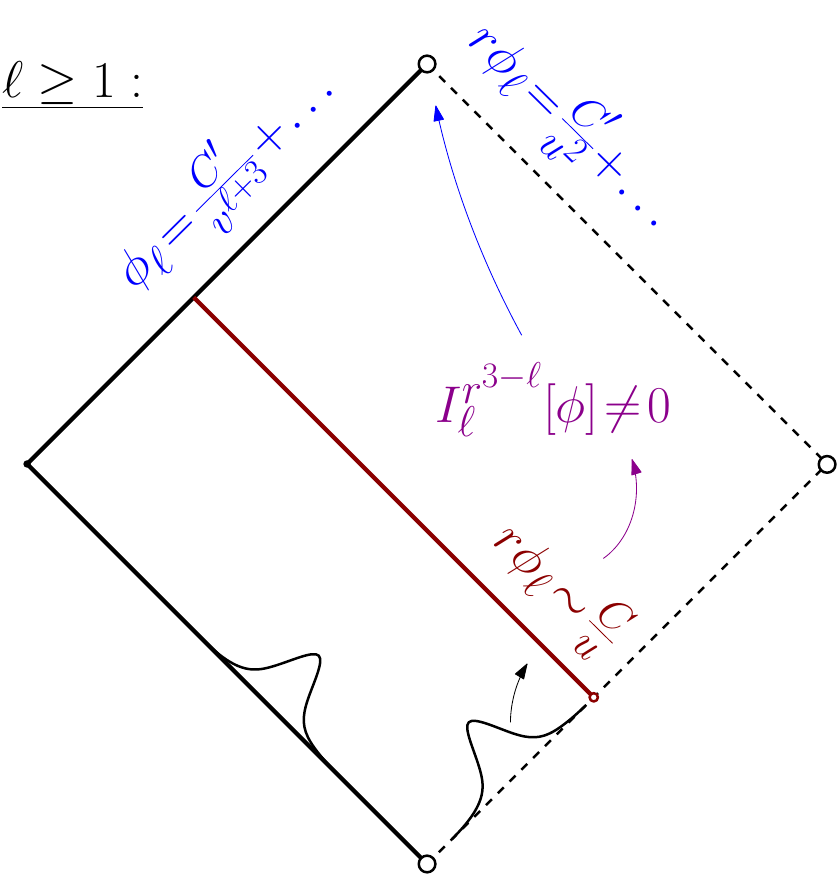}}
\end{figure}

The asymptotics \eqref{super} would be in stark contrast to the usual expectation that the asymptotic behaviour on $\mathcal I^+$, i.e.\ the \textit{physically measurable} behaviour, is dominated by low frequencies. 
It would therefore also be interesting to find the precise form of the constants $ C_{\ell m}$ to see how much each frequency contributes. We leave this, as well as the resolution of Conjectures~\ref{conj1}--\ref{conj3}, to future work.

\subsection{Structure and guide to reading the paper}\label{sec:intro:gouda}
This paper is structured as follows:
We first recall the family of Schwarzschild spacetimes and recall some geometric preliminaries in \S\ref{sec:geom}. We  then recall relevant results on the wave equation on a Schwarzschild background in \S\ref{sec:wave}. 

The rest of the paper is divided into two parts: In part~\ref{part1}, which comprises sections~\ref{sec:nl}--\ref{sec:timelikenonconstantr}, we focus solely on the $\ell=1$-case. 
This part is written with an emphasis on being instructive and providing intuition for the main results and some (but not all) of the methods used to prove them, which might otherwise be camouflaged by the large amount of inductions in the case of general $\ell$. 
In part~\ref{part2}, which comprises sections~\ref{sec:generalNP}--\ref{sec:moreg:null}, we then develop a more systematic approach for the case of general $\ell$. 

Part~\ref{part1} is structured as follows: 
In \S\ref{sec:nl}, we treat the case of data on an ingoing null hypersurface and prove Theorem~\ref{thm:intro:nl}. 
In \S\ref{sec:timelikeconstantR}, we then treat the case of boundary data on a timelike hypersurface $\Gamma$ of constant area radius and prove Theorem~\ref{thm:intro:tl}. 
We shall explain how to lift the restriction to constant area radii and treat boundary data on more general spherically symmetric hypersurfaces in \S\ref{sec:timelikenonconstantr}.

Part~\ref{part2} is structured as follows: 
In \S\ref{sec:generalNP}, we derive the higher-order approximate conservation laws for general $\ell$-modes and the associated higher-order Newman--Penrose constants.
Equipped with these, we then consider the case of boundary data on a hypersurface $\Gamma$ of constant area radius and prove Theorem~\ref{thm:intro:gtl} in \S\ref{sec:general:timelike}. The generalisation to more general $\Gamma$ proceeds similarly to the one in \S\ref{sec:timelikenonconstantr} and is left to the reader.
The last two sections, \S\ref{sec:general:null} and \S\ref{sec:moreg:null}, again concern data on a null hypersurface. 
In \S\ref{sec:general:null}, we consider the fast initial decay implied by Thm.~\ref{thm:intro:gtl} and prove  Theorem~\ref{thm:intro:gnl}. 
Section~\ref{sec:moreg:null} then generalises these results to slowly decaying data (using different methods) and contains the proof of Theorem~\ref{thm:intro:moreg}.
Various inductive proofs of statements made in part~\ref{part2} are deferred to the appendix~\ref{sec:app}.

Depending on the reader's taste, she can either begin with a thorough reading of part \ref{part1} and then skim through \S\ref{sec:generalNP}--\S\ref{sec:general:null} of part \ref{part2} and carefully read \S\ref{sec:moreg:null}, which introduces an approach not presented in part~\ref{part1}.

Alternatively, she can skip directly to part~\ref{part2} and occasionally refer back to part~\ref{part1} for details, e.g.\ on the treatment of boundary data on a timelike hypersurface of varying area radius in~\S\ref{sec:timelikenonconstantr}.

In any case, an effort was made to make each section of the paper as self-contained as possible.
%
%

\section{Geometric preliminaries}\label{sec:geom}

\subsection{The Schwarzschild spacetime manifold}
The (exterior of the) Schwarzschild family of spacetimes $(\mathcal{M}_M,g_M)$, $M>0$,\footnote{For $M=0$, one recovers the Minkowski spacetime, to which all the results of the paper apply as well.} is given by the family of manifolds
\begin{equation*}
\mathcal{M}_M=\mathbb{R}\times(2M, \infty)\times \mathbb{S}^2,
\end{equation*}
covered by the coordinate chart $(v,r,\theta, \varphi)$, with $v\in\mathbb{R}$, $r\in(2M,\infty)$, $\theta\in(0,\pi)$ and $\varphi\in(0,2\pi)$, where $(\theta,\varphi)$ denote the standard spherical coordinates on $\mathbb{S}^2$, and by the family of metrics
\begin{equation}\label{eq:geometry:vr}
g_M=-4D(r)\dd v^2+4\dd v\dd r+r^2(\dd\theta^2+\sin^2\theta \dd\varphi^2),
\end{equation}
where 
\begin{equation}\label{eq:geometry:D}
D(r)=1-\frac{2M}{r}.
\end{equation}
Upon introducing the \textit{tortoise coordinate} $r^*$ as
\begin{equation}\label{eq:geometry:tortoise}
r^*(r):=R+\int_R^r D^{-1}(r')\dd r'
\end{equation}
for some $R>2M$, and defining 
\begin{equation}\label{eq:geometry:u}
u:=v-r^*,
\end{equation}
one obtains a double null covering  $(u,v,\theta,\varphi)$ of $\mathcal{M}_M$, with $u\in(-\infty, \infty)$, $v\in(-\infty, \infty)$. 
In these coordinates, the metric takes the form 
\begin{equation}\label{eq:geometry:null}
g_M=-4D(r)\dd u\dd v+r^2(\dd\theta^2+\sin^2\theta \dd\varphi^2).
\end{equation}
Throughout the remainder of this paper, we will always work within this $(u,v)$-coordinate system.

From the definitions \eqref{eq:geometry:tortoise}, \eqref{eq:geometry:u}, it follows that $\pv r=-\pu r=D$ and that, for sufficiently large values of $r$:
\begin{equation}\label{eq:v-u=r}
|r-(v-u)+2M\log r|=\mathcal{O}(1).
\end{equation}
The estimate \eqref{eq:v-u=r} will be used frequently throughout the paper.

The vector field $T=\pu+\pv$ is a Killing vector field, the \textit{static} Killing vector field of the Schwarzschild spacetime, which equips $(\mathcal{M}_M,g_M)$ with a time orientation.

While the metric \eqref{eq:geometry:null} in double null coordinates $(u,v)$ becomes singular near $u=\infty$, we see from the form \eqref{eq:geometry:vr} that one can smoothly extend $(\mathcal M_M,g_M)$ to and beyond "$u=\infty$" in $(v,r)$-coordinates. The set "$u=\infty$" is referred to as $\mathcal H^+$, or \textit{future event horizon}, and the region beyond it as the black hole region of the Schwarzschild spacetime. Similarly, one can extend $(\mathcal M_M,g_M)$ to and beyond "$v=-\infty$" (denoted $\mathcal H^-$) by working in coordinates $(u,r)$.

On the other hand, we will often consider functions $f\in C^{\infty}(\mathcal M_M)$ such that  e.g.\ the limit $\lim_{v\to\infty}f(u,v,\theta,\varphi)$ exists and is continuous in $u,\theta$ and $\varphi$.
 In these cases, we will interpret the limit as living on the abstract set $\{u,v=\infty,\theta,\varphi \}$, which we well refer to as \textit{future null infinity} or $\mathcal I^+$.
 Similarly, \textit{past null infinity} $\mathcal I^-$ corresponds to the set of points $\{u=-\infty,v,\theta,\varphi\}$. 
 One can think of these sets as being attached to $\mathcal M$ as boundaries, but the differentiable structure of this extension plays no role in this paper.
 See Figure~\ref{fig:III:4}.

We introduce two null foliations of $\mathcal{M}_M$: A foliation by ingoing null hypersurfaces
\begin{equation}
\mathcal{C}_{v=V}=\mathcal{M}_M\cap\{v=V\},
\end{equation}
and a foliation by outgoing null hypersurfaces
\begin{equation}
\mathcal{C}_{u=U}=\mathcal{M}_M\cap\{u=U\}.
\end{equation}
We will often just write $\mathcal C_U$ instead of $\mathcal C_{u=U}$, and, similarly, $\mathcal C_{V}$ instead of $\mathcal C_{v=V}$. It will always be clear from the context whether we refer to ingoing or outgoing null hypersurfaces.
Moreover, if $f:\mathbb{R}\to (2M,\infty)$ is a smooth function of $u$, we shall denote by $\Gamma_f$ the following timelike hypersurface:
\begin{equation}
\Gamma_{f}=\mathcal{M}_M\cap \{v=v_{\Gamma_{f}}(u)\},
\end{equation}
where $v_{\Gamma_{f}}(u)$ is defined via
\[v_{\Gamma_{f}}(u)-u=2r^*(f(u)).\]
In the special case where $f(u)=R>2M$ is a constant, we simply write 
$\Gamma_f=\Gamma_R$
and $v_{\Gamma_{f}}(u)=v_R(u).$

In the sequel, we will drop the subscript $M$ in $\mathcal{M}_M$ and $g_M$, and we will frequently quotient out the spheres for a given spherically symmetric subset of $\mathcal M$ without writing it (for instance, we will denote the set of all points $(u,v)$ s.t.\ $(u,v,\theta,\varphi)\in \Gamma_R$ by $\Gamma_R$, too).
\begin{figure}[htbp]
\floatbox[{\capbeside\thisfloatsetup{capbesideposition={right,top},capbesidewidth=4.4cm}}]{figure}[\FBwidth]
{\caption{Depiction of the Schwarzschild manifold $\mathcal M$. Also depicted is the region $\Delta$ in which we apply the divergence theorem \eqref{eq:divergencetheorem}.}
\label{fig:III:4}}
{\includegraphics[width=200pt]{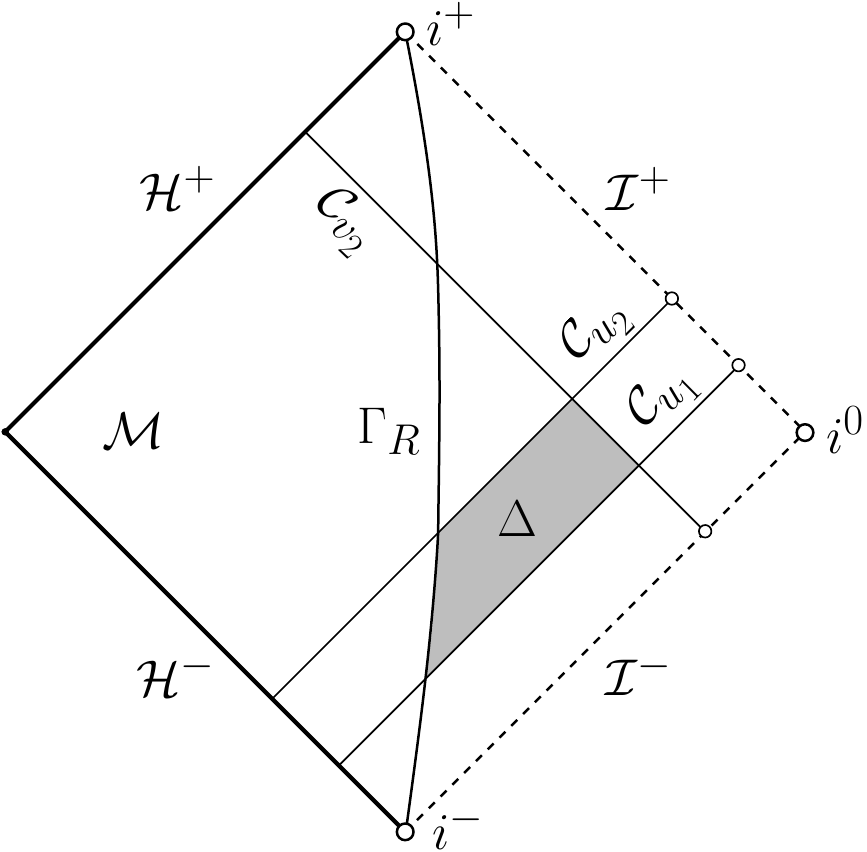}}
\end{figure}
\subsection{The divergence theorem}
Let $\mathcal{D}$ be any simply connected subset of $\mathcal{M}$ with piecewise smooth boundary $\partial\mathcal{D}$. If $J$ is a smooth 1-form, then we have by the divergence theorem:
\begin{equation}\label{eq:divergencetheorem2}
\int_{\mathcal{D}}\div J=\int_{\partial\mathcal{D}}J\cdot n_{\partial\mathcal{D}}
\end{equation}
Here, $n_{\partial\mathcal{D}}$ is the normal to $\partial\mathcal{D}$, and integration over the canonical volume form is implied. If $\partial\mathcal{D}$ contains null pieces, then there is no canonical choice of volume form or normal on these. In this case, we shall choose the product of volume form and normal in such a way that the divergence theorem \eqref{eq:divergencetheorem2} applies (using Stokes' Theorem). 
For instance, if $\Delta$ is the region bounded by $\Gamma_R\cap\{u_1\leq u\leq u_2\}$, $\mathcal{C}_{u_1}\cap\{v_R(u_1)\leq v\leq v_2\}$, $\mathcal{C}_{u_2}\cap\{v_R(u_2)\leq v\leq v_2\}$ and $\mathcal{C}_{v_2}\cap\{u_1\leq u\leq u_2\}$, then we have
\begin{align}\label{eq:divergencetheorem}
&\int_{\Gamma_R\cap\{u_1\leq u\leq u_2\}}  r^2\dd (u+v)\dd \Omega	\,J\cdot(\pu-\pv)
+\int_{\mathcal{C}_{u_1}\cap\{v_R(u_1)\leq v\leq v_2\}} r^2 \dd v \dd \Omega \,J\cdot \pv \nonumber\\
=&\int_{\mathcal{C}_{v_2}\cap\{u_1\leq u\leq u_2\}} r^2\dd u \dd \Omega\, J\cdot \pu 
+\int_{\mathcal{C}_{u_2}\cap\{v_R(u_2)\leq v\leq v_2\}} r^2 \dd v \dd \Omega \,J\cdot \pv
-\int_{\Delta} 2Dr^2 \dd u\dd v\dd \Omega\, \div J ,
\end{align}
where $d\Omega=\sin\theta \dd \theta \dd \varphi$ is the volume form of the unit sphere. 
See Figure~\ref{fig:III:4} for a depiction of this region.

\section{Generalities on the wave equation}\label{sec:wave}
In this section, we collect some important facts about the wave equation 
\begin{equation}
\Box_g\phi:=\nabla^\mu\nabla_\mu \phi=0
\end{equation}
on a Schwarzschild background, where $\nabla$ denotes the Levi--Civita connection of $g$.
\subsection{Existence and uniqueness}
We recall the following two standard existence results:
\begin{prop}[Existence for characteristic initial data]
Let $f\in C^{\infty}(\mathcal{C}_{v_1}\cap\{u_1\leq u\leq u_2\})$ and $h\in C^{\infty}(\mathcal{C}_{u_1}\cap\{v_1\leq v\leq v_2\})$ be two smooth functions satisfying the usual corner condition. 
Then there exists a unique smooth function $\phi:\mathcal{M}\cap\{v_1\leq v\leq v_2,u_1\leq u\leq u_2\}\to \mathbb{R}$ such that
\begin{align*}
\phi|_{\mathcal{C}_{v_1}\cap\{u_1\leq u\leq u_2\}}=f,&& \phi|_{\mathcal{C}_{u_1}\cap\{v_1\leq v\leq v_2\}} =h,
\end{align*}
and
$$\Box_g\phi=0.$$
\end{prop}

\begin{prop}[Existence for mixed characteristic/boundary data]\label{prop:existence:mixedboundary}
Let $f\in C^{\infty}(\Gamma_R\cap\{u_1\leq u\leq u_2\})$ and $h\in C^{\infty}(\mathcal{C}_{u_1}\cap\{v_R(u_1)\leq v\leq v_2\})$ be two smooth functions satisfying the usual corner condition.
Then there exists a unique smooth function $\phi:\mathcal{M}\cap\{u_1\leq u\leq u_2,v_R(u)\leq v\leq v_2\}\to \mathbb{R}$ such that
\begin{align*}
\phi|_{\Gamma_R\cap\{u_1\leq u\leq u_2\}}=f,&& \phi|_{\mathcal{C}_{u_1}\cap\{v_R(u_1)\leq v\leq v_2\}} =h,
\end{align*}
and
$$\Box_g\phi=0.$$
\end{prop}

\subsection{The basic energy currents}
We define, with respect to any coordinate basis, and for any smooth scalar field $f\in C^{\infty}(\mathcal{M})$, the following \textit{energy momentum tensor}:
\begin{align*}
\mathbf{T}_{\mu\nu}[f]:=\partial_\mu f\partial_\nu f-\frac12 g_{\mu\nu} \partial^\xi f\partial_\xi f
\end{align*}
Moreover, if $V$ is any smooth vector field on $\mathcal{M}$, we define the energy current $J^V[f]$ according to
\begin{align*}
J^V[\phi](\cdot):=\mathbf{T}[\phi](V,\cdot).
\end{align*}
With the divergence theorem \eqref{eq:divergencetheorem2} in mind, we compute
\begin{align}\label{eq:def:J}
\div J^V[f]=K^V[f]+\mathcal{E}^V[f],
\end{align}
where
\begin{align}
K^V[f]:=\mathbf{T}^{\mu\nu}\nabla_{\mu} V_{\nu},\\
\mathcal{E}^V[f]:=V(f)\Box_g f.
\end{align}
Note that $K^V[f]$ vanishes if $V$ is Killing (in view of the symmetry of $\mathbf{T}$), whereas $\mathcal{E}^V[f]$ vanishes if $f$ is a solution to the wave equation. Thus, $K^V[f]$ measures the failure of $V$ to be Killing and $\mathcal{E}^V[f]$ measures the failure of $f$ to solve the wave equation.

\subsection{Decomposition into spherical harmonics}\label{sec:decompositionintoYlm}
One can decompose any smooth function $f:\mathcal{M}\to \mathbb{R}$ into its projections onto spherical harmonics,
\[
f=\sum_{\ell'=0}^\infty f_{\ell=\ell'},
\]
such that
\[
f_{\ell=\ell'}(u,v,\theta, \varphi)=\sum_{m=-\ell'}^{m=\ell'}f_{\ell'm}(u,v)Y_{\ell'm}(\theta, \varphi),
\]
where the $Y_{\ell'm}$ are the spherical harmonics. These form a complete basis on $L^2(\mathbb{S}^2)$ of orthogonal eigenfunctions to the spherical Laplacian $\slashed{\Delta}_{\mathbb{S}^2}$, with eigenvalues $-\ell'(\ell'+1)$.
In particular, in view of the spherical symmetry of the Schwarzschild spacetime, if $\phi$ solves $\Box_g\phi=0$, so does $\phi_{\ell=L}$:
$$ \Box_g\phi=0 \implies \Box_g\phi_{\ell=L}=0$$
for any $L\geq 0$. In the sequel, we will frequently suppress the $m$-index of $\phi_{\ell m}(u,v)$ and just write $\phi_\ell$ instead.

Finally, we recall the Poincar\'e inequality on the sphere:
\begin{lemma}
Let $L>0$, and let $f_{\ell\geq L}\in C^2({\mathbb S^2})$ be supported only on  $\ell$-modes with $\ell\geq L$. Then
\begin{equation}\label{poincare}
\int\limits_{\mathbb S^2} f_{\ell\geq L}^2\dd \Omega\leq 
-\frac{1}{L(L+1)}\int\limits_{\mathbb S^2}f_{\ell\geq L}\cdot \slashed{\Delta}_{\mathbb S^2}f_{\ell\geq L} 	\dd \Omega=
\frac{1}{L(L+1)}\int\limits_{\mathbb S^2} |\slashed{\nabla}_{\mathbb S^2}f_{\ell\geq L}|^2\dd \Omega.
\end{equation}
\end{lemma}

\subsection{The commuted wave equations and the higher-order Newman--Penrose constants}\label{sec:thecommutedequations}
In the double null coordinates \eqref{eq:geometry:null}, the wave operator $\Box_g$ acting on any scalar function $f$ takes the form
\begin{equation}\label{eq:box}
\Box_g f=-\frac{\pu\pv f}{D}+\frac{1}{r}\pv f-\frac{1}{r}\pu f+\frac{1}{r^2}\slashed{\Delta}_{\mathbb{S}^2}f.
\end{equation}
Hence, if $\phi$ solves the wave equation $\Box_g\phi=0$, then we obtain the following wave equation for the \textit{radiation field} $r\phi$ (recall that $\pv r=D=-\pu r$):
\begin{equation}\label{eq:pupv rphi general}
\pu\pv(r\phi)=\frac{D}{r^2}\slashed{\Delta}_{\mathbb{S}^2}(r\phi)-\frac{2MD}{r^3}r\phi.
\end{equation}
Notice that if we restrict to the spherically symmetric mode $r\phi_{\ell=0}$, this gives rise to the approximate conservation law
\begin{equation}\label{eq:l=0cons.law}
\pu\pv(r\phi_0)=-\frac{2MD}{r^3}r\phi_0.
\end{equation}
This equation \eqref{eq:l=0cons.law} is closely related to the existence of conserved quantities along null infinity, the so-called the Newman--Penrose constants
\begin{align}
\ILn0[\phi](u)&:=\lim_{v\to\infty}r^2\pv(r\phi_0)(u,v),\\
\ILpn0[\phi](v)&:=\lim_{u\to-\infty}r^2\pu (r\phi_0)(u,v),
\end{align}
which, under suitable assumptions on $\phi$, remain conserved along $\mathcal{I}^+$, $\mathcal{I}^-$, respectively. Equation \eqref{eq:l=0cons.law} (or rather, the non-linear analogue thereof) played a crucial role in proving our results from~\cite{I} and is, in fact, ubiquitous in the studies of asymptotics for the wave equation on Schwarzschild backgrounds, see e.g.~\cite{DR05},~\cite{AAG18b}.

However, for higher $\ell$-modes, the approximate conservation law \eqref{eq:l=0cons.law} is no longer available, and the RHS of $\pu\pv(r\phi_{\ell=L})$ has a bad $r^{-2}$-weight. This difficulty appears already in the Minkowski spacetime, i.e.\ for $M=0$. There, it can be resolved by commuting with $(r^2\pv)^\ell$, $(r^2\pu)^\ell$, respectively. Indeed, if $ M=0$, one has the following precise conservation laws:
\begin{align*}
\pu(r^{-2L-2}(r^2\pv)^{(L+1)}(r\phi_L))=0,\\
\pv(r^{-2L-2}(r^2\pu)^{(L+1)}(r\phi_L))=0.
\end{align*}
One can find generalisations of these conservation laws in Schwarzschild. This is done in \S \ref{sec:generalNP} of the paper. For now, we believe it to be more instructive to only explain what happens to the $\ell=1$-modes.
If we naively commute the wave equation for $\ell =1$, namely
\begin{align}\label{eq:l=1waveequation}
\pu\pv(r\phi_1)=-\frac{2D}{r^2}r\phi_1\left(1+\frac{M}{r}\right),
\end{align}
with $r^2\pv$, then we find
\begin{equation}\label{eq:l)1}
\pu(r^{-2}\pv(r^2\pv(r\phi_1)))=-10 MD\frac{r^2\pv(r\phi_1)}{r^5}-2MD\frac{r\phi_1}{r^4}\left(1+\frac{4M}{r}\right).
\end{equation}
We see that the top-order term in \eqref{eq:l)1} comes with a good $r^{-5}$-weight. Moreover, the problematic $r^{-4}$-weight multiplying $r\phi_1$ can be removed by subtracting $Mr\phi_1$ in the following way:
\begin{equation}\label{eq:l=1cons.law.in.u}
\pu(r^{-2}\pv(r^2\pv(r\phi_1)-M r\phi_1))=-12 MD\frac{r^2\pv(r\phi_1)}{r^5}-6M^2 D\frac{r\phi_1}{r^5}.
\end{equation}
Similarly, for $u$ and $v$ interchanged, we obtain
\begin{align}
\pv(r^{-2}\pu(r^2\pu(r\phi_1)))=-10 MD\frac{r^2\pu(r\phi_1)}{r^5}+2MD\frac{r\phi_1}{r^4}\left(1+\frac{4M}{r}\right)
\end{align}
and 
\begin{equation}\label{eq:l=1cons.law.in.v}
\pv(r^{-2}\pu(r^2\pu(r\phi_1)+M r\phi_1))=-12 MD\frac{r^2\pu(r\phi_1)}{r^5}+6M^2 D\frac{r\phi_1}{r^5}.
\end{equation}
The \textit{approximate conservation laws} \eqref{eq:l=1cons.law.in.u}, \eqref{eq:l=1cons.law.in.v} give rise to the following higher-order Newman--Penrose constants:
\begin{align}\label{eq:NPl=1future}
\ILn1[\phi](u)&:=\lim_{v\to\infty}r^2\pv(r^2\pv(r\phi_1)-Mr\phi_1)(u,v),\\
\ILpn1[\phi](v)&:=\lim_{u\to-\infty}r^2\pu(r^2\pu (r\phi_1)+Mr\phi_1)(u,v),
\end{align}
which, under suitable assumptions on $\phi$, remain conserved along $\mathcal{I}^+$, $\mathcal{I}^-$, respectively.
Equations \eqref{eq:l=1cons.law.in.u} and \eqref{eq:l=1cons.law.in.v} will play a similar role in the asymptotic analysis of the $\ell=1$-mode as equation \eqref{eq:l=0cons.law} did in the analysis of~\cite{I}.
\subsection{Notational conventions}

We use the notation that $f\sim g$ (or $f\lesssim g$) if there exists a uniform constant $C>0$ such that $C^{-1}g\leq f\leq Cg$ (or $f\leq C g$). Similarly, we use the convention that $f=\mathcal O(g)$ if there exists a uniform constant $C>0$ such that $|f|\leq C g$. If $f$ and $g$ are functions depending on a single variable $x$, and if $k\in\mathbb N$, we also say that $f=\mathcal O_k(g)$ if there exist uniform constants $C_j>0$ such that $|\partial_x^j f|\leq C_j|\partial_x^j g|$ for all $j\leq k$.
 Finally, we use the usual algebra of constants ($C+D=C=CD\dots$).

\newpage
\part{The case \texorpdfstring{$\ell=1.$}{L=1.}}\label{part1}
In this part of the paper, we focus solely on the analysis of the $\ell=1$-modes. 
The aim of this part is to give some intuition for the decay rates and the methods used to prove them. 
The confident reader may wish to skip directly to the discussion of general $\ell$ in Part~\ref{part2}.

We first treat the case of data on an ingoing null hypersurface and prove Theorem~\ref{thm:intro:nl} in \S \ref{sec:nl}. 
We then treat the case of boundary data on a timelike hypersurface of constant area radius $r$ and prove Theorem~\ref{thm:intro:tl} in \S \ref{sec:timelikeconstantR}. 
Finally, we explain how to generalise to the case of boundary data on timelike hypersurfaces on which $r$ is allowed to vary in \S \ref{sec:timelikenonconstantr}.

Throughout Part~\ref{part1}, $\phi$ will always denote a solution to $\Box_g\phi=0$ which is localised on an $(\ell,m)$-frequency with $\ell=1$, $|m|\leq 1$ fixed. We use the notation from \S\ref{sec:decompositionintoYlm}, that is, we write $\phi=\phi_{\ell=1}=\phi_1(u,v)\cdot Y_{1m}(\theta,\varphi)$.
\section{Data on an ingoing null hypersurface \texorpdfstring{$\mathcal C_{v=1}$}{C(v=1)}}\label{sec:nl}
In this section, we consider solutions $\phi$ arising from polynomially decaying data on an ingoing null hypersurface $\mathcal{C}_{v=1}$ and from vanishing data on $\mathcal I^-$, and we show asymptotic estimates near spatial infinity for these. In particular, this section contains the proof of Theorem~\ref{thm:intro:nl} from the introduction.

\subsection{Initial data assumptions and the main theorem (Theorem~\ref{thm:nl})}
Prescribe smooth characteristic/scattering data for the wave equation \eqref{waveequation} restricted to $(1,m)$  which satisfy on $\mathcal C_{v=1}$
\begin{align}\label{eq:nl:ass1}
	r^2\pu\pho(u,1)=C_{\mathrm{in}}^{(1)}+\mathcal{O}(r^{-1}),\\
	r^2\pu(r^2\pu\pho)(u,1)=\ccc+\mathcal{O}(r^{-\eta})\label{eq:nl:ass2}
\end{align}
for some $\eta\in(0,1)$, and which moreover satisfy for all $v\geq 1$:
\begin{equation}\label{eq:nl:assNoIncoming}
\lim_{u\to-\infty}\pv^n(r\phi_1)(u,v)=0
\end{equation}
for $n=0,1,2$. We interpret this latter assumption as the no incoming radiation condition.

The main result of this section then is:
\begin{thm}\label{thm:nl}
By standard scattering theory~\cite{DRSR18}, there exists a unique smooth scattering solution $\phi_1\cdot Y_{1m}$ in $\mathcal M\cap\{v\geq 1\}$ attaining these data.
Let $U_0$ be a sufficiently large negative number. Then, for all $(u,v)\in\mathcal{D}:=(-\infty,U_0]\times [1,\infty)$, the outgoing derivative of $r\phi_1$ satisfies, for fixed values of $u$, the following asymptotic expansion as $\mathcal{I}^+$ is approached:
\begin{align}
\begin{split}
r^2\pv\pho(u,v)=&-\cc-2\int_{-\infty}^uF(u')\dd u'-\frac{2M\cc-2M\int_{-\infty}^uF(u')\dd u'}{r}\\
&-2M(\ccc-2M\cc)\frac{\log r-\log|u|}{r^2}+\mathcal{O}(r^{-2}),
\end{split}
\end{align}
where $F(u)$ is given by the limit of the radiation field $r\phi_1$ on $\mathcal{I}^+$
\begin{equation}
F(u):=\lim_{v\to\infty}r\phi_1(u,v)=\frac{\ccc-2M\cc}{6|u|^2}+\mathcal{O}(|u|^{-2-\eta}).
\end{equation}
In particular, if $M(\ccc-2M\cc)\neq 0$, then peeling fails at future null infinity.
\end{thm}
\begin{rem}
The methods of our proof can also directly be  applied to data which only have 
\[r^2\pu\pho(u,1)=C_{\mathrm{in}}^{(1)}+\mathcal{O}_1(r^{-\eta})\]
for $\eta\in(0,1)$. In that case, one would, schematically, obtain $\pv\pho=\frac{f_1(u)}{r^2}+\frac{f_2(u)}{r^3}+\mathcal{O}(r^{-3-\eta})$.
\end{rem}
In order to prove the theorem, we shall first establish the asymptotics of $r\phi_1$, using equations \eqref{eq:l=1waveequation} and \eqref{eq:l=1cons.law.in.v}, in \S\ref{sec:nl:phi}, and then establish the asymptotics of $\pv(r\phi_1)$, using \eqref{eq:l=1waveequation} and \eqref{eq:l=1cons.law.in.u}, in \S\ref{sec:nl:dvphi}. 
We shall make some important comments in \S\ref{sec:nl:comments}.

\subsection{Asymptotics for \texorpdfstring{$r\phi_1$}{r phi1}}\label{sec:nl:phi}
We recall from \S \ref{sec:thecommutedequations} the two wave equations
\begin{align}\label{eq:nl:waveequation}
\pu\pv(r\phi_1)=-\frac{2D}{r^2}r\phi_1\left(1+\frac{M}{r}\right)
\end{align}
and
\begin{align}\label{eq:nl:valmost}
\pv(r^{-2}\pu(r^2\pu(r\phi_1)))=-10 MD\frac{r^2\pu(r\phi_1)}{r^5}+2MD\frac{r\phi_1}{r^4}\left(1+\frac{4M}{r}\right).
\end{align}
The reason that we here choose to work with \eqref{eq:nl:valmost} rather than \eqref{eq:l=1cons.law.in.v} is that, in view of the no incoming radiation condition, the bad $r^{-4}$-weight multiplying $r\phi_1$ in \eqref{eq:nl:valmost} is not a problem (since  $r\phi_1$ itself will decay).

Throughout the rest of \S\ref{sec:nl}, $U_0$ will be a sufficiently large negative number (the largeness depending only on data), and $\mathcal D$ will be as in Thm.~\ref{thm:nl}.

\subsubsection{A weighted energy estimate and almost-sharp decay for \texorpdfstring{$r\phi_1$}{r phi1}}
We first prove almost-sharp decay using an energy estimate:
\begin{prop}\label{prop:l:energy}
Define the following energies:
\begin{align*}
E_{q}^{[u_1,u_2]}(v):=\int_{u_1}^{u_2} |u|^q \left( 	(\pu\pho)^2 +(r\phi_1)^2 \frac{2D}{r^2}\left(1+\frac{M}{r}\right)\right) (u,v)\dd u,\\
E_{q}^{[v_1,v_2]}(u):=\int_{v_1}^{v_2}|u|^q \left( 	(\pv\pho)^2 +(r\phi_1)^2 \frac{2D}{r^2}\left(1+\frac{M}{r}\right)\right)(u,v)\dd v.
\end{align*}
Then the following energy inequality holds for all $v_2>v_1\geq 1$, $q\geq0$ and for  $0>U_0\geq u_2>u_1$:
\begin{align}
	E_{q}^{[u_1,u_2]}(v_2)+E_{q}^{[v_1,v_2]}(u_2)\leq E_{q}^{[u_1,u_2]}(v_1)+E_{q}^{[v_1,v_2]}(u_1).
\end{align}
\end{prop}
\begin{proof}
Multiply the wave equation \eqref{eq:nl:waveequation} with $2 T(r\phi_1)$ (recall that $T=\pu+\pv$) to obtain:
\begin{align*}
0=\pu\left((\pv\pho)^2\right)+\pv\left((\pu\pho)^2\right)+T\left(\frac{2D(r\phi_1)^2}{r^2}\left(1+\frac{M}{r}\right)\right).
\end{align*}
This would already lead to the standard energy estimate, but we can exploit a certain monotonicity to obtain a weighted energy estimate: 
For this, we multiply the above expression  with $|u|^q$ and recall that $u<0$:
\begin{align*}
0=\pv\left(|u|^q	(\pu\pho)^2 +\frac{2D|u|^q(r\phi_1)^2 }{r^2}\left(1+\frac{M}{r}\right)\right)\\
+\pu\left(|u|^q	(\pv\pho)^2 +\frac{2D|u|^q(r\phi_1)^2 }{r^2}\left(1+\frac{M}{r}\right)\right)\\
+q|u|^{q-1}\left((\pv\pho)^2+\frac{2D\pho^2}{r^2}\left(1+\frac{M}{r}\right)\right).
\end{align*}
Finally, integrating this in $u$ and $v$ using the fundamental theorem of calculus gives
\begin{align}\begin{split}
	E_{q}^{[u_1,u_2]}(v_2)&+E_{q}^{[v_1,v_2]}(u_2)= E_{q}^{[u_1,u_2]}(v_1)+E_{q}^{[v_1,v_2]}(u_1)\\
	&-\int_{v_1}^{v_2}\int_{u_1}^{u_2}q|u|^{q-1}\left((\pv\pho)^2+\frac{2D\pho^2}{r^2}\left(1+\frac{M}{r}\right)\right)\dd u\dd v.
\end{split}\end{align}
\end{proof}
\begin{rem}
A similar result holds for any fixed angular frequency solution. Moreover, in view of Lemma~\ref{poincare}, the above proof also works for any $\phi$ supported on angular frequencies $\ell\geq L$, for some $L\geq 1$.
\end{rem}
From this weighted $L^2$-estimate, we can already derive almost-sharp pointwise decay:
\begin{cor}\label{cor:almostsharp}
There is a constant $C$ depending only on data such that, throughout $\mathcal{D}$:
\begin{align}\label{eq:nl:almost.sharp}
|r\phi_1(u,v)|\leq \frac{C}{|u|},&& |\pu(r\phi_1)|\leq\frac{C}{|u|^2}.
\end{align}
Moreover, we have that, for all $v\geq 1$:
\begin{equation}
\lim_{u\to-\infty}r^2\pu(r\phi_1)(u,v)\equiv\cc.
\end{equation}
\end{cor}
\begin{proof}
We consider the energy estimate above with  $q=2$ and let $(u,v)\in\mathcal D$. Then
\begin{align}\label{eq:nl:almostsharponphi}
\begin{split}
r\phi_1(u,v)&=\int_{-\infty}^u\pu\pho(u',v)\dd u'\\
&\leq \left(\int_{-\infty}^u |u'|^{-2}\dd u'\right)^\frac12  \left(\int_{-\infty}^u |u'|^{2}(\pu\pho)^2(u,v)\dd u'\right)^\frac12 \\
&\leq \left(\int_{-\infty}^u |u'|^{-2}\dd u'\right)^\frac12  \left(E_2^{[-\infty,u]}(1)+\lim_{u'\to-\infty}E_2^{[1,v]}(u')\right)^\frac12\leq \frac{C}{|u|}
\end{split}
\end{align}
for some constant $C$ solely determined by initial data. 
Here, we used the no incoming radiation condition \eqref{eq:nl:assNoIncoming} in the first step, Cauchy--Schwarz in the second step, and the energy estimate in the third step. 
In the last estimate, we then inserted the initial data assumptions\footnote{Recall that, in view of \eqref{eq:v-u=r}, $r(1,u)=|u|+\mathcal{O}(\log |u|)$.} \eqref{eq:nl:ass1} and used that $\lim_{u'\to-\infty}E_2^{[1,v]}(u')=0$. To show this latter statement, consider first the energy estimate with $q=0$ to obtain a bound of the form $\phi_1\lesssim r^{-\frac12}$. Then, insert this bound into \eqref{eq:nl:waveequation}  to obtain $\pv(r\phi_1)\lesssim r^{-\frac12}$, and repeat the argument with, say, $q=1/2$, and iterate.

Plugging the bound \eqref{eq:nl:almostsharponphi} into the wave equation \eqref{eq:nl:waveequation} and integrating from initial data $v=1$, we moreover obtain that
\begin{align*}
|\pu(r\phi_1)|\leq \frac{C}{u^2},
\end{align*}
and that, in fact, the limit of $|u|^2\pu\pho$ remains constant along $\mathcal{I}^-$. 
\end{proof}

\subsubsection{Asymptotics for \texorpdfstring{$\pu\pho$}{d/du (r phi1)} and \texorpdfstring{$r\phi_1$}{r phi1}}
We now make the decay from Corollary~\ref{cor:almostsharp} sharp:
\begin{prop}
There is a constant $C$ depending only on data such that $r\phi_1$ satisfies the following asymptotic expansion throughout $\mathcal{D}$:
\begin{equation}\label{eq:nl:asy.phi.r}
\left|r\phi_1(u,v)-\frac{\cc}{r}-\frac{\ccc-2M\cc}{6|u|^2}\right|\leq\frac{C}{|u|^{2+\eta}}+\frac{C}{r|u|}.
\end{equation}
In particular, we have
\begin{equation}\label{eq:nl:asy.phi.u}
\lim_{v\to\infty}r\phi_1(u,v)=\frac{\ccc-2M\cc}{6|u|^2}+\mathcal{O}(|u|^{-2-\eta}).
\end{equation}
\end{prop}
\begin{proof}
We integrate the approximate conservation law \eqref{eq:nl:valmost} from $v=1$:
\begin{align}\label{eq:nl:proof1}
	\begin{split}
	r^{-2}\pu(r^2\pu\pho)&(u,v)=r^{-2}\pu(r^2\pu\pho)(u,1)\\
						&+\int_1^v \left(\frac{-10MD\pu\pho}{r^3}+\frac{2MDr\phi_1}{r^4}\left(1+\frac{4M}{r}\right)\right)(u,v') \dd v'.
	\end{split}
\end{align}
Using that $\pv r=D$ and plugging in the initial data assumption \eqref{eq:nl:ass2} as well as the almost sharp bounds obtained in Corollary~\ref{cor:almostsharp}, we obtain
\begin{equation}\label{eq:l:canc:1}
\pu(r^2\pu\pho)(u,v)\lesssim\frac{r^2}{|u|^{4}},
\end{equation}
from which, in turn, we obtain via integrating that 
\begin{equation}\label{eq:l:canc:2}
\left|r^2\pu\pho-\lim_{u\to -\infty}r^2\pu\pho\right|\lesssim\int_{-\infty}^u \frac{r^2}{|u'|^{4}}\dd u'\lesssim\frac{r^2}{|u|^{3}},
\end{equation}
where the last inequality can be seen by recalling that $r\sim v-u$, or by an integration by parts, see also eq.\ \eqref{eq:nl:stupidintegral1} below. Now, by Corollary~\ref{cor:almostsharp},  we have
$
\lim_{u\to -\infty}r^2\pu\pho=\cc.
$
Thus, integrating once more in $u$ and using that $r\phi_1$ vanishes on $\mathcal{I}^-$, we obtain that
\begin{equation*}
\left|r\phi_1-\cc r^{-1}\right|\lesssim |u|^{-2}.
\end{equation*}

This estimate provides us with the leading-order behaviour of $r\phi_1$ in $r$. 
To also understand the leading-order $u$-decay of $r\phi_1$, we insert our improved bounds back into equation \eqref{eq:nl:proof1}:
\begin{align}
\begin{split}
r^{-2}\pu(r^2\pu\pho)(u,v)=\frac{\ccc}{|u|^4}
+\int_{1}^v\frac{-10M\cc}{r^5}+\frac{2M\cc}{r^5}\dd v' 
+\mathcal{O}(|u|^{-4-\eta}).
\end{split}
\end{align}
 Hence, by again converting the $v$-integration into $r$-integration using $\pv r=D$,
\begin{equation}
\pu(r^2\pu\pho)=r^2\left( \frac{\ccc}{|u|^4}-\frac{2M\cc}{|u|^4}+\frac{2M\cc}{|r|^4}\right)+\mathcal{O}(r^2|u|^{-4-\eta}).
\end{equation}
Integrating this from past null infinity, we again encounter the integral 
$\int_{-\infty}^u\frac{r^2}{|u'|^4}\dd u'$. We compute this as follows:
\begin{align}\label{eq:nl:stupidintegral1}
\begin{split}
\int\frac{r^2}{|u|^4}\dd u&=\int \pu\left(\frac{r^2}{3|u|^3}\right)+\frac{2Dr}{3|u|^3}\dd u\\
&=\frac{r^2}{3|u|^3}+\int \pu\left(\frac{r}{3|u|^2}\right)+\frac{D}{3|u|^2}-\frac{4M}{3|u|^3}\dd u
=\frac{1}{3}\sum_{k=0}^2\frac{r^k}{|u|^{k+1}}+\mathcal{O}(|u|^{-2}).
\end{split}
\end{align}
We therefore obtain the following estimate for $\pu\pho$:
\begin{equation}
\pu\pho(u,v)=\frac{\cc}{r^2}+\frac{\ccc-2M\cc}{3}\left(\frac{1}{|u|^3}+\frac{1}{|u|^2r}+\frac{1}{|u|r^2}\right)+\mathcal{O}(r^{-3}+|u|^{-3-\eta}).
\end{equation}
In particular, we thus get that
\begin{equation}
\lim_{v\to\infty}\pu(r\phi_1)(u,v)=\frac{\ccc-2M\cc}{3|u|^3}+\mathcal{O}(|u|^{-3-\eta}).
\end{equation}
Integrating once more in $u$ finishes the proof of the proposition.
\end{proof}

\subsection{Asymptotics for \texorpdfstring{$\pv(r\phi_1)$}{d/dv(r phi1)} and proof of Thm.~\ref{thm:nl}}\label{sec:nl:dvphi}
Equipped with an asymptotic expression for $r\phi_1$, we can now compute the asymptotics of $\pv(r\phi_1)$. We first derive the leading-order asymptotics of $\pv\pho$ up to order $\mathcal{O}(r^{-3})$,  using only the wave equation \eqref{eq:nl:waveequation}, and then determine the next-to-leading-order asymptotics up to $\mathcal{O}(r^{-4}\log r)$ using the commuted equation \eqref{eq:l)1}.
\subsubsection{Leading-order asymptotics of \texorpdfstring{$\pv(r\phi_1)$}{d/dv(r phi1)}}
Plugging the asymptotics \eqref{eq:nl:asy.phi.r} of $r\phi_1$ into the wave equation \eqref{eq:nl:waveequation} and integrating the latter from past null infinity, we obtain
\begin{align}\label{eq:nl:pups}
\pv\pho(u,v)=-\frac{\cc}{r^2}+\mathcal{O}(r^{-2}|u|^{-1}).
\end{align}
In order to find the $\mathcal{O}(r^{-2}|u|^{-1})$-term, we commute the wave equation with $r^2$,
\begin{equation}\label{eq:l:proof2}
\pu(r^2\pv\pho)=-2Dr\pv\pho-2D\left(r\phi_1+\frac{M}{r}r\phi_1\right),
\end{equation}
to find, upon integrating, that
\begin{equation}\label{eq:nl:proof3}
r^2\pv\pho(u,v)=-\cc-\frac{\ccc-2M\cc}{3|u|}+\mathcal{O}\left(\frac{\log (v-u)-\log |u|}{v}+\frac{1}{|u|^{1+\eta}}+\frac{1}{r}\right),
\end{equation}
where we used eq.\ \eqref{eq:v-u=r} and the fact that
\begin{align}\label{eq:nl:stupidintegral1.5}
\int_{-\infty}^u\frac{1}{r(u',v)|u'|}\dd u'\sim\int_{-\infty}^u\frac{1}{(v-u')|u'|}\dd u'=\frac{\log(v-u)-\log|u|}{v}.
\end{align}
In fact, the $\mathcal O(\log r)$-terms in \eqref{eq:nl:proof3} do not appear: By writing $r\phi_1$ as\footnote{We write $(\lim_{\mathcal{I}^+}f)(u)=\lim_{v\to\infty}f(u,v)$.} $r\phi_1=\lim_{\mathcal{I}^+}r\phi_1-\int_v^\infty\pv\pho$ in eq.\ \eqref{eq:l:proof2}, we can improve the asymptotic estimate \eqref{eq:nl:proof3} to
\begin{align*}
r^2\pv\pho(u,v)=-\cc-2\int_{-\infty}^u\lim_{v\to\infty}r\phi_1(u',v)\dd u'+\mathcal{O}(r^{-1}).
\end{align*}
This cancellation is related to the one that gives rise to the approximate conservation law \eqref{eq:l=1cons.law.in.u}.
In the above, we used (see also eq.\ (4.49) of~\cite{I}) that 
\begin{equation}\label{eq:nl:stupidintegral2}
\int_{-\infty}^u \frac{\log(v-u')-\log|u'|}{vr}\dd u'\lesssim\int_{-\infty}^u \frac{\log(v-u')-\log|u'|}{v(v-u')}\dd u'\leq \frac{\pi^2}{6}\frac{1}{v-u}\lesssim \frac{1}{r}.
\end{equation}
We summarise our findings in
\begin{prop}\label{prop4.3}
We have the following asymptotics throughout $\mathcal{D}$:
\begin{align}
\pv\pho(u,v)&=\frac{\lim_{\mathcal{I}^+}r^2\pv\pho(u)}{r^2}+\mathcal{O}(r^{-3}),\\
r\phi_1(u,v)&=\lim_{\mathcal{I}^+}r\phi_1(u)-\frac{\lim_{\mathcal{I}^+}r^2\pv\pho(u)}{r}+\mathcal{O}(r^{-2}),
\end{align}
where $\lim_{\mathcal{I}^+}r\phi_1(u)$ is given by \eqref{eq:nl:asy.phi.u}, and where
\begin{equation}
\lim_{\mathcal{I}^+}r^2\pv\pho(u)=-\cc-2\int_{-\infty}^u\lim_{\mathcal{I}^+}r\phi_1(u')=-\cc-\frac{\ccc-2M\cc}{3|u|}+\mathcal{O}(|u|^{-1-\eta}).
\end{equation}
\end{prop}
\subsubsection{Next-to-leading-order asymptotics for \texorpdfstring{$\pv\pho$}{d/dv(r phi1)}  (Proof of Thm.~\ref{thm:nl})}
\begin{proof}[Proof of Theorem~\ref{thm:nl}]
Equipped with the leading-order asymptotics for $\pv\pho$ and $r\phi_1$, we now find the asymptotic behaviour of $\pv(r^2\pv\pho)$ using the commuted wave equation
\begin{align}\label{eq:nl:ualmost}
\pu(r^{-2}\pv(r^2\pv(r\phi_1)))=-10 MD\frac{r^2\pv(r\phi_1)}{r^5}-2MD\frac{r\phi_1}{r^4}\left(1+\frac{4M}{r}\right).
\end{align}
By the no incoming radiation condition \eqref{eq:nl:assNoIncoming} and the fundamental theorem of calculus, we have
\begin{align}
	\begin{split}
	r^{-2}\pv(r^2\pv\pho)&(u,v)
						=\int_{-\infty}^u \frac{-10MDr^2\pv\pho}{r^5}-\frac{2MDr\phi_1}{r^4}\left(1+\frac{4M}{r}\right) \dd u'.
	\end{split}
\end{align}
Plugging the asymptotics from Prop.~\ref{prop4.3} into the above, we obtain that
\begin{align}\label{check1}
	\begin{split}
	&r^{-2}\pv(r^2\pv\pho)(u,v)\\
	&=\int_{-\infty}^u -\frac{8MD\lim_{\mathcal{I}^+}r^2\pv\pho(u')}{r^5}-\frac{2MD\lim_{\mathcal{I}^+}r\phi_1(u')}{r^4} \dd u'+\mathcal{O}(r^{-5}).
		\end{split}
\end{align}
Evaluating the integrals in a similar way to \eqref{eq:nl:stupidintegral1.5}, we thus find
\begin{align}\label{eq:rd1}
	r^2\pv(r^2\pv\pho)=2M\cc-M\frac{\ccc-2M\cc}{3|u|}+\mathcal{O}\left(\frac{\log (1-v/u)}{v}+\frac{1}{|u|^{1+\eta}}+\frac{1}{r}\right).
\end{align}
Notice that the $\mathcal{O}$-terms in \eqref{eq:rd1} all integrate to $\mathcal{O}(1/r)$ when multiplied by $1/r$ (cf.\ \eqref{eq:nl:stupidintegral2}).

To find the next-to-leading-order logarithmic terms, we commute the approximate conservation law \eqref{eq:nl:ualmost} with $r^4$:
\begin{align*}
\pu(r^2\pv(r^2\pv\pho))&=-\frac{4D}{r} r^2\pv(r^2\pv\pho)
-\frac{10MD}{r}r^2\pv\pho-2MDr\phi_1\left(1+\frac{4M}{r}\right).
\end{align*}
Integrating this from past null infinity and plugging in (as in \eqref{check1}) the asymptotics for $r^2\pv(r^2\pv\pho)$, $r^2\pv\pho$ and $r\phi_1$ from \eqref{eq:rd1} and Prop.~\ref{prop4.3}, respectively,  we find:
\begin{nalign}
	&r^{2}\pv(r^2\pv\pho)(u,v)\\
	=&2M\cc
	+\int_{-\infty}^u	\frac{12MD}{r}\frac{\ccc-2M\cc}{3|u'|} -2M\lim_{\mathcal{I}^+}r\phi_1(u') \dd u'+\mathcal{O}(r^{-1})\\
	=&2M\cc-2M\int_{-\infty}^u\lim_{\mathcal{I}^+}r\phi_1\dd u'+4M(\ccc-2M\cc)\frac{\log (v-u)-\log|u|}{v}+\mathcal{O}(r^{-1}).
\end{nalign}

We can now fix $u$ and integrate the above in $v$ from $\mathcal I^+$ to obtain for $\pv\pho$:
\begin{align}
\begin{split}
r^2\pv\pho(u,v)=&\lim_{\mathcal{I}^+}r^2\pv\pho(u)-\frac{\lim_{\mathcal{I}^+ }r^2\pv(r^2\pv\pho)(u)}{r}\\
-&2M(\ccc-2M\cc)\frac{\log (v-u)-\log|u|}{r^2}+\mathcal{O}(r^{-2}),
\end{split}
\end{align}
where
\begin{align*}
\lim_{\mathcal{I}^+ }r^2\pv(r^2\pv\pho)(u)=2M\cc-2M\int_{-\infty}^u\lim_{\mathcal{I}^+}r\phi_1\dd u'.
\end{align*}
This concludes the proof of Theorem~\ref{thm:nl}.
\end{proof}

\subsection{Comments}\label{sec:nl:comments}
\subsubsection{The Newman--Penrose constant \texorpdfstring{$\ILn1[\phi]$}{I[phi]}}
It is instructive to also write down the asymptotics of the quantity related to the higher-order Newman--Penrose constant $\ILn1[\phi]$ (recall the definition \eqref{eq:NPl=1future}):
\begin{thm}\label{thm:NP}
Let $U_0$ be a sufficiently large negative number. Then, throughout $\mathcal{D}=(-\infty,U_0]\times [1,\infty)$, the outgoing derivative of the combination $r^2\pv\pho-Mr\phi_1$ satisfies, for fixed values of $u$, the following asymptotic expansion as $\mathcal{I}^+$ is approached:
\begin{align}
\begin{split}
r^2\pv(r^2\pv\pho-Mr\phi_1)=3M\cc+4M(\ccc-2M\cc)\frac{\log r-\log|u|}{r}+\mathcal{O}(r^{-1}).
\end{split}
\end{align}
In particular, $\ILn1[\phi](u)\equiv 3M\cc$ is conserved along $\mathcal{I}^+$.
\end{thm}
\subsubsection{The case \texorpdfstring{$\cc=0$}{Cin,1=0}: A logarithmically modified Price's law}\label{sec:nl:comments:logpricelaw}
Notice that if $\cc=0$, then $I_{\ell=1}^{\mathrm{future}}[\phi]=0$. However,  one can still define a conserved quantity along future null infinity in this case, namely
\begin{equation}
\ILlog1[\phi](u):=\lim_{v\to\infty}\frac{r^3}{\log r}\pv(r^2\pv(r\phi_1)-Mr\phi_1)(u,v),
\end{equation}
which, in our case, is given by $4M\ccc$.
In particular, by using similar methods to the ones from~\cite{II}, which combined the results of~\cite{I} and~\cite{AAG18b}, one can thus obtain that the late time asymptotics of the $\ell=1$-mode, if one smoothly extends the data to $\mathcal{H}^+$, have \textit{logarithmic corrections at leading order}. 
In particular, one can obtain that $r\phi_1(u,\infty)=Cu^{-3}\log u+\mathcal{O}(u^{-3})$ along $\mathcal{I}^+$, and that $\pv\phi_1(\infty,v)=C'v^{-5}\log v+\mathcal{O}(v^{-5})$ along the event horizon $\mathcal{H}^+$, where the constants $C$ and $C'$ can be  expressed explicitly in terms of $\ccc$. 

In order to show this, one needs to combine the results of the present paper with those of the recent~\cite{AAG21} and make modifications to~\cite{AAG21} similar to those in~\cite{II}, see the discussion of \S\ref{sec:intro:NP}.

\subsubsection{Discussion of the cancellations of Remark~\ref{rem:intro:cancel} and the case of general \texorpdfstring{$\ell$}{L}}\label{sec:4.4.3}
Recall the cancellations discussed for general $\ell$ in Remark~\ref{rem:intro:cancel}. 
Let us here give some intuition for them, restricting, of course, to the case $\ell=1$.

 Theorem~\ref{thm:nl} shows that, if $r\phi_1\sim 1/|u|$ initially, this translates to $r\phi|_{\mathcal I^+}\sim u^{-2}$ on null infinity. 
We found this "cancellation" somewhat tacitly, namely by transporting  decay for the commuted quantity $r^2\pu(r^2\pu(r\phi_1))$ along $\mathcal I^-$.  It is maybe easiest to explain why this approach produces no cancellations for $p\in(0,1)$ or $p\in(1,2)$:
If $p\in(1,2)$, then the estimate \eqref{eq:l:canc:1} becomes \textit{worse}, not better, since the initial data term of \eqref{eq:nl:proof1} now decays slower.
On the other hand, if $p\in(0,1)$, then \eqref{eq:l:canc:2} fails, as the limit $\lim r^2\pu(r\phi_1)$ diverges. In fact, this shows that the proof of the present section fails for $p<1$. 

There also is  a more direct way of understanding the cancellation for $p=1$:
In view of the estimate \eqref{eq:nl:pups}, we have that, schematically,
$$r\phi(u,v)=r\phi(u,1)+\int_1^v \pv(r\phi)\dd v'=r\phi(u,1)+\int_1^v \frac{-\cc}{r^2(u,v')}\dd v'  =\frac{\cc}{|u|}+\frac{\cc}{r}-\frac{\cc}{|u|}=\frac{\cc}{r},
$$
where we used that $r(u,1)\sim |u|$. From this point of view, it is clear that such cancellations only happen if $r\phi_1\sim 1/|u|^p$ for $p=1$. Our more systematic approach of \S\ref{sec:moreg:null}, in which we analyse general $\ell$-modes,  will understand the cancellations of Remark~\ref{rem:intro:cancel} in a generalised form of the above computation. 
Indeed, in \S\ref{sec:moreg:null}, we will avoid using the conservation law in the $v$-direction entirely, and instead only use the conservation law in the $u$-direction:
 Instead of propagating decay for $(r^2\pu)^{\ell+1}(r\phi_\ell)$ in $v$ and then integrating this $\ell+1$ times from $\mathcal I^-$ , we will directly obtain an estimate for $(r^2\pv)^{\ell+1}(r\phi_\ell)$ by integrating from $\mathcal I^-$ in $u$, and then integrate this estimate $\ell$ times from $v=1$, carefully analysing at each step the initial data contributions. 
 In particular, this approach will  also allow for slower decay in the initial data. 
 See already \S\ref{sec:mg:overview} for a more detailed overview of the approach for general $\ell$.

\newpage

\section{Boundary data on a timelike hypersurface \texorpdfstring{$\Gamma_R$}{Gamma-R}}\label{sec:timelikeconstantR}
Having obtained asymptotic estimates for solutions arising from polynomially decaying initial data on an ingoing null hypersurface in the previous section, we now want to obtain similar estimates for solutions arising from polynomially decaying \textit{boundary data on a timelike hypersurface} $\Gamma_R$. 
The main result of this section is the proof of Theorem~\ref{thm:intro:tl}. 

In contrast to the previous section, we here need to construct our solutions at the same time as we prove estimates on them.

We use the notation from \S\ref{sec:decompositionintoYlm}, that is, we write $\phi=\phi_{\ell=1}=\phi_1(u,v)\cdot Y_{1m}(\theta,\varphi)$.  
\subsection{Overview of the ideas and structure of the section}
Let us briefly recall the approach that we followed in our treatment of the $\ell=0$-mode in~\cite{I}:  
Given polynomially decaying boundary data on $\Gamma_R$, we first considered a sequence of compactly supported boundary data that would approach the original boundary data. 
This allowed us to use the method of continuity, i.e. bootstrap arguments. 
We then \textit{assumed} decay for $r\phi_0$, and improved it by essentially  integrating the wave equation \eqref{eq:l=0cons.law} first in $u$ and then in $v$ (from $\Gamma_R$) and exploiting $2M/R$ as a "small" parameter. 
In fact, we also showed that one can avoid exploiting smallness in $2M/R$ using a Gr\"onwall argument.

If we want to follow a similar approach for $\ell=1$, it is not sufficient to consider the uncommuted wave equation \eqref{eq:l=1waveequation} in view of its non-integrable $r^{-2}$-weight. 
Instead, it seems more appropriate to use the approximate conservation law \eqref{eq:l=1cons.law.in.u} and bootstrap decay on the combination 
$$\Phi:=r^2\pv(r\phi_1)-Mr\phi_1.$$ 

The first and main difficulty then becomes apparent: $\Phi|_{\Gamma_R}$ is not given by boundary data (we prescribe boundary data \textit{tangent} to $\Gamma$).
 One way of overcoming this difficulty  is to exploit certain cancellations in the wave equation; this however requires one to have knowledge on the $T$-derivative of $r\phi_1$.
  Alternatively, one can  estimate $r^2\pv\pho|_{\Gamma_R}$ using an energy estimate which only uses "a square root" of the bootstrapped decay of $r^2\pv\pho$. 
  We will make use of both of these approaches, the former for lower-order derivatives $r^2\pv T^n\pho$ (where we have room to make assumptions on $T^{n+1}\pho$), and the latter for the top-order derivative $r^2\pv T^N\pho$, $n<N$. 
  In fact,  using only the latter approach is sufficient, but we find it instructive to also include the former as it since it highlights the importance of commuting with $T$.   
   In the more systematic approach of the discussion of general $\ell$ in \S\ref{sec:general:timelike}, we will, however, exclusively use the latter approach.
  
  Equipped with a boundary estimate on $\Phi$, we can then hope to close the bootstrap argument by simply integrating \eqref{eq:l=1cons.law.in.u} first in $u$ and then in $v$, and exploiting $2M/R$ as a small parameter.
   In fact, as in the $\ell=0$-case, one can avoid this smallness assumption. 
   The only additional subtlety here is that, in order to estimate the RHS of \eqref{eq:l=1cons.law.in.u}, we need to control $r\phi_1$ and $\pv(r\phi_1)$, which is not directly provided by a bootstrap assumption on the combination $\Phi$. We will deal with this by estimating $r\phi_1$ against the integral over $\pv(r\phi_1)$ from $\Gamma_R$, and either just exploiting smallness in $2M/R$ or using a more elaborate Gr\"onwall argument.


%

\paragraph*{Structure}
We first state our initial boundary data assumptions for $\phi_1$, as well as the main theorem, in \S \ref{sec:subsec:tl:setup}. 
Then, in order to gain access to the method of continuity, we smoothly cut-off the boundary data in \S \ref{sec:tl:cutoffdata}. 
These will lead to finite solutions $\phi_1^{(k)}$ in the sense of Proposition~\ref{prop:existence:mixedboundary}.
Using bootstrap methods as outlined above, we can then estimate $r^2\pv T^n(r\phi_1^{(k)})$ and $T^n(r\phi_1^{(k)})$ in \S \ref{sec:tl:pvT}. 

In order to later show that Theorem~\ref{thm:nl} can be applied (i.e.\ to show that the limit $\lim_{u\to-\infty}(r^2\pu)^2(r\phi_1)(u,v)$ exists), we will also need to show some auxiliary estimates on the differences $r^2\pv T^n(r\phi_1^{(k)}-|u|T(r\phi_1^{(k)})$. This is done in \S \ref{sec:tl:pvT-T2}. 

In \S \ref{sec:tl:limit}, we finally show that the finite solutions $\phi_1^{(k)}$ tend to a limiting solution and show that Theorem~\ref{thm:nl} can be applied to it, thus proving Theorem~\ref{thm:intro:tl}. 
We  make some closing comments in \S \ref{sec:tl:comments}.

\subsection{The setup}\label{sec:tl:setup}
\subsubsection{The initial/boundary data and the main theorem (Theorem~\ref{thm:tl})}\label{sec:subsec:tl:setup}
Throughout the rest of this section, we shall assume that $R>2M$ is a constant. 
In particular, $T=\pu+\pv$ will be tangent to $\Gamma_R$. 
We then prescribe smooth boundary data $\hat{\phi}_1$ on $\Gamma_R=\mathcal{M}_M\cap\{v=v_R(u)\}$ that satisfy, for $u\leq U_0<0$ and $|U_0|$ sufficiently large, the upper bounds
\newcommand{\cin}{C_{\mathrm{in}}^\Gamma}
\newcommand{\ce}{C_{\mathrm{in},\epsilon}^\Gamma}
\begin{align}\label{eq:tl:boundarydata1}
\left|T^n(r\hat{\phi}_1)\right|&\leq \frac{n! \cin}{R|u|^{n+1}},&&n=0,1,\dots,N+1,\\
\left|T^n\left(r\hat{\phi}_1-|u|T(r\hat{\phi}_1)\right)\right|&\leq \frac{\ce}{R|u|^{n+1+\epsilon}},&& n=0,\dots,N'+1\label{eq:tl:boundarydata2}
\end{align}
for some positive constants $\cin$, $\ce$, $\epsilon\in(0,1)$ and $N,N'\geq 0$  integers, and which also satisfy the following lower bound:
\begin{equation}\label{eq:tl:thm:lowerbound}
\left|T(r\hat{\phi}_1)\right|\geq \frac{\cin}{2R|u|^{2}}>0.
\end{equation}
Moreover, we demand, in a limiting sense, that, for all $v$,
\begin{equation}\label{eq:tl:noincomingradiation}
\lim_{u\to-\infty}\pv^n(r\phi_1)(u,v)=0, \quad n=1,\dots, N+1.
\end{equation} 
Then the main result of this section is
\begin{thm}\label{thm:tl}
Let $R>2M$ be a constant. Then there exists a unique solution $\phi_1$ to eq.\ \eqref{eq:l=1waveequation} in $\mathcal{D}_{\Gamma_R}:=\mathcal{M}\cap\{v\geq v_R(u)\}$ that restricts correctly to $\hat{\phi}_1$ on $\Gamma_R$, $\phi_1|_{\Gamma_R}=\hat{\phi}_1$, and that satisfies \eqref{eq:tl:noincomingradiation}.
Moreover, if $U_0$ is a sufficiently large negative number, then there exists a constant $C=C(2M/R,\cin)$, depending only on data, such that $\phi_1$ obeys the following bounds throughout $\mathcal{D}_{\Gamma_R}\cap\{u\leq U_0\}$:
\begin{align}\label{eq:thm5.5}
\left|r^2\pv T^n(r\phi_1)(u,v)\right|&\leq \frac{C}{|u|^{n+1}},&&n=0,\dots,N,\\
\left|T^n(r\phi_1)(u,v)\right|&\leq \frac{C}{|u|^{n+1}}\max\left(r^{-1},|u|^{-1}\right),&&n=0,\dots,N-1.
\end{align}
Finally, if  $N\geq 4$ and $N'\geq 2$, then we have, along any ingoing null hypersurface $\mathcal{C}_v$, that
\begin{align}
r^2\pu(r\phi_1)(u,v)&=\mathcal{O}(r^{-1}),\\
r^2\pu(r^2\pu\pho)(u,v)&=\tilde{C}+\mathcal{O}(r^{-1}+|u|^{-\epsilon}),\label{eq:tl:thm:limit}
\end{align}
where $\tilde{C}$ is a constant that is non-vanishing if $R/2M$ is sufficiently large. 
In particular, $\phi_1$ satisfies the assumptions of Theorem~\ref{thm:nl} with $\cc=0$, $\ccc=\tilde{C}$ and $\epsilon=\eta$.
\end{thm}
\begin{rem}
Let us already draw the reader's attention to the fact that the data above lead to solutions with $\cc=0$ (cf.\ \eqref{eq:nl:ass1}). In view of the comments in \S\ref{sec:nl:comments:logpricelaw}, this suggests that the data considered here lead to a logarithmically modified Price's law near $i^+$. 
\end{rem}
\begin{rem}
Instead of considering data with $\hat{\phi}_1\sim|u|^{-1}$, we can also consider data with $\hat{\phi}_1\sim|u|^{-p}$ for $p>0$ and derive a similar result with some obvious modifications.
\end{rem}
\begin{rem}
It may be instructive for the reader to keep the following solution to \eqref{eq:l=1waveequation} in the case $M=0$ in mind:
\begin{equation}\label{eq:tl:examplesolution}
\pu\pv\left(\frac{1}{2|u|^2}+\frac{1}{|u|r}\right)=-\frac{2}{r^2}\left(\frac{1}{2|u|^2}+\frac{1}{|u|r}\right).
\end{equation}
\end{rem}

\subsubsection{Cutting of the data and replacing \texorpdfstring{$\mathcal{I}^-$}{I-} with \texorpdfstring{$\mathcal{C}_{u=-k}$}{C(u=-k)}}\label{sec:tl:cutoffdata}
As mentioned before, in order to appeal to bootstrap arguments, we need to work in compact regions. 
We therefore need to cut the boundary data off and then recover the original data using a limiting argument.  Let $(\chi_k(u))_{k\in\mathbb{N}}$ be a sequence of positive smooth cut-off functions such that
    \begin{equation*}
        \chi_k=\begin{cases}
        1,& u\geq- k+1,\\
        0,&u\leq -k,
        \end{cases}
    \end{equation*}
and cut off the highest-order derivative: $\chi_k \cdot T^{N+1}\hat{\phi}$.
We then have
\begin{equation*}
\int_{-\infty}^u \chi_k T^{N+1}\hat{\phi}_1=\chi_k T^{N}\hat{\phi}_1-\int_{-\infty}^{u}(T\chi_k)( T^{N}\hat{\phi}_1 )=\chi_k T^{N}\hat{\phi}_1+\theta_k\cdot\mathcal{O}(k^{-N-1}),
\end{equation*}
where $\theta_k$ equals 1 on $\{u\geq -k\}$ and 0 elsewhere. Similarly, we obtain inductively that
\begin{equation*}
\underbrace{\int\dots\int}_{n\, \mathrm{ times}} \chi_k T^{N+1}\hat{\phi}_1=\chi_k T^{N+1-n}\hat{\phi}_1+\theta_k\cdot\mathcal{O}(k^{-N-2+n}).
\end{equation*}
In particular, if we denote $\overbrace{\int\dots\int}^{N+1\text{ times}} \chi_k T^{N+1}\hat{\phi}_1$ as $\hat{\phi}_1^{(k)}$, then the bounds  \eqref{eq:tl:boundarydata1}, \eqref{eq:tl:boundarydata2} imply, for sufficiently large negative values of $u$ and for some constant $C'_{\mathrm{in}}=C'_{\mathrm{in}}(N,N')$:
\begin{align}
\left|T^n\left(r\hat{\phi}^{(k)}_1\right)\right|&\leq \frac{n! \cin}{R|u|^{n+1}},&&n=0,1,\dots,N+1,\label{eq:tl:boundarydata1cutoff}\\
\left|T^n\left(r\hat{\phi}^{(k)}_1-|u|T\left(r\hat{\phi}^{(k)}_1\right)\right)\right|&\leq \frac{\ce}{R|u|^{n+1+\epsilon}}+C'_{\mathrm{in}}\theta_k\cdot \frac{\cin}{Rk^{n+1}},&& n=0,1,\dots,N'+1\label{eq:tl:boundarydata2cutoff}
\end{align}
Notice that, in the second line above, we lose some decay due to the $\theta_k$-term arising from the cut-off. Since we will take the limit $k\to\infty$ in the end, this only poses a minor difficulty.

Throughout the next two sections (\S \ref{sec:tl:pvT} and \S \ref{sec:tl:pvT-T2}), we shall assume initial/boundary data satisfying the estimates \eqref{eq:tl:boundarydata1cutoff} and \eqref{eq:tl:boundarydata2cutoff} and moreover satisfying
\begin{equation}\label{eq:tl:noincomingfinite}
\phi_1(u=-k,v)=0
\end{equation}
for all $v\geq v_R(-k)$. We shall denote the unique solutions to these initial/boundary value problems as $\phi_1^{(k)}$.
For the next two sections, we shall drop the superscript $(k)$, only to reinstate it in \S\ref{sec:tl:limit}, where we will show that the solutions $\phi_1^{(k)}$ tend towards a limiting solution.
\subsection{Estimates for \texorpdfstring{$\pv T^n (r\phi_1)$}{d/dv(Tn(r phi1))} and \texorpdfstring{$T^n(r\phi_1)$}{Tn(r phi1)}}\label{sec:tl:pvT}
Let $U_0$ be a sufficiently large negative number, and let $\hat{\phi}_1$ be smooth data on $\Gamma_R$, supported on $\Gamma_R\cap\{-k< u\}$ and satisfying \eqref{eq:tl:boundarydata1cutoff}. 
By Prop.~\ref{prop:existence:mixedboundary}, there exists a unique smooth solution $\phi_1$ throughout $\mathcal{D}_{\Gamma_R}\cap\{-k\leq u\}$  such that $\phi_1(-k,v)=0$ for all $v\geq v_R(-k)$ and such that $\phi_1|_{\Gamma_R}=\hat{\phi}_1$. 
We will now derive the following uniform-in-$k$ estimates on this solution $\phi_1$:
\begin{prop}\label{prop:tl:pv}
Let $\phi_1$ be the solution as described above, and let $N\geq 1$. Then, if  $|U_0|$ is sufficiently large, there exists a constant $C=C(2M/R,\cin)$ (in particular, this constant does not depend on $k$), which can be chosen to be independent of $R/2M$ for large enough $R/2M$,  such that the following estimates hold throughout $\mathcal{D}_{\Gamma_R}\cap\{-k\leq u\leq U_0\}$:
\begin{align}
\left|r^2\pv T^n (r\phi_1)(u,v)\right|&\leq \frac{C}{|u|^{n+1}},&&n=0,1,\dots,N,\\
\left|T^n (r\phi_1)(u,v)\right|&\leq \frac{C}{|u|^{n+1}}\max\left(r^{-1},|u|^{-1}\right),&&n=0,1,\dots,N-1.
\end{align}
\end{prop}
\begin{proof}
The proof is divided into the sections \S \ref{sec:tl:BS}--\S \ref{sec:tl:removelargeness}. 
 In \S \ref{sec:tl:BS}--\S \ref{sec:subsubsec:timelikej=N}, we present a bootstrap argument and exploit $2M/R$ as a small parameter to improve the bootstrap assumptions. 
 An overview over this bootstrap argument will be given in \S\ref{sec:tl:BS}.
 
We will then explain how to lift the smallness assumption on $2M/R$ by partially replacing the bootstrap argument with a Gr\"onwall-type argument in \S \ref{sec:tl:removelargeness}. 
\subsubsection{The bootstrap assumptions} \label{sec:tl:BS}
\newcommand{\cbs}{C_{\mathrm{BS}}}
\newcommand{\cbsp}{C_{\mathrm{BS},\phi}}
Let $\{\cbs^{(n)}, n=1,\dots,N\}$ and $\{\cbsp^{(m)},m=0,\dots,N-1\}$ be two sets of sufficiently large positive constants. We shall make the following bootstrap assumptions on $\phi_1$:
\begin{align}\tag{BS(n)}\label{eq:tl:BSn}
\left|r^2\pv T^n (r\phi_1)(u,v)\right|\leq \frac{\cbs^{(n)}}{|u|^{n+1}}
\end{align}
for $n=1,\dots,N$, and
\begin{align}\tag{BS'(m)}\label{eq:tl:BS1n}
\left|T^m (r\phi_1)(u,v)\right|\leq \frac{\cbsp^{(m)}}{|u|^{m+1}}\max\left(r^{-1},|u|^{-1}\right)
\end{align}
for $m=0,\dots,N-1$.

We now define $\Delta$ to be the subset of all $(u,v)\in X:= \{(u,v)|-k< u\leq U_0, v_R(u)< v\}$ such that, for all $(u',v')\in X$ with $u'\leq u$ and $v'\leq v$, \eqref{eq:tl:BSn} and \eqref{eq:tl:BS1n} hold for all $n=1,\dots,N$, $m=0,\dots,N-1$, respectively.

By compactness and continuity, $\Delta$ is non-empty if the constants $\cbs^{(n)},\cbsp^{(m)}$ are chosen sufficiently large.
 Moreover, $\Delta$ is trivially closed in $X$.
  We shall show that $\Delta$ is also open by improving each of the bootstrap assumptions within $\Delta$.

We shall first improve the bootstrap assumptions for the lower-order $T$-derivatives ($n\leq N-2$) \textit{by explicitly exploiting the precise behaviour for higher $T$-derivatives}  in \S \ref{sec5.3.2} in order for the reader to get a clear intuition for the origin of the assumed rates. 
In \S \ref{sec5.3.3}, we will then improve the bootstrap assumption away from the top-order derivative ($n\leq N-1$), where we no longer have the sharp decay for $T^N(r\phi_1)$ available. 
Finally, in \S \ref{sec:subsubsec:timelikej=N}, we will improve the bootstrap assumptions for the top-order derivatives.

Since the approach of \S \ref{sec:subsubsec:timelikej=N} applies to derivatives of any order,  the reader can in principle skip \S \ref{sec5.3.2}--\S \ref{sec5.3.3}, which are  included for pedagogical reasons, and go directly to \S \ref{sec:subsubsec:timelikej=N}.
 In fact, \S \ref{sec:subsubsec:timelikej=N} only requires the bootstrap assumptions \eqref{eq:tl:BSn} (and not \eqref{eq:tl:BS1n}). 
In particular, when going through \S\ref{sec5.3.2}--\S\ref{sec5.3.3}, the reader can focus on the arguments without having to pay close attention to the bootstrap constants.

\subsubsection{Closing away from the top-order derivatives \texorpdfstring{$j\leq N-2$}{j<N-1}}\label{sec5.3.2}
The idea is to exploit the fact that, for $M=0$, $\phi_1=1/r^2$ is a stationary solution. 
 In particular, we expect $\pv(r^2\phi_1)$ to have some cancellations (see \eqref{2}), and $r^2\pv(r\phi_1)$ to remain approximately conserved in $u$ and $v$ (see \eqref{3}). (We remind the reader of the example solution \eqref{eq:tl:examplesolution}.)
\begin{prop}\label{prop:tlbs1}
Let $0\leq j\leq N-2$. Then, for sufficiently large values of $R/2M$ and $|U_0|$, and if the ratios $\cbsp^{(j)}/\cbsp^{(j+1)}$, $\cbs^{(j)}/\cbsp^{(j+1)}$  are chosen large enough, we have throughout $\Delta$ that, in fact,
  \begin{align}
  \left|r^2\pv T^j (r\phi_1)(u,v)\right|&\leq\frac12 \frac{\cbs^{(j)}}{|u|^{j+1}},\\
  \left|T^j (r\phi_1)(u,v)\right|&\leq\frac12 \frac{\cbsp^{(j)}}{|u|^{j+1}}\max\left(r^{-1},|u|^{-1}\right).
  \end{align}
\end{prop}
\begin{proof}
Fix $j\leq N-2$ and assume \eqref{eq:tl:BS1n} for $m=j,j+1$.
Motivated by the comment above, we compute
\begin{equation}\label{2}
\pu(r^{-2}\pv(r^2\phi_1))=\frac{DT(r\phi_1)}{r^2}-\frac{8MD}{r^4}r\phi_1.
\end{equation}
Commuting with $T^j$, plugging in the bootstrap assumptions, and integrating \eqref{2} from $u=-k$, we find (recall that $\pu r=-D$):
\begin{align}\label{proof5.2x}
\left|r^{-2}T^j\pv(r^2\phi_1)(u,v)\right|\leq\int_{r_v(-k)}^{r_v(u)} \frac{\cbsp^{(j+1)}}{r^3|u|^{j+2}}+\frac{8M\cbsp^{(j)}}{r^5|u|^{j+1}}\dd r\leq  \frac{\cbsp^{(j+1)}}{2r^2|u|^{j+2}}+\frac{2M\cbsp^{(j)}}{r^4|u|^{j+1}}.
\end{align}
Here, we denoted $r_v(u)$ as the unique $r$ such that $r^*(r)=v-u$. 
Now, we similarly compute 
\begin{equation}\label{3}
	\pu(r^2\pv(r\p))=-2D\pv(r^2\p)-6DM\frac{r\p}{r}.
\end{equation}
Commuting again with $T^j$, plugging in the bound \eqref{proof5.2x} for $T^j\pv(r^2\phi_1)$ from above, and integrating \eqref{3} in $u$, we then find:
\begin{align}\label{proof5.2xx}
\begin{split}
\left|r^2\pv T^j\pho(u,v)\right|&\leq \int_{r_v(-k)}^{r_v(u)}\frac{\cbsp^{(j+1)}}{|u|^{j+2}}\dd u+\int_{r_v(-k)}^{r_v(u)}\frac{4M\cbsp^{(j)}}{r^2|u|^{j+1}}+\frac{6M\cbsp^{(j)}}{r^2|u|^{j+1}}\dd r\\
&\leq\frac{\cbsp^{(j+1)}}{(j+1)|u|^{j+1}}+\frac{10M\cbsp^{(j)}}{r|u|^{j+1}}.
\end{split}
\end{align}
For large enough $R$ and $\cbs^{(j)}/\cbsp^{(j+1)}$, this proves the first part of the proposition.

Moreover, inserting \eqref{proof5.2xx} back into \eqref{proof5.2x} and writing $\pv(r^2\p)=r\pv\pho+Dr\p$, we obtain
\begin{align}
\begin{split}
D\left|T^j\pho\right|&\leq \left|T^j\pv(r^2\p)\right|+\frac{1}{r}\left|r^2T^j\pv\pho\right|\\
			&\leq \frac{\cbsp^{(j+1)}}{2|u|^{j+2}}+\frac{2M\cbsp^{(j)}}{r^2|u|^{j+1}}+\frac{\cbsp^{(j+1)}}{(j+1)r|u|^{j+1}}+\frac{10M\cbsp^{(j)}}{r^2|u|^{j+1}}.
\end{split}
\end{align}
This proves the second part of the proposition for large enough $R$ and $\cbsp^{(j)}/\cbsp^{(j+1)}$.
\end{proof}
\subsubsection{Closing away from the top-order derivatives \texorpdfstring{$j= N-1$}{j=N-1}}\label{sec5.3.3}
In the previous proof, we crucially needed the sharp decay of $T^{j+1}\pho$, which we no longer have access to if $j+1=N$. We therefore proceed differently now. We will use the approximate conservation law \eqref{eq:l=1cons.law.in.u}. In fact, since we still have sharp decay for $T^j\pho$, it will suffice to consider (the $T$-commuted)
\begin{equation}\label{4}
\pu(r^{-2}\pv(r^2\pv(r\phi_1)))=-10 MD\frac{r^2\pv(r\phi_1)}{r^5}-2MD\frac{r\phi_1}{r^4}\left(1+\frac{4M}{r}\right),
\end{equation}
since, as long as we have the extra $r$-decay of $T^j\pho$, the bad $r^{-4}$-weight multiplying $T^j\pho$ poses no problem.
\begin{prop}\label{prop:tlbs2}
Let $0\leq j\leq N-1$. Then, for sufficiently large values of $R/2M$ and $|U_0|$, and if $\cbs^{(j)}$ and $\cbsp^{(j)}$ are chosen large enough, we have throughout $\Delta$ that, in fact,
  \begin{align}
 \left |r^2\pv T^j (r\phi_1)(u,v)\right|&\leq\frac12 \frac{\cbs^{(j)}}{|u|^{j+1}},\\
 \left |T^j (r\phi_1)(u,v)\right|&\leq\frac12 \frac{\cbsp^{(j)}}{r|u|^{j+1}}\max\left(r^{-1},|u|^{-1}\right).
  \end{align}
\end{prop}
\begin{proof}
Fix $j\leq N-1$ and assume \eqref{eq:tl:BSn} for $n=j,j+1$ and \eqref{eq:tl:BS1n} for $m=j$.
The idea is to integrate \eqref{4} twice, first from $u=-k$ and then from $\Gamma_R$.
In doing so, we will pick up the boundary term $r^2T^j\pho|_{\Gamma_R}$, which is not given by data. 
 We will therefore estimate this boundary term by using the ($T$-commuted) eq.\ \eqref{2}:
First, note that, by integrating the bound  \eqref{eq:tl:BSn} for $n=j+1$ from $\Gamma_R$, we have 
\begin{equation}
\left|T^{j+1}\pho(u,v)\right|\leq \frac{(j+1)!\cin }{R|u|^{j+1}}+\frac{\cbs^{(j+1)}}{R|u|^{j+1}}.
\end{equation}
Hence, by integrating equation \eqref{2} from $u=-k$, we obtain
\begin{align*}
\begin{split}
\left|r^{-2}T^j\pv(r^2\phi_1)(u,v)\right|\leq\int_{r_v(-k)}^{r_v(u)} \frac{(j+1)!\cin+\cbs^{(j+1)}}{r^2R|u|^{j+2}}+\frac{8M\cbsp^{(j)}}{r^5|u|^{j+1}}\dd r,
\end{split}
\end{align*}
and, consequentially,
\begin{align*}
\left|r^2T^j\pv\pho+Dr\cdot T^j(r\phi_1)\right|\leq r^2\frac{(j+1)!\cin+\cbs^{(j+1)}}{R|u|^{j+2}}+\frac{2M\cbsp^{(j)}}{r|u|^{j+1}}.
\end{align*}
Evaluating this bound on $\Gamma_R$ and applying the triangle inequality gives
\begin{equation}\label{proof5.3x}
\left|r^2T^j\pv\pho|_{\Gamma_R}\right|\leq \left(1-\frac{2M}{R}\right)\frac{j!\cin}{|u|^{j+1}}+\frac{2M\cbsp^{(j)}}{R|u|^{j+1}}+R\frac{(j+1)!\cin+\cbs^{(j+1)}}{|u|^{j+2}}.
\end{equation}
Notice that, for sufficiently large $|U_0|$, the last term becomes subleading. Moreover, if $2M/R$ is suitably small, the first term, in fact, dominates.

 Equipped with this estimate for the boundary term, we can now integrate (the $T^j$-commuted) approximate conservation law \eqref{4}, first from $u=-k$:
\begin{align*}
&r^{-2}\pv(r^2\pv T^j\pho)(u,v)\\
&=\int_{-k}^u -10 MD\frac{r^2\pv T^j(r\phi_1)}{r^5}-2MD\frac{T^j(r\phi_1)}{r^4}\left(1+\frac{4M}{r}\right)\dd u' \\
&\leq \int_{-r_v(k)}^{r_v(u)} \frac{10M\cbs^{(j)}}{r^5|u|^{j+1}}+\frac{2M\cbsp^{(j)}}{r^5|u|^{j+1}}\left(1+\frac{4M}{r}\right)\dd r.
\end{align*}
We thus obtain:
\begin{equation}\label{proof5.3xx}
\left|\pv(r^2\pv T^j\pho)\right| \leq \frac{10M\cbs^{(j)}}{4r^2|u|^{j+1}}+\frac{M\cbsp^{(j)}}{2r^2|u|^{j+1}}\left(1+\frac{16M}{5R}\right).
\end{equation}
Finally, integrating \eqref{proof5.3xx} from $\Gamma_R$, and estimating the boundary term via \eqref{proof5.3x}, we obtain
\begin{align}\label{proof5.3xxx}
\begin{split}
\left|r^2\pv T^j\pho\right|\leq& D(R)\frac{j!\cin}{|u|^{j+1}}+\frac{2M\cbsp^{(j)}}{R|u|^{j+1}}+R\frac{(j+1)!\cin+\cbs^{(j+1)}}{|u|^{j+2}}\\
+&D^{-1}(R)\left(\frac{10M\cbs^{(j)}}{4R|u|^{j+1}}+\frac{M\cbsp^{(j)}}{2R|u|^{j+1}}\left(1+\frac{16M}{5R}\right)\right).
\end{split}
\end{align}
(The factor $D^{-1}$ comes from substituting $\mathrm{d}v$ with $\mathrm{d}r$ in the integral.) If $U_0$ and $R/2M$ are sufficiently large, and if $\cbs^{(j)}$ is chosen suitably large relative to $j!\cin$, then  the RHS can be shown to be smaller than $\frac{\cbs^{(j)}}{2|u|^{j+1}}$. We thus recover the first statement of the proposition. 

In order to show the second statement, we exploit the fact that we have also obtained an estimate on $r^2\pv(r^2 \pv T^j\pho)$ and use the following identity (which follows directly from the wave equation \eqref{eq:l=1waveequation}):
\begin{equation}\label{eq:waveequationtoestimatephi}
-\frac{2D}{r^2}T^j\pho\left(1+\frac{M}{r}\right)=T^j\pv\pu (r\phi_1)=-\pv^2T^j\pho+\pv T^{j+1}\pho.
\end{equation} 
The last term of the equation above is controlled by the bootstrap assumption \eqref{eq:tl:BSn} for $n=j+1$. For the other term, we can write:
\begin{align*}
\pv^2T^j\pho=\frac{1}{r^2}\pv\left(r^2\pv T^j\pho\right)-\frac{2D}{r}\pv T^{j}\pho.
\end{align*}
We therefore can estimate $T^j\pho$ as follows:
\begin{equation}\label{eq:nl:nl;nl}
2D\left(1+\frac{M}{r}\right)\left|T^j(r\phi_1)\right|\leq \left|\pv\left(r^2\pv T^j\pho\right)\right|+\frac{2D}{r}\left|r^2\pv T^{j}\pho\right|+\left|r^2\pv T^{j+1}\pho\right|.
\end{equation}
Finally, plugging in the estimates \eqref{proof5.3xx}, \eqref{proof5.3xxx}, as well as the assumption \eqref{eq:tl:BSn} for $n=j+1$, into the estimate \eqref{eq:nl:nl;nl} shows that if $U_0$ and $R/2M$ are sufficiently large, and if $\cbsp^{(j)}$ is chosen suitably large relative to $j!\cin/2$, then the RHS is smaller than $\frac{\cbsp^{(j)}}{2r|u|^{j+1}}$, thus proving the proposition.\footnote{Notice that the last term on the RHS of \eqref{eq:nl:nl;nl} is the reason why we need the $\max\left(r^{-1},|u|^{-1}\right)$ in \eqref{eq:tl:BS1n} rather than just $r^{-1}$. This only becomes relevant near $\mathcal{I}^+$.}
\end{proof}

\subsubsection{Closing the top-order derivative \texorpdfstring{$j=N$}{j=N}}\label{sec:subsubsec:timelikej=N}
In the previous proof, we used the sharp $u$-decay of $T^{j+1}\pho$ (see \eqref{proof5.3x}), combined with equation \eqref{2}, to estimate the boundary term $r^2\pv T^j\pho|_{\Gamma_R}$ on $\Gamma_R$. 
At the highest order in derivatives, we can no longer do this. 
Instead, we will estimate the boundary term using an energy estimate. This energy estimate will be wasteful in terms of $r$-decay, but sharp in terms of $u$-decay and, therefore, useful on $\Gamma_R$. Moreover, it only requires a "square root" of the bootstrap estimate on $r^2\pv T^j\pho$ and, thus, allows for improvement.

Another difference to the previous section will be that, since we can no longer assume the sharp decay of $T^{j+1}\pho$, we will have to work with the approximate conservation law \eqref{eq:l=1cons.law.in.u} instead of \eqref{4}. (Recall that the former has a better $r$-weight multiplying $T^j(r\phi_1)$.) This will give us an estimate on $T^j\Phi=r^2\pv T^j\pho-MT^j\pho$. As mentioned in the introduction to this section, we will simply exploit the largeness in $R/2M$ to estimate $r^2\pv T^j\pho$ in terms of $T^j\Phi$.
\begin{prop}\label{prop:tlbs3}
Let $0\leq j\leq N$. Then, for sufficiently large values of $R/2M$ and $|U_0|$, and if $\cbs^{(j)}$ is chosen large enough, we have throughout $\Delta$ that, in fact,
  \begin{align}
 \left |r^2\pv T^j (r\phi_1)(u,v)\right|&\leq\frac12 \frac{\cbs^{(j)}}{|u|^{j+1}}.
  \end{align}
\end{prop}
\begin{proof} Fix $j\leq N$ and assume \eqref{eq:tl:BSn} for $n=j$.
 
Recall the definition \eqref{eq:def:J}. Since $T$ is Killing and $T^j\phi_{\ell=1}$ solves the wave equation, we have
\begin{equation}\label{check2}
\div J^T[T^j\phi_{\ell=1}]=0.
\end{equation}
We want to apply the divergence theorem to this identity. We recall the notation $\phi_{\ell=1}=\phi_1\cdot Y_{1m}$, and denote by $\slashed{\nabla}_{\mathbb{S}^2}$ the covariant derivative on the unit sphere. We compute
\begin{align*}
J^T[T^j\phi_{\ell=1}]\cdot \pu&=\mathbf{T}[T^j\phi_{\ell=1}](T,\pu)			=	(\pu T^j\phi_{\ell=1})^2+\frac{D}{r^2}\left|\slashed{\nabla}_{\mathbb{S}^2} T^j\phi_{\ell=1}\right|^2,		\\
J^T[T^j\phi_{\ell=1}]\cdot \pv&=\mathbf{T}[T^j\phi_{\ell=1}](T,\pv)			=		(\pv T^j\phi_{\ell=1})^2+\frac{D}{r^2}\left|\slashed{\nabla}_{\mathbb{S}^2} T^j \phi_{\ell=1}\right|^2	,	\\
J^T[T^j\phi_{\ell=1}]\cdot (\pu-\pv)&=\mathbf{T}[T^j\phi_{\ell=1}](T,\pu-\pv)	=	T^{j+1}\phi_{\ell=1}\cdot (\pu-\pv)T^j \phi_{\ell=1}	.
\end{align*}
Let now $(u,v)\in\Delta$. Then, applying the divergence theorem as in \eqref{eq:divergencetheorem} to \eqref{check2}, we obtain
\begin{align}
\begin{split}
&\int_{\mathcal{C}_{v}\cap\{-k\leq u'\leq u\}} r^2\dd u'\,  \dd\Omega\, J^T[T^j\phi_{\ell=1}]\cdot \pu \\
\leq &\int_{\Gamma_R\cap\{-k\leq u'\leq u\}}  r^2 \dd (u'+v')\dd \Omega	\,J^T[T^j\phi_{\ell=1}]\cdot(\pu-\pv).
\end{split}
\end{align}
Doing the integrals over the sphere, using that $u+v=2u+r^*(R)$ along $\Gamma_R$, and plugging in the expressions for the fluxes from above, we obtain
\begin{align}\label{random1}
\begin{split}
&\int_{-k}^u\left( r^2(\pu T^j\phi_{1})^2+D|T^j\phi_{1}|^2\right)(u',v)\dd u'	 \\
\leq&\int_{\Gamma_R\cap\{-k\leq u'\leq u\}}  \left(2r^2 T^{j+1}\phi_{1}\cdot (T-2\pv)T^j \phi_{1}\right)\left(u',u'+r^*(R)\right)	 \dd u'.
\end{split}
\end{align}
Observe that we can estimate the right-hand side of \eqref{random1} by using the boundary data assumption \eqref{eq:tl:boundarydata1cutoff} for $T^{j+1}\phi_1$ and the bootstrap assumption \eqref{eq:tl:BSn} with $n=j$ for $\pv T^j\phi_1$. 

On the other hand, the left-hand side of \eqref{random1} controls $\sqrt{r}T^j\phi$, as can be seen by applying first the fundamental theorem of calculus and then the Cauchy--Schwarz inequality:
\begin{align*}
\left|T^j\phi_1(u,v)\right|\leq \left(\int_{-k}^u \frac{1}{r^2(u',v)}\dd u'\right)^\frac12\left(\int_{-k}^u r^2\left(\pu T^j\phi_1\right)^2(u',v)\dd u'\right)^\frac12.
\end{align*}
Applying the energy identity \eqref{random1} to the above estimate gives
\begin{align*}
Dr(T^j\phi_1)^2&\leq \int_{\Gamma_R} 2|T^{j+1}(r\phi_1)|\left(|T^{j+1}(r\phi_1)|+2|\pv T^j(r\phi_1)|+\frac{2D}{r} |T^j(r\phi_1)|\right)\dd u'\\
&\leq \int_{\Gamma_R} \frac{2(j+1)!\cin}{R|u'|^{j+2}}\left( \frac{2(j+1)!\cin}{R|u'|^{j+2}}+\frac{2\cbs^{(j)}}{R^2|u'|^{j+1}}+\frac{2D}{R}  \frac{2j!\cin}{R|u'|^{j+1}}\right)\dd u'\\
&\leq \frac{1}{R^3|u|^{2j+2}}\underbrace{\left(4D(j!)^2(\cin)^2+2j!\cin\cbs^{(j)}\right)}_{:=A}+\mathcal{O}\left(\frac{1}{|u|^{2j+3}R^2}\right).
\end{align*}
Plugging this bound\footnote{We from now on ignore the $\mathcal{O}(|u|^{-2j-3})$-term, for it can be easily absorbed into the slightly wasteful estimates we make at each step by choosing $U_0$ large enough.} into the wave equation \eqref{eq:l=1waveequation} and integrating \eqref{eq:l=1waveequation} from $u=-k$ results in the following bound on the boundary term $\pv T^j(r\phi_1)|_{\Gamma_R}$:
\begin{equation}\label{proof5.4x}
\left|\pv T^j(r\phi_1)|_{\Gamma_R}\right|\leq \int \frac{2\sqrt{AD}}{R^{\frac32}|u|^{j+1}}\frac{1}{r^{3/2}}\left(1+\frac{M}{r}\right) \dd u\leq \frac{4\sqrt{A}}{R^2|u|^{j+1}}\frac{1+\frac{M}{3R}}{\sqrt{1-\frac{2M}{R}}}.
\end{equation}
\textit{In fact, we see that the estimate on the boundary term closes by itself!}

Having obtained a bound on the boundary term, we can proceed as in the previous proof. We  insert the bootstrap estimate \eqref{eq:tl:BSn} for $n=j$ and the estimate 
\begin{equation}\label{proof5.4y}
 |T^j(r\phi_1)|\leq \frac{j!\cin}{R|u|^{j+1}}+\frac{\cbs^{(j)}}{R|u|^{j+1}}
\end{equation}
implied by it into the approximate conservation law \eqref{eq:l=1cons.law.in.u} to find (recall $\Phi=r^2\pv\pho-Mr\phi_1$):
\begin{align}\label{proof5.4xx}
\begin{split}
|r^{-2}\pv T^j\Phi(u,v)|&\leq \int_{-k}^{u} \frac{6M^2D}{r^5}\frac{j!\cin+\cbs^{(j)}}{R|u'|^{j+1}}+\frac{12MD}{r^5}\frac{\cbs^{(j)}}{|u'|^{j+1}}\dd u'\\
&\leq \frac{3M^2}{2r^4}\frac{j!\cin+\cbs^{(j)}}{R|u|^{j+1}}+\frac{3M\cbs^{(j)}}{r^4|u|^{j+1}}.
\end{split}
\end{align}
Multiplying the above estimate by $r^2$, integrating  from $\Gamma_R$ and using the bound \eqref{proof5.4x} to estimate the boundary term, we thus obtain
\begin{align}\label{proof5.4xxx}
\begin{split}
|T^j\Phi|\leq \frac{4\sqrt{A}}{|u|^{j+1}}\frac{1+\frac{M}{3R}}{\sqrt{1-\frac{2M}{R}}}+\frac{Mj!\cin}{R|u|^{j+1}}+\frac{1}{1-\frac{2M}{R}}\left(\frac{3M^2}{2}\frac{j!\cin+\cbs^{(j)}}{R^2|u|^{j+1}}+\frac{3M\cbs^{(j)}}{R|u|^{j+1}}\right).
\end{split}
\end{align}
Importantly, $\cbs^{(j)}$ in the above estimate is either multiplied by decaying $R$-weights, or appears sublinearly inside a square root (in $A$).

We can now combine \eqref{proof5.4xxx} with \eqref{proof5.4y} and write
\begin{equation}
|r^2\pv T^j\pho|\leq |T^j\Phi|+M|T^j\pho|
\end{equation}
to close the bootstrap assumption, provided that $R$ and $\cbs^{(j)}$ are chosen large enough.
\end{proof}

In order to close the entire bootstrap argument, one can now first apply Proposition~\ref{prop:tlbs3} to $j=N$, then apply Proposition~\ref{prop:tlbs2} to $j=N-1$ and, finally, apply Proposition~\ref{prop:tlbs1} to all $j\leq N-2$. 
 One thus obtains that $\Delta$ is open and hence closes the bootstrap argument. In particular, we have established the proof of Proposition~\ref{prop:tl:pv}.

More systematically, one could instead make only the bootstrap assumptions \eqref{eq:tl:BSn} (without assuming \eqref{eq:tl:BS1n}), apply Proposition~\ref{prop:tlbs3} to all $j\leq N$ in order to close the bootstrap argument, and then use the identity \eqref{eq:waveequationtoestimatephi} in order to obtain the remaining estimates for $T^j(r\phi_1)$. This will be the approach followed in section~\ref{sec:general:timelike}.
\end{proof}
\subsubsection{Removing the smallness assumption on \texorpdfstring{$2M/R$}{2M/R}.}\label{sec:tl:removelargeness}
In the proofs of the previous sections \S\ref{sec5.3.2}--\S\ref{sec:subsubsec:timelikej=N}, we exclusively followed continuity methods, which required us to exploit $2M/R$ as a small parameter at various steps. 
It turns out that one can partially replace the continuity argument with a Gr\"onwall argument to remove all smallness assumptions on $2M/R$. 
Let us briefly sketch how this works. 

Let $\phi$ denote the finite solution as described in the beginning of \S\ref{sec:tl:pvT}.
First, we remark that the proof of Proposition~\ref{prop:tlbs3} shows that one can obtain an estimate of the form 
\begin{equation}\label{eq:tl:beyondN}
|\pv T^{j}(r\phi_1)|_{\Gamma_R}| \leq\frac{C}{\sqrt{R}}|u|^{-j-1}
\end{equation}
without requiring largeness in $R$ (this can be obtained by assuming a bootstrap estimate on $r^2\pv T^j(r\phi_1)|_{\Gamma_R}$ and improving it using the energy estimate, cf.\ \eqref{proof5.4x}). 

Equipped with this boundary term estimate, one can then obtain an estimate on $r^2\pv T^j(r\phi)$ throughout $\mathcal D_{\Gamma_R}\cap\{u\geq -k\}$ as follows: Let $(u,v)$ in $\mathcal D_{\Gamma_R}\cap\{u\geq -k\}$. For simpler notation, set $j=0$. Then, by the fundamental theorem of calculus,
\begin{align}
\left|\Phi(u,v)-\Phi(u,v_{R}(u))\right|\leq\int_{v_{R}(u)}^v r^2(u,v')\int_{-k}^u \left|	12MD\frac{r^2\pv(r\phi_1)}{r^5}-6M^2 D\frac{r\phi_1}{r^5}	\right|(u',v')\dd u'\dd v'.
\end{align}
Recalling the definition of $\Phi=r^2\pv(r\phi_1)-Mr\phi_1$, estimating the $Mr\phi_1$-term against the integral over $\pv(r\phi_1)$ from $\Gamma$, and applying Tonelli in the inequality above, we obtain:
\begin{multline}
|r^2\pv(r\phi_1)(u,v)|\leq  \left|r^2\pv(r\phi_1)(u,v_{R}(u))\right|+\int_{v_{R}(u)}^vM \frac{|r^2\pv(r\phi_1)|}{r^2}(u,v')\dd v'\\
+
\int_{-k}^u\int_{v_{R}(u)}^v r^2(u,v') \left|	12MD\frac{r^2\pv(r\phi_1)}{r^5}-6M^2 D\frac{r\phi_1}{r^5}	\right|(u',v')\dd v'\dd u'.
\end{multline}
We now also write the $r\phi_1$-term in the double integral as an integral over $\pv(r\phi_1)$ from $\Gamma$ and pull out the relevant suprema:
\begin{multline}
|r^2\pv(r\phi_1)(u,v)|\leq  \left|r^2\pv(r\phi_1)(u,v_{R}(u))\right|+\int_{v_{R}(u)}^v M\frac{|r^2\pv(r\phi_1)|}{r^2}(u,v')\dd v'\\
+
C\int_{-k}^u\frac{\sup_{v'\in[v_{R}(u),v]} |r^2\pv(r\phi_1)(u',v')|}{r^2(u',v_{\Gamma_R}(u))} + \frac{ \sup_{v'\in[v_{R}(u'),v]}|r^2\pv(r\phi_1)(u',v')|}{r^2(u',v_{R}(u))} \dd u'.
\end{multline}
We already control the  boundary  term on the RHS. 
If we now fix $u$ and just regard the last two integrals on the RHS as some monotonically increasing function of $v$, we can apply Gr\"onwall's inequality in the $v$-direction to obtain that
\begin{nalign}
|r^2\pv(r\phi_1)(u,v)|\leq  C\left(\left|r^2\pv(r\phi_1)(u,v_{R}(u))\right|
+
\int_{-k}^u \frac{ \sup_{v'\in[v_{R}(u'),v]}|r^2\pv(r\phi_1)(u',v')|}{r^2(u',v_{\Gamma_R}(u))} \dd u'\right).
\end{nalign}
Finally, we take the supremum in $v$, $\sup_{v'\in[v_R(u),v]}$, on the RHS and apply another Gr\"onwall inequality, this time in $u$. This then shows that
\begin{equation}
\sup_{v'\in[v_{R}(u),v]}|r^2\pv(r\phi_1)(u,v')|\leq  C\left|r^2\pv(r\phi_1)(u,v_{R}(u))\right|,
\end{equation}
and thus shows \eqref{eq:thm5.5} for $n=0$. 

Clearly, this approach requires no smallness assumption on $2M/R$ other than $R>2M$.
 Nevertheless, in hopes of simplifying the presentation, we will keep exploiting $2M/R$ as a small parameter throughout the remainder of the paper (i.e.\ \S\ref{sec:tl:pvT-T2} and \S \ref{sec:general:timelike}). 
 However, as the argument above shows, these smallness assumptions can always be lifted if one only wants to show \textit{upper bounds}. 
 The only times where we really need $2M/R$ as a small parameter is when we show \textit{lower bounds} on $r\phi_1$ etc., see already \S \ref{sec:tl:limitsub}.

\subsection{Estimates for \texorpdfstring{$\pv T^n(r\phi_1-|u|T(r\phi_1))$}{d/dv(Tn(r phi1+u T(r phi1))}}\label{sec:tl:pvT-T2}
The results obtained thus far are sufficient to show the first two estimates of Theorem~\ref{thm:tl}. 
In fact, not much modification is needed to also show certain lower bounds. 
However, something different needs to be done in order to establish the existence of the limit $\lim_{u\to-\infty}r^2\pu(r^2\pu\pho)$ (i.e.\ to prove eq.\ \eqref{eq:tl:thm:limit}).
 A crucial ingredient for this is to prove decay estimates for the differences $T^j(\pv(r\phi_1)-|u|\pv T(r\phi_1))$ (the reader may wish to already have a look at \S\ref{sec:tl:limitsub} to understand the role played by these quantities).

Therefore, let from now on $\phi_1$ be as described in the beginning of \S \ref{sec:tl:pvT}, but with the additional assumption that  also the lower bound \eqref{eq:tl:boundarydata2cutoff} holds on the boundary data.
We will now establish the following uniform decay estimates:
\begin{prop}\label{prop:tl:pvT-T2}
Let $\phi_1$ be the solution described above, and let $1\leq N'\leq N+2$. Then, if  $|U_0|$  is sufficiently large, there exists a constant $C=C(2M/R,\cin\,\ce)$ (in particular, this constant does not depend on $k$), which can be chosen independent of $R$ for large enough $R$, such that the following estimates hold throughout $\mathcal{D}_{\Gamma_R}\cap\{-k\leq u\leq U_0\}$:
\begin{align}
\left|r^2T^n\pv\left(r\phi_1-|u|T(r\phi_1)\right)\right|\leq \frac{C}{|u|^{n+1+\epsilon}}+\theta_k\cdot \frac{C}{k^{n+1}},&& n=0,\dots,N'.
\end{align}
\end{prop}
\begin{proof}The proof will be very similar to the proof of Proposition~\ref{prop:tl:pv}. We will again treat $2M/R$ as a small parameter, keeping in mind that this restriction can lifted as in \S \ref{sec:tl:removelargeness}.

\subsubsection{The bootstrap assumptions}
\newcommand{\cbse}{C_{\mathrm{BS},\epsilon}}
Let $\{\cbse^{(n)},0=1,\dots,N'\}$ be a set of sufficiently large positive constants, and let $\Delta$ be defined as in~\ref{sec:tl:BS}, with the additional requirement that also 
\begin{equation}\label{eq:tl:BS2}
\tag{BS''(n)}
\left|r^2T^n\pv\left(r\phi_1-|u|T(r\phi_1)\right)\right|\leq \frac{\cbse^{(n)}}{|u|^{n+1+\epsilon}}+\theta_k\cdot \frac{\cbse^{(n)}}{k^{n+1}}
\end{equation}
holds for $n=0,\dots,N$'.
We shall improve these estimates in the following. Note that we only assume estimates on the $\pv$-derivatives, so we can just use the method of the proof of Proposition~\ref{prop:tlbs3} with some adaptations. The crucial observation is that, while the differences $T^n(\phi_1-|u|T\phi_1)$ do not solve the wave equation, the error term is of the form $\pv T^{n+1}(r\phi_1)$, over which we already have sharp control by Proposition~\ref{prop:tl:pv}.
\subsubsection{Improving the bootstrap assumptions}
\begin{prop}
Let $j\in\{0,\dots,N'\}$ for $N'\leq N+2$. Then, for sufficiently large values of $R/2M$ and $|U_0|$, and if $\cbse^{(j)}$ is chosen large enough, we have throughout $\Delta$ that, in fact,
\begin{equation}
\left|r^2T^j\pv\left(r\phi_1-|u|T(r\phi_1)\right)\right|\leq \frac12\frac{\cbse^{(j)}}{|u|^{j+1+\epsilon}}+\frac12 \theta_k\cdot \frac{\cbse^{(j)}}{k^{n+1}}.
\end{equation}
\end{prop}
\begin{proof}
We shall only need to assume \eqref{eq:tl:BS2} for $n=j$ and, in addition, the results of Proposition~\ref{prop:tl:pv}. We shall follow the structure of the proof of Proposition~\ref{prop:tlbs3}.

First, we require an estimate of $r^2T^j\pv(r\phi_1-|u|T(r\phi_1))$ on the boundary $\Gamma_R$. 
Recall that, in the previous proof, we obtained such an estimate by using an energy estimate to obtain a bound on $\sqrt{r}T^j\phi_1$ with sharp decay in $u$, and by then using the wave equation to convert this into a bound for $\pv T^j(r\phi_1)$ on $\Gamma_R$. 
Proceeding along the same lines for the differences under consideration, we are led to consider the current $J^T[T^j(\phi_1-|u|T\phi_1)]$. The divergence of this current is no longer vanishing. Instead, we have\footnote{We abuse notation and write $\phi_{\ell=1}=\phi_1$.}
\newcommand{\pvv}{r^2T^j\pv(r\phi_1-|u|T(r\phi_1))}
\begin{align}
\begin{split}
\div J^T\left[T^j\left(\phi_1-|u|T\phi_1\right)\right]&=\Box_g\left(T^j\left(\phi_1-|u|T\phi_1\right)\right)\cdot T\left(T^j\left(\phi_1-|u|T\phi_1\right)\right)\\
&=-\frac{1}{Dr}\pv T^{j+1}(r\phi_1)\cdot\frac{1}{r}T^{j+1}\left(r\phi_1-|u|T(r\phi_1)\right),
\end{split}
\end{align}
where we used the formula \eqref{eq:box} for $\Box_g$. Using the estimates from Proposition~\ref{prop:tl:pv} and the fact that $j+2\leq N'+2\leq N$, we can thus bound $\div J^T$ as follows:
\begin{align}\label{eq:bulkestimate}
\left|\div J^T\left[T^j\left(\phi_1-|u|T\phi_1\right)\right]\right|\leq 2\cdot\frac{C^2}{R\cdot r^{4}|u|^{2j+4}}.
\end{align}
Applying the divergence theorem, in the form of \eqref{eq:divergencetheorem},  to the current $J^T$, and doing the integrals over the sphere, we then arrive at (compare to eq.\ \eqref{random1})
\begin{align}\label{proof5.6z}
\begin{split}
&\int_{-k}^u r^2\left(\pu T^j(\phi_1-|u'|T\phi_1)\right)^2\dd u'	 \\
\leq &\int_{\Gamma_R\cap\{-k\leq u'\leq u\}}  2r^2 T^{j+1}(\phi_1-|u'|T\phi_1)\cdot (2\pv-T)T^j(\phi_1-|u'|T\phi_1)	 \dd u'\\
+&\int_{-k}^u\int_{v_R(u')}^v r^2\left|\div J^T\left[T^j(\phi_1-|u'|T\phi_1)\right]\right|\dd u'\dd v'.
\end{split}
\end{align}
We have already estimated the terms inside the bulk term in \eqref{eq:bulkestimate}. Indeed, 
we can see that the contribution to the RHS of the estimate \eqref{proof5.6z} above is  subleading: 
\begin{equation}
\int_{-k}^u\int_{v_R(u)}^v r^2\left|\div J^T[T^j(\phi_1-|u'|T\phi_1)]\right|\dd u'\dd v'\leq \frac{\tilde{C}}{R^2|u|^{2j+3}}
\end{equation}
for some constant $\tilde{C}$.
On the other hand, we  can estimate the boundary term in \eqref{proof5.6z} by plugging in the boundary data assumptions \eqref{eq:tl:boundarydata2cutoff} for $n=j,j+1$ and the bootstrap assumption \eqref{eq:tl:BS2} for $n=j$. This gives:
\begin{align*}
& \int_{\Gamma_R\cap\{-k\leq u'\leq u\}}  2r^2 T^{j+1}(\phi_1-|u'|T\phi_1)\cdot (2\pv-T)T^j(\phi_1-|u'|T\phi_1)	 \dd u'\\
 \leq&  \int_{\Gamma_R\cap\{-k\leq u'\leq u\}}2\left(\frac{\ce}{R|u'|^{2+j+\epsilon}}+\frac{\cin}{Rk^{j+2}}	\right)\\
 &\cdot\left(\frac{2}{R^2}\left(\frac{\cbse^{(j)}}{|u'|^{j+1+\epsilon}}+\frac{\cbse^{(j)}}{k^{j+1}}\right)+\left(\frac{2D}{R}+\frac{1}{|u'|}\right)\left(\frac{\ce}{R|u'|^{1+j+\epsilon}}+\frac{\cin}{Rk^{j+1}}	\right)\right)\dd u' \\
 \leq&\frac{8}{R^3}\max(\ce,\cin)\left(\cbse^{(j)}+\max(\ce,\cin)\right)\left(	\frac{1}{k^{j+1}}+\frac{1}{|u|^{j+1+\epsilon}}	\right)^2\cdot (1+\mathcal{O}(|u|^{-1}).
\end{align*}
We thus find, using the fundamental theorem of calculus, Cauchy--Schwarz and the energy estimate \eqref{proof5.6z} above, that
\begin{equation}\label{proof5.6x}
\sqrt{Dr}|T^j(\phi_1-|u|T\phi_1)|\leq \frac{\sqrt{A'}}{R^{\frac32}}\left(\frac{1}{k^{j+1}}+\frac{1}{|u|^{j+1+\epsilon}}\right)(1+\mathcal{O}(|u|^{-\frac12})),
\end{equation}
where $A'$ is a constant which, importantly, only depends \textit{linearly} on $\cbse^{(j)}$. We now use the wave equation \eqref{eq:l=1waveequation} in order to derive an estimate for the $\pv$-derivative on the boundary. We compute from \eqref{eq:l=1waveequation} that
\begin{align}\label{proof5.6yy}
\begin{split}
\pu\pv(T^j(r\phi_1&-|u| T(r\phi_1)))=-\frac{2D}{r^2}\left(1+\frac{M}{r}\right)(T^j(r\phi_1-|u| T(r\phi_1)))+\pv T^{j+1}(r\phi_1).
\end{split}
\end{align}
Note that we control the error term $\pv T^{j+1}(r\phi_1)$ by Proposition~\ref{prop:tl:pv}; in fact, it is subleading in terms of $u$-decay.
Integrating \eqref{proof5.6yy} from $u=-k$, and plugging in the estimate \eqref{proof5.6x}, we find that (see also \eqref{proof5.4x})
\begin{equation}
|\pv T^j(r\phi_1-|u|T(r\phi_1))|\leq \frac{4\sqrt{A'}}{R^2}\frac{1+\frac{M}{3R}}{\sqrt{1-\frac{2M}{R}}}\left(\frac{1}{k^{j+1}}+\frac{1}{|u|^{j+1+\epsilon}}\right)(1+\mathcal{O}(|u|^{-\frac12})).
\end{equation}
We have thus established an estimate on the boundary term. Now, in order to improve the bootstrap assumption, we want to appeal to the approximate conservation law \eqref{eq:l=1cons.law.in.u}. We compute that
\begin{align}
\begin{split}
	&\pu\left(r^{-2}\pv T^j(\Phi-|u|T\Phi)\right)
	=r^{-2}\pv T^{j+1}\Phi\\
	&-\frac{6M^2D}{r^5}(T^j(r\phi_1-|u| T(r\phi_1)))-\frac{12MD}{r^5}r^2\pv T^j(r\phi_1-|u| T(r\phi_1)).
\end{split}
\end{align}
Again, we control the error term $r^{-2}\pv T^{j+1}\Phi$ by Proposition~\ref{prop:tl:pv}; in fact, it has more $u$-decay than the other terms:
\[|r^{-2}\pv T^{j+1}\Phi|\leq 	\frac{C}{r^4|u|^{j+2}}.	\]
 Converting some of the additional $|u|$-decay present in $\pv T^{j+1}\Phi$ into $r$-decay,
 \[|r^{-2}\pv T^{j+1}\Phi|\leq 	\frac{C}{r^{5-\delta}|u|^{j+1+\delta}},	\]
 for some suitable $1>\delta>\epsilon$, and repeating the computations leading to \eqref{proof5.4xxx}, we thus find
\begin{multline}\label{proof5.6xx}
|T^j(\Phi-|u|T\Phi)|\leq \frac{4\sqrt{A'}}{R^2}\frac{1+\frac{M}{3R}}{\sqrt{1-\frac{2M}{R}}}\left(\frac{1}{k^{j+1}}+\frac{1}{|u|^{j+1+\epsilon}}\right)+\frac{M\ce}{R|u|^{j+1+\epsilon}}+ \frac{M\cin}{Rk^{j+1}}\\
+\frac{1}{1-\frac{2M}{R}}\left(\frac{3M^2}{2}\frac{\ce+\cbse^{(j)}}{R^2}+\frac{3M\cbse^{(j)}}{R}\right)\left(\frac{1}{k^{j+1}}+\frac{1}{|u|^{j+1+\epsilon}}\right)
+\mathcal O(|u|^{-{j-1-\delta}}).
\end{multline}
Importantly, $\cbse^{(j)}$ in the above estimate appears either multiplied by decaying $R$-weights or sublinearly inside a square root (in $A'$).
Therefore, if $R/2M$ and $\cbse^{(j)}$ are sufficiently large, we can improve the bootstrap assumption (and thus prove the proposition) by again writing
\begin{equation}
\left|r^2 T^j\pv\left(r\phi_1-|u|T(r\phi_1)\right)\right|\leq \left|T^j\left(\Phi-|u|T\Phi\right)\right|+M\left|T^j\left(r\phi_1-|u|T(r\phi_1)\right)\right|.
\end{equation}
\end{proof}
This concludes the proof of Proposition~\ref{prop:tl:pvT-T2}.
\end{proof}

\subsection{Proof of Thm.~\ref{thm:tl}}\label{sec:tl:limit}
Recall from \S\ref{sec:tl:cutoffdata} the definition of the sequence of solutions $\phi_1^{(k)}$, each arising from  data satisfying \eqref{eq:tl:boundarydata1cutoff}, \eqref{eq:tl:boundarydata2cutoff} and \eqref{eq:tl:noincomingfinite}. We have shown sharp, uniform-in-$k$ decay for these solutions in Propositions~\ref{prop:tl:pv} and~\ref{prop:tl:pvT-T2}. 

We will now smoothly extend these solutions to the zero solution for $u\leq -k$ and show that they  converge uniformly to a pointwise limit $\phi_1$ as $k\to -\infty$, which therefore still satisfies the uniform bounds of Propositions~\ref{prop:tl:pv} and~\ref{prop:tl:pvT-T2}.
\subsubsection{Sending \texorpdfstring{$\mathcal{C}_{u=-k}\to\mathcal{I}^-$}{C(u=-k) to I-}}
\begin{prop}\label{prop:tl:limit}
Let $\{\phi_1^{(k)}\}_{k\in\mathbb{N}}$ be the sequence of solutions described in \S\ref{sec:tl:cutoffdata} extended with the zero solution for $u\leq -k$. This sequence $\{\phi_1^{(k)}\}_{k\in\mathbb{N}}$ tends to a uniform limit $\phi_1$ as $k\to\infty$,
\begin{equation}
\lim_{k\to\infty}||\phi_1^{(k)}-\phi_1 	||_{C^N(\mathcal{D}_{\Gamma_R})}=0.
\end{equation}
In fact, this limiting solution is the unique smooth solution that restricts correctly to the data of \S \ref{sec:subsec:tl:setup}, and it satisfies, throughout $\mathcal{D}_{\Gamma_R}\cap\{u\leq U_0\}$, and for sufficiently large negative values of $U_0$, the following bounds for some constant $C=C(2M/R,\cin,\ce)$ which can be chosen independent of $R$ for large enough $R$:
\begin{align}\label{eq:proplimit2}
\left|r^2\pv T^n (r\phi_1)(u,v)\right|&\leq \frac{C}{|u|^{n+1}},&&n=0,1,\dots,N,\\
\left|T^n (r\phi_1)(u,v)\right|&\leq \frac{C}{|u|^{n+1}}\max\left(r^{-1},|u|^{-1}\right),&&n=0,1,\dots,N-1.\label{eq:proplimit3}
\end{align}
Moreover, if $N'\leq N+2$, we also have
\begin{align}\label{eq:proplimit4}
\left|r^2T^n\pv(r\phi_1-|u|T(r\phi_1))\right|\leq \frac{C}{|u|^{n+1+\epsilon}},&& n=0,1,\dots,N'.
\end{align}
\end{prop}
\begin{proof}
We show that the sequence is Cauchy. Let $\delta>0$ arbitrary. We need to show that there exists $K\in\mathbb{N}$ such that
\begin{equation}\label{eq:tl:nu}
||\phi_1^{(n)}-\phi_1^{(k)}||_{C^N(\mathcal{D}_{\Gamma_R})}<\delta
\end{equation}
for all $n,k>K$. 
This is done by splitting $\mathcal{D}_{\Gamma_R}$ into three regions: $u\leq -n$, $-n\leq u\leq -k+1$ and $-k+1\leq u$, where we assumed without loss of generality that $n>k$. 

In the first region, $u\leq -n$, both solutions vanish, so there is nothing to show. 

Notice that, by linearity, the difference $\Delta\phi_1:=\phi_1^{(n)}-\phi_1^{(k)}$ is itself a solution to the wave equation \eqref{eq:l=1waveequation}, with vanishing data on $u=-n$ and compactly supported boundary data on $\Gamma_R\cap\{-n\leq u\leq -k+1\}$. 
We can therefore simply apply the results of Proposition~\ref{prop:tl:pv} to $\Delta\phi_1$ in the second region, $-n\leq u\leq -k+1$, and obtain, for some constant $C_1$, that
\begin{equation}\label{eq:tl:nuk}
||\phi_1^{(n)}-\phi_1^{(k)}||_{C^N(\mathcal{D}_{\Gamma_R}\cap\{-n\leq u\leq -k+1\})}\leq \frac{C_1}{k}.
\end{equation}

In the third region\footnote{In fact, the approach for the third region can be used in all of $\mathcal D$.}, $-k+1\leq u$, we  apply the energy estimate \eqref{random1} to the difference $\Delta\phi_1$:
\begin{align*}
\begin{split}
&\int_{-n}^u r^2(\pu T^j\Delta\phi_{1})^2+D|T^j\Delta \phi_{1}|^2\dd u'	 \\
\leq&\int_{\Gamma_R\cap\{-n\leq u'\leq -k+1 \}}  2r^2 T^{j+1}\Delta\phi_{1}\cdot (2\pv-T)T^j \Delta\phi_{1}	 \dd u'.
\end{split}
\end{align*}
Here, we used that the boundary data for $\Delta\phi_1$ are compactly supported in $u\leq -k+1$.
We can now estimate the integral over $\Gamma_R$ by plugging in the boundary data assumptions for the $T^j\Delta\phi_1$-terms and by plugging in the previously obtained estimate \eqref{eq:tl:nuk} for the terms $\pv T^j \Delta\phi_1$. We thus find that
\begin{align*}
|T^j\Delta\phi(u,v)|\leq \left(\int_{-n}^u \frac{1}{r^2}\dd u'\right)^{\frac12}\left(\int_{-n}^u	r^2(\pu T^j\Delta\phi_1(u,v))^2\dd u'	\right)^{\frac12}\leq\frac{1}{\sqrt{r}} \frac{C}{k^{j+1}}
\end{align*}
for some constant $C>0$. From these estimates on $T^j\Delta\phi_1$, we can obtain estimates on $\pv T^j\Delta\phi_1$ by simply integrating \eqref{eq:l=1waveequation} from $u=-k+1$ (where $\pv T^j \Delta\phi_1\lesssim k^{-j-1}$). This shows that there exists a constant $C_2$ such that
 \begin{equation}\label{eq:tl:nuk2}
||\phi_1^{(n)}-\phi_1^{(k)}||_{C^N(\mathcal{D}_{\Gamma_R}\cap\{u\geq -k+1\})}\leq \frac{C_2}{k}.
\end{equation}
Combining \eqref{eq:tl:nuk} and \eqref{eq:tl:nuk2} shows that \eqref{eq:tl:nu} holds for all $n>k>K$ provided that $K>\frac{C_1+C_2}{2\delta}$.

We have thus established the uniform convergence of the sequence $\{\phi_1^{(k)}\}$. In view of the uniformity of the convergence, the bounds from Propositions~\ref{prop:tl:pv} and~\ref{prop:tl:pvT-T2}  carry over to the limiting solution, thus proving the estimates  \eqref{eq:proplimit2}--\eqref{eq:proplimit4}. Moreover, the methods of the proof show that this is the unique solution that has vanishing energy flux on $\mathcal I^-$ and satisfies the assumptions of \S \ref{sec:subsec:tl:setup}. This concludes the proof.
\end{proof}
\subsubsection{The limit \texorpdfstring{$\lim_{u\to-\infty,v=\text{constant}}r^2\pu\left(r^2\pu(r\phi)\right)$}{lim (r2 d/du (r2 d/du(r phi1)))}}\label{sec:tl:limitsub}
Finally, we establish that the limiting solution constructed above satisfies \eqref{eq:tl:thm:limit}. For this, we will also need to assume the lower bound  \eqref{eq:tl:thm:lowerbound} on data.
\begin{prop}\label{prop:limit}
Consider the solution of Proposition~\ref{prop:tl:limit}, and assume in addition that $N\geq 4$ and $N'\geq2$, as well as the lower bound on data \eqref{eq:tl:thm:lowerbound}. Then  the following limit exists, is independent of $v$, and is non-vanishing so long as $\cin$ is non-vanishing and $R/2M$ is sufficiently large:
\begin{equation}\label{eq:prop8limit1}
\lim_{u\to -\infty}r^2\pu\left(r^2\pu(r\phi_1(u,v))\right)=\tilde C\neq 0. 
\end{equation}
Moreover, we have that $\lim_{u\to-\infty}r^2\pu(r\phi_1)(u,v)=0$ and
\begin{equation}\label{eq:prop8limit2}
r^2\pu\left(r^2\pu(r\phi_1(u,v))\right)-\lim_{u\to-\infty}r^2\pu\left(r^2\pu(r\phi_1(u,v))\right)=\mathcal{O}(\max(|u|^{-\epsilon},r^{-1})).
\end{equation}
\end{prop}
\begin{proof}
We first establish the existence of the limit
\begin{equation}
\lim_{u\to-\infty}|u|^2 r^2\pv T(r\phi_1)(u,v)=:\mathcal{L}(v)
\end{equation}
by computing 
\begin{align}\label{eq:tl:limdifference}
\begin{split}
\pu(|u|^2 r^2\pv T(r\phi_1))
=\underbrace{-2|u|r^2\pv T(r\phi_1)+|u|^2r^2\pv T^2(r\phi_1)}_{=\mathcal{O}(|u|^{-1-\epsilon})}-\underbrace{|u|^2\pv(r^2\pv T(r\phi_1))}_{=\mathcal{O}(r^{-2})},
\end{split}
\end{align}
where we bounded the first two terms using \eqref{eq:proplimit4}, and the third term by plugging in the bounds \eqref{eq:proplimit2}, \eqref{eq:proplimit3} into the approximate conservation law \eqref{eq:l=1cons.law.in.u} and integrating in $u$ (see also \eqref{proof5.3xx}). 
In fact, the bound on the last term also shows that $\mathcal{L}(v)$ is independent of $v$, $\mathcal L(v)\equiv \mathcal L$.  We have thus shown that
\begin{equation}
|u|^2 r^2\pv T(r\phi_1)-\lim_{u\to-\infty}|u|^2 r^2\pv T(r\phi_1)=\mathcal{O}(\max(|u|^{-\epsilon},r^{-1})).
\end{equation}

We now show that $\mathcal{L}$ is non-vanishing by using the lower bound \eqref{eq:tl:thm:lowerbound}: We have from (the $T$-commuted) equation \eqref{2} that
\begin{align*}
r^{-2}\pv T(r^2\phi_1)=\int \frac{D T^2(r\phi_1)}{r^2}-\frac{8MD}{r^4}T(r\phi_1)\dd u
		\leq \frac{MC}{r^4|u|^2}+\mathcal{O}(|u|^{-3}r^{-2}).
\end{align*}
Evaluating the above on $\Gamma_R$ gives
\begin{align*}
\left|R^2\pv T(r\phi_1)+D R^2 T\phi_1|_{\Gamma_R}\right|\leq \frac{MC}{R|u|^2}+\mathcal{O}(|u|^{-3}R^{-2}),
\end{align*}
which, if $R$ is chosen large enough (recall that the constant $C$ in the estimates above can be chosen independently of $R$ if $R$ is large enough), can be chosen to be less than $\frac{\cin}{4|u|^2}$. This results in the following lower bound on $\Gamma_R$:
\begin{align}
\left|r^2\pv T(r\phi_1)|_{\Gamma_R}\right|\geq \frac{\cin}{4|u|^2}.
\end{align}
Using once more the estimate on $|u|^2\pv(r^2\pv T(r\phi_1))\sim r^{-2}$ and integrating it from $\Gamma_R$ shows that, if $R$ is chosen suitably large, we in fact have
\begin{align}
\left|r^2\pv T(r\phi_1)|_{\Gamma_R}\right|\geq \frac{\cin}{8|u|^2}.
\end{align}
In summary, we have thus established that
$
\lim_{u\to-\infty}|u|^2 r^2\pv T(r\phi_1)(u,v)=\mathcal{L}(v)\equiv \mathcal{L}\neq 0.
$

Finally, in order to relate $\mathcal L$ to the limit of $r^2\pu(r^2\pu(r\phi_1))$ in question, we write
\begin{nalign}\label{eq:tl:pupu}
r^2\pu(r^2\pu(r\phi_1))=&r^4 T^2(r\phi_1)-r^2\pv(r^2 T(r\phi_1))-r^2\pu(r^2\pv(r\phi_1))\\
=&r^4 T^2(r\phi_1)-r^4\pv T(r\phi_1)-2D r^3 T(r\phi_1)+2D r^2\pv(r^2\phi_1)+6MD r^2\phi_1,
\end{nalign}
where we used eq.\ \eqref{3} in the last line. Notice that the last term decays like $|u|^{-1}$, whereas the other terms do not decay. We now express the limits of each of the other terms in terms of $\mathcal{L}$.
\paragraph{Computing $\lim_{u\to-\infty} |u|^{j+1}rT^j(r\phi_1)$.}
Observe that \eqref{eq:proplimit4} implies that 
\begin{align}\label{proof5.8x}
\begin{split}
\lim_{u\to-\infty} |u|^2r^2\pv T(r\phi_1)=&\frac12\lim_{u\to-\infty} |u|^3r^2\pv T^2(r\phi_1)\\
=\frac16 \lim_{u\to-\infty}|u|^4r^2\pv T^3(r\phi_1)
=&\dots=\frac{1}{(j+1)!}\lim_{u\to-\infty} |u|^{2+j}r^2\pv T^{j+1}(r\phi_1)
\end{split}
\end{align}
for all $j\leq N'$.
Now, use the wave equation \eqref{eq:l=1waveequation} to write
\begin{align*}
-2D T^j(r\phi_1)\left(1+\frac{M}{r}\right)&=r^2\pv T^{j+1}(r\phi_1)-r^2\pv^2 T^{j}(r\phi_1)\\
&=r^2\pv T^{j+1}(r\phi_1)+2Dr\pv T^j(r\phi_1)-\pv(r^2\pv T^j(r\phi_1)).
\end{align*}
The last term decays faster than the others, and we conclude that
\begin{align}
\begin{split}
-2\lim_{u\to-\infty} |u|^{j+1}r T^j(r\phi_1)
=\lim_{u\to-\infty}|u|^{j+2}r^2\pv T^{j+1}(r\phi_1) +2\lim_{u\to-\infty}|u|^{j+1}r^2\pv T^j(r\phi_1).
\end{split}
\end{align}
Plugging \eqref{proof5.8x} into the expression above, and setting $j=1, 2$, respectively, we thus obtain
\begin{align}\label{proof5.8y}
\lim_{u\to-\infty}|u|^2r T(r\phi_1) &=-\frac{1}{2}(2\mathcal L+2\mathcal L)=-2\mathcal{L},\\
\lim_{u\to-\infty}|u|^3 r T^2(r\phi_1)&=-\frac12 (6\mathcal L+2\cdot 2\mathcal L)=-5\mathcal{L}.\label{proof5.8yy}
\end{align}
\paragraph{Computing $\lim_{u\to-\infty}r^2\pv(r^2\phi_1)$.}
In order to compute the limit of $r^2\pv(r^2\phi_1)$, we use equation \eqref{2} to write
\begin{align*}
r^2\pv(r^2\phi_1)&=r^4\int \frac{DT(r\phi_1)}{r^2}-\frac{8MD(r\phi_1)}{r^4}\dd u\\
				 &=r^4\int \frac{\lim_{u\to-\infty} |u|^2 r T(r\phi_1)}{r^3|u|^2}\dd u+\mathcal{O}(|u|^{-1}+r|u|^{-1-\epsilon}),
\end{align*}
from which we conclude, using \eqref{eq:v-u=r}, that
\begin{equation}\label{proof5.8yyy}
\lim_{u\to-\infty}r^2\pv(r^2\phi_1)=\frac14 \lim_{u\to-\infty}|u|^2r T(r\phi_1)=-\frac12\mathcal{L}.
\end{equation}

Finally, inserting the identities \eqref{proof5.8y}, \eqref{proof5.8yy} and \eqref{proof5.8yyy} back into \eqref{eq:tl:pupu}, we find that
\begin{equation}
\lim_{u\to-\infty}r^2\pu(r^2\pu(r\phi_1))=-5\mathcal{L}-\mathcal{L}+2\cdot2\mathcal{L}-\frac{2}{2}\mathcal{L}=-3\mathcal{L}.
\end{equation}
This proves equation \eqref{eq:prop8limit1}. Estimate \eqref{eq:prop8limit2} as well as the vanishing of $\lim_{u\to-\infty}r^2\pu(r\phi)$ follow similarly.
\end{proof}
Combining the previous two propositions, Propositions~\ref{prop:tl:limit} and~\ref{prop:limit}, proves Theorem~\ref{thm:tl}.

\subsection{A comment on the stationary solution}\label{sec:tl:comments}
We have already remarked in \S\ref{sec:nl:comments} that we expect the vanishing of $\lim_{u\to-\infty} r^2\pu(r\phi_1)$ to lead to late-time asymptotics with logarithmic terms appearing at leading order if the data on $\Gamma_R$ are smoothly extended to $\mathcal{H}^+$ (for instance, we would have $r\phi_1|_{\mathcal{I}^+}(u)\sim \frac{\log u}{u^3}$ as $u\to\infty$).
 In other words, we expect our choice of polynomially decaying boundary data to lead to a logarithmically modified Price's law for $\ell=1$.

Note that the limit $\lim_{u\to-\infty} r^2\pu(r\phi_1)$ would not vanish if we included the stationary solution, that is to say: if we added a constant to our initial data. Using the structure of equations\footnote{For the sake of completeness, we mention here that these equations are special cases of \begin{align}\label{55}
\pu(r^{-2L}\pv(r^L\cdot r\phi_L))&=\frac{LD}{r^{L+1}}T(r\phi_L)-\frac{2MD(L+1)^2}{r^{L+3}}r\phi_L,\\
\pu(r^{L+1}\pv(r\phi_L))&=-(L+1)D\pv(r^L\cdot r\phi_L)-(1+L(L+1))2MD r^{L-2}r\phi_L.\label{66}
\end{align}
} \eqref{2}, \eqref{3} presented in the proof of Proposition~\ref{prop:tlbs1}, or other methods, it is indeed not difficult to see that the stationary solution behaves like $\pv(r\phi_1)\sim -\pu(r\phi_1)\sim r^{-2}$. 

On the other hand, we see that if we prescribe \textit{decaying} data on $\Gamma_R$, then the solution will behave, roughly speaking, like the stationary solution multiplied by that decay.
Now, since the stationary solution for higher $\ell$-modes will decay faster in $\ell$, $r\phi_\ell\sim r^{-\ell}$ (see eqns. \eqref{55}, \eqref{66}), we thus expect that if we prescribe decaying boundary data for higher $\ell$-modes, then the corresponding solution will decay increasingly faster towards $\mathcal{I}^-$, and $(r^2\pu)^j(r\phi_\ell)$ will vanish to higher and higher orders. 
We will build on this intuition and make it precise in \S\ref{sec:general:timelike}. 

\newpage
\section{Boundary data on a timelike hypersurface \texorpdfstring{$\Gamma_f$}{Gamma-f}}\label{sec:timelikenonconstantr}
In the previous section, we showed how to construct solutions and prove sharp decay in the case of polynomially decaying boundary data on hypersurfaces of constant $r=R$. We now outline how to generalise to spherically symmetric hypersurfaces on which $r$ is allowed to vary. In fact, not much modification will be needed.
\newcommand{\te}{T^{(s)}}
\subsection{The setup}\label{sec:tf:setup}
For the sake of notational simplicity, we restrict our attention to spherically symmetric hypersurfaces $\Gamma_{f}\subset\mathcal M$ that have timelike generators  that are given by\footnote{The proofs presented below also directly apply to slightly more general spherically symmetric timelike generators, e.g.\ $\te \sim T-f_s\partial_v$, with $f_s\sim |u|^{-s}$, or $\te \sim T-f_s\log|u|\partial_v$ etc.}
\begin{align}
\te=\pu+\frac{1}{1+|u|^{-s}}\pv=T-\frac{|u|^{-s}}{1+|u|^{-s}}\pv,&&s>0.
\end{align}
Notice that we normalised $\te$ such that $\te u=1$.


Since the cases $s>1$ and $s\leq1$ are quite different, we shall treat them separately. Let's first consider the case $s>1$:
\subsection{The case where \texorpdfstring{$r|_{\Gamma_f}$}{r-Gamma-f} attains a finite limit (\texorpdfstring{$s>1$}{s>1}):}
\subsubsection{Initial/boundary data assumptions and the main theorem}
Let $\Gamma_f$ be as described above, and let $s>1$. We then prescribe smooth boundary data $\hat{\phi}_1$ for $\phi_1$ on $\Gamma_f$ which satisfy, for $u\leq U_0< 0$ and $|U_0|$ sufficiently large:
\begin{align}\label{eq:tlf:boundarydata1}
\left|(\te)^n(r\hat{\phi}_1)\right|&\leq \frac{n! \cin}{r|u|^{n+1}},&&n=0,1,\dots,N+1,\\
\left|(\te)^n(r\hat{\phi}_1-|u|\te(r\hat{\phi}_1))\right|&\leq \frac{\ce}{r|u|^{n+1+\epsilon}},&& n=0\dots,N'+1\label{eq:tlf:boundarydata2}
\end{align}
for some positive constants $\cin$, $\ce$, $\epsilon\in(0,1)$ and for  $N,N'>1$ positive integers.
Moreover, we demand, in a limiting sense, that, for all $v$
\begin{equation}\label{eq:tlf:noincomingradiation}
\lim_{u\to-\infty}\pv^n(r\phi_1)(u,v)=0,\quad n=1,\dots N+1.
\end{equation} 
Then we obtain, as in the previous section:
\begin{thm}\label{thm:tlf}
Let $N\geq 4$ and $N'\geq 2$. Then there exists a unique solution $\phi_1$ to eq.\ \eqref{eq:l=1waveequation} in $\mathcal{D}_{\Gamma_f}=\mathcal{M}\cap\{v\geq v_{\Gamma_f}(u)\}$ which restricts correctly to $\hat{\phi}_1$ on $\Gamma_f$, $\phi_1|_{\Gamma_f}=\hat{\phi}_1$, and which satisfies \eqref{eq:tlf:noincomingradiation}.

Moreover, the estimates from Theorem~\ref{thm:tl} apply to this solution, with $\tilde C\neq 0$ being non-zero provided that a lower bound on data is specified and that $R/2M$ is sufficiently large.
\end{thm}
\subsubsection{Outline of the proof}
As the proof only requires small modifications to the proof of Theorem~\ref{thm:tl}, we will only give an outline. 
There are two closely related ways one can go about this:
One can either work with the generators of $\Gamma_f$, i.e.\ replace all $T$'s from the proof of Theorem~\ref{thm:tl} with $\te$'s, and estimate the resulting error terms (which would always exhibit faster decay than the other terms) -- this was the approach of~\cite{I}.
Alternatively, one can continue working with $T$ and exploit the fact that the difference of, say, $T(r\phi_1)-\te(r\phi_1)=\frac{|u|^{-s}}{1+|u|^{-s}}\pv(r\phi_1)$ decays faster than either of the terms on the left-hand side as long as $s>1$. 
Thus, an estimate on $T(r\phi_1)$ immediately gives control on $\te(r\phi_1)$ and vice versa.\footnote{Note that this is no longer true if $s\leq 1$.} 
We shall follow the second approach:
\begin{proof}[Proof of Theorem~\ref{thm:tlf}]
First, we cut off the data as in section~\ref{sec:tl:cutoffdata}. We then introduce the set of bootstrap assumptions as in section~\ref{sec:tl:BS} (with the only modification that the set $X$ defined below eq.\ \eqref{eq:tl:BS1n} now contains all $v\geq v_{\Gamma_f}(u)$). The proof of Proposition~\ref{prop:tlbs1} remains unchanged. The proof of Proposition~\ref{prop:tlbs2} requires the modification that, now, it isn't $T^j(r\phi_1)$ which on $\Gamma_f$ is given by data, but $(\te)^j\pho$.
However, this can easily be dealt with by writing
\begin{equation}
T(r\phi_1)=\te(r\phi_1)+\frac{|u|^{-s}}{1+|u|^{-s}}\pv(r\phi_1)
\end{equation}
(and similarly for $T^j$), and then plugging in the bootstrap assumption for $\pv(r\phi_1)$, using the fact that, because $s>1$, the $\pv(r\phi_1)$-term has more $u$-decay than the $\te(r\phi_1)$-term. It can thus be absorbed into the latter for large enough $|U_0|$.

Let's now move to the proof of Proposition~\ref{prop:tlbs3}. Applying the divergence theorem gives (we denote the induced volume element on $\Gamma_f$ by $r^2\dd t_{\Gamma_f}\dd \Omega$)
\begin{align}
\begin{split}
&\int_{\mathcal{C}_{v}\cap\{-k\leq u'\leq u\}} r^2\dd u' \dd\Omega\, J^T[T^j\phi_{\ell=1}]\cdot \pu \\
\leq&\int_{\Gamma_f\cap\{-k\leq u'\leq u\}}  r^2 \dd t_{\Gamma_f}\dd \Omega	\,J^T[T^j\phi_{\ell=1}]\cdot\left(\pu-\pv+\frac{|u'|^{-s}}{1+|u'|^{-s}}\pv\right),
\end{split}
\end{align}
which implies (cf.\ \eqref{random1})
\begin{nalign}\label{random11}
&\int_{-k}^u r^2(\pu T^j\phi_{1})^2+D|T^j\phi_{1}|^2\dd u'	 \\
\lesssim &\int\limits_{\Gamma_f\cap\{-k\leq u'\leq u\}}  r^2 \left(T^{j+1}\phi_{1}\cdot (T-2\pv)T^j \phi_{1}+\frac{|u'|^{-s}}{1+|u'|^{-s}}\left((T^j \pv \phi_1)^2+\frac{2D}{r^2}(T^j\phi_1)^2\right)	\right) \dd u'.
\end{nalign}
As before, we can now write 
$
T^j(r\phi_1)|_{\Gamma_f}=(\te)^{j}(r\phi_1)|_{\Gamma_f}+\mathcal{O}(|u|^{-j-s})
$
to find that
\begin{multline*}
\int_{-k}^u r^2(\pu T^j\phi_{1})^2+D|T^j\phi_{1}|^2\dd u'	
\lesssim \int\limits_{\Gamma_f\cap\{-k\leq u'\leq u\}}  -2r^2 \left((\te)^{j+1}\phi_{1}\cdot 2\pv T \phi_{1}	\right) +\mathcal{O}(|u'|^{-2j-3-s})\dd u'.
\end{multline*}
From here, we arrive at the analogue of the estimate \eqref{proof5.4xxx}. We can thus prove the analogue of Proposition~\ref{prop:tlbs3}. 

In a similar fashion, one can then follow the steps of sections~\ref{sec:tl:pvT-T2} and~\ref{sec:tl:limit} to conclude the proof of Theorem~\ref{thm:tlf}.
\end{proof}
\subsection{The case where \texorpdfstring{$r|_{\Gamma_f}$}{r-Gamma-f} diverges (\texorpdfstring{$s\leq 1$}{s leq 1}):}
There are two main differences in the case $s\leq 1$. On the one hand, if we write, as above, 
\begin{equation}
T(r\phi_1)=\underbrace{\te(r\phi_1)}_{\sim |u|^{-2}r^{-1}}+\frac{|u|^{-s}}{1+|u|^{-s}}\underbrace{\pv(r\phi_1)}_{\sim |u|^{-1}r^{-2}},
\end{equation}
then we immediately see that, on $\Gamma_f$, where $r\sim|u|^{1-s}$ if $s\neq 1$, both terms on the RHS can be expected to have the same decay.\footnote{Notice that if $s=1$, the second term on the RHS decays faster by $\log^{-2}|u|$.} This means that we have to be more careful in estimating the boundary terms in the energy estimate. On the other hand, since now $r|_{\Gamma}$ tends to infinity, it will be much more straight-forward to show the existence of the limit $\lim_{u\to-\infty}|u|^2r^2\pv T(r\phi_1)$. 

\subsubsection{Initial/boundary data assumptions and the main theorem}
Let $\Gamma_f$ be as described in section~\ref{sec:tf:setup}, and let $s\leq 1$. We prescribe smooth boundary data $\hat{\phi}_1$ for $\phi_1$ on $\Gamma_f$ which satisfy, for $u\leq U_0< 0$ and $|U_0|$ sufficiently large:
\begin{align}\label{eq:tlf2:boundarydata1}
\left|(\te)^n(r^2\hat{\phi}_1)- \frac{n!\cin}{|u|^{n+1}}\right|=\mathcal{O}(|u|^{-n-1-\epsilon}),&&n=0,\dots,5
\end{align}
for some positive constant $\cin$.
Moreover, we demand, in a limiting sense, that, for all $v$
\begin{equation}\label{eq:tlf2:noincomingradiation}
\lim_{u\to-\infty}\pv^n(r\phi_1)(u,v)=0\quad n=1,\dots 5.
\end{equation} 
Then we obtain, as in the previous section:
\begin{thm}\label{thm:tlf2}
There exists a unique solution $\phi_1$ to eq.\ \eqref{eq:l=1waveequation} in $\mathcal{D}_{\Gamma_f}=\mathcal{M}\cap\{v\geq v_{\Gamma_f}(u)\}$ which restricts correctly to $\hat{\phi}_1$ on $\Gamma_f$, $\phi_1|_{\Gamma_f}=\hat{\phi}_1$, and which satisfies \eqref{eq:tlf2:noincomingradiation}.

Moreover, if $U_0$ is a sufficiently large negative number, then there exists a constant $C=C(\cin)$ (depending only on data) such that $\phi_1$ obeys the following bounds throughout $\mathcal{D}_{\Gamma_f}\cap\{u\leq U_0\}$:
\begin{align}
|r^2\pv T^n(r\phi_1)(u,v)|&\leq \frac{C}{|u|^{n+1}},&&n=0,\dots,4\label{eq:thm:tf1},\\
|T^n(r\phi_1)(u,v)|&\leq \frac{C}{|u|^{n+1}}\max\left(r^{-1},|u|^{-1}\right),&&n=0,\dots,3.\label{eq:thm:tf2}
\end{align}
Finally, along any ingoing null hypersurface $\mathcal{C}_v$, we have
\begin{align}
r^2\pu(r\phi_1)(u,v)&=\mathcal{O}(|u|^{-1})\label{eq:thm:tf3},\\
r^2\pu(r^2\pu\pho))(u,v)&=\begin{cases}	\tilde{C}+\mathcal{O}(|u|^{-\epsilon'}),& \text{if }s<1,\\
											\tilde{C}+\mathcal{O}(\log^{-1} |u|),&\text{if }s=1, \end{cases}\label{eq:tf:thm:limit}
\end{align}
where $\tilde{C}=3\cin$ is determined explicitly by initial data, and $\epsilon'=\min(\epsilon, s, 1-s)$.
\end{thm}
\begin{rem}
We remark that, in the case $s=1$, the fact that the $\mathcal{O}(\log^{-1} |u|)$-term in \eqref{eq:tf:thm:limit} is non-integrable means that Theorem~\ref{thm:nl} cannot be applied directly. Since this is a very specific issue, we make no attempts to fix it in this presentation.
\end{rem}
\subsubsection{Outline of the proof}
\begin{proof}
As before, only a sketch of the proof will be provided.

We cut the data off as before. Let us first show \eqref{eq:thm:tf1} for $n=0$:
\paragraph{Proof of \eqref{eq:thm:tf1} for $\mathbf{n=0}$:}
We follow the proof of Proposition~\ref{prop:tlbs3}. We first need to acquire an estimate for $\pv(r\phi_1)$ on $\Gamma_f$. We assume as a bootstrap assumption that
\begin{equation}
|r^2\pv(r\phi_1)|\leq \frac{C_{\mathrm{BS}}}{|u|}
\end{equation} 
for a suitable constant $C_{\mathrm{BS}}$.
We recall from the energy estimate \eqref{random11}:
\begin{align}\label{random2}
\begin{split}
&\int_{-k}^u r^2(\pu \phi_{1})^2+D|\phi_{1}|^2\dd u'	 \\
\lesssim &\int\limits_{\Gamma_f\cap\{-k\leq u'\leq u\}}  r^2 \left(T\phi_{1}\cdot (T-2\pv) \phi_{1}+\frac{|u'|^{-s}}{1+|u'|^{-s}}\left((\pv \phi_1)^2+\frac{2D}{r^2}\phi_1^2\right)	\right) \dd u'
\end{split}
\end{align}
Note that the $(\pv\phi_1)^2$-terms in the above are potentially dangerous since they could lead to a $C_{\mathrm{BS}}^2$-term, which would make it impossible to improve the bootstrap assumption. However, upon writing again 
\begin{equation}
T\phi_1=\te \phi_1 +\frac{|u|^{-s}}{1+|u|^{-s}}\pv\phi_1,
\end{equation}
we find that that the $(\pv\phi_1)^2$-terms, in fact, appear with a benign sign:
\begin{align*}
&\int_{\Gamma_f}  r^2 \left(T\phi_{1}\cdot (T-2\pv) \phi_{1}+\frac{|u'|^{-s}}{1+|u'|^{-s}}\left((\pv \phi_1)^2+\frac{2D}{r^2}\phi_1^2\right)	\right) \dd u'\\
=&\int_{\Gamma_f}   r^2 \left(		(T\phi_1)^2-\frac{|u'|^{-s}}{1+|u'|^{-s}}(\pv\phi_1)^2-2\pv\phi_1\te\phi_1+\frac{|u'|^{-s}}{1+|u'|^{-s}}\frac{2D}{r^2}\phi_1^2\right)\dd u'\\
\leq& \int_{\Gamma_f}   r^2 \left(		(\te\phi_1)^2+2|\pv\phi_1||\te\phi_1|\frac{1+2|u'|^{-s}}{1+|u'|^{-s}}+\frac{|u'|^{-s}}{1+|u'|^{-s}}\frac{2D}{r^2}\phi_1^2\right)
				+r^2 \frac{|u'|^{-2s}}{(1+|u'|^{-s})^2}	(\pv\phi_1)^2					\dd u'\\
\lesssim& \int_{\Gamma_f} r^2 \left(\frac{(\cin)^2}{r^4|u'|^4}+\frac{(\cin)^2}{r^6|u'|^{2+2s}} +\frac{\cin(\cin+C_{\mathrm{BS}})}{r^5|u'|^2}\left(\frac{1}{|u'|}+\frac{1}{r|u'|^{s}}\right)+\frac{(\cin)^2}{r^6|u'|^{2+s}}\right)
				+r^2 \frac{(\cin+C_{\mathrm{BS}})^2}{r^6|u'|^{2+2s}}				\dd u' ,
\end{align*}
where we used the boundary data assumption and the bootstrap assumption in the last estimate. Using now the fact that $|u|^{1-s}\lesssim r$ if $s\neq 1$ (or $\log |u|\lesssim r$ if $s=1$), as well as the fundamental theorem of calculus and the Cauchy--Schwarz inequality, combined with the energy estimate \eqref{random2}, we obtain that
\begin{align}
r\phi_1^2\lesssim \frac{\cin(\cin+C_{\mathrm{BS}})}{r^3|_{\Gamma_f}|u|^2}
\end{align}
Importantly, $C_{\mathrm{BS}}$ does not appear quadratically in the above estimate. Plugging this bound into the wave equation \eqref{eq:l=1waveequation} and integrating in $u$, we obtain that
\begin{equation}
\left|r^2\pv(r\phi_1)|_{\Gamma_f}\right| \lesssim \frac{\sqrt{C_{\mathrm{BS}}^2+C_{\mathrm{BS}}\cin}}{|u|^2}.
\end{equation}

Having obtained a bound for the boundary term, we can now, as in the proof of Proposition~\ref{prop:tlbs3}, use the approximate conservation law \eqref{eq:l=1cons.law.in.u} to close the bootstrap argument for $\pv(r\phi_1)$. Indeed, we can obtain a bound for $\Phi$ (similarly to how we obtained \eqref{proof5.4xxx}) and then use the fact that, by integrating the bootstrap assumption from $\Gamma_f$, we have
\begin{equation}
|r\phi_1|\leq \frac{\cin+C_{\mathrm{BS}}}{r|_{\Gamma_f}|u|}.
\end{equation}
In view of $\log |u|\lesssim r|_{\Gamma_f}$, this decays faster than $r^2\pv(r\phi_1)$. Therefore, the bound for $\Phi$ immediately translates into a bound for $r^2\pv\pho$. This closes the bootstrap argument.

\textbf{Proof of \eqref{eq:thm:tf1} for $\mathbf{n>0}$:}
Having proved \eqref{eq:thm:tf1} for $n=0$, we now outline the proof for $n>0$. In fact, the only thing that changes is that, in the energy equality \eqref{random2}, we now need to express $T^j\phi_1$ in terms of $(\te)^j\phi_1$ for $j>1$, which leads to more "error" terms.
For instance, we have
\begin{equation}
T^2\phi_1=(\te)^2\phi_1+2\frac{|u|^{-s}}{1+|u|^{-s}}\pv T\phi_1+T\left(\frac{|u|^{-s}}{1+|u|^{-s}}\right)\pv\phi_1-\frac{|u|^{-2s}}{(1+|u|^{-s})^2}\pv^2\phi_1.
\end{equation}
We  have already obtained estimates for the last two terms. 
Moreover, we can estimate the first term above from the boundary data assumptions, and the second one via a bootstrap assumption on $\pv T(r\phi_1)$. 
Plugging these estimates into \eqref{random11} for $j=2$ then improves the bootstrap assumption.

We leave the cases $j>2$ to the reader. 
(Notice that when e.g.\ expressing $T^4\phi_1$ in terms of $(\te)^4 \phi_1$, there will also be, for instance, a term containing $\pv^4\phi_1$. 
We can estimate this by simply commuting the wave equation twice with $\pv$ and appealing to the proof of \eqref{eq:thm:tf1} for $n=0$. 
The other terms can be dealt with similarly.)

\paragraph{Proof of \eqref{eq:thm:tf2}:}
We can obtain the estimates \eqref{eq:thm:tf2} for $n\leq 3$ by using the wave equation as in \eqref{eq:waveequationtoestimatephi} and the already obtained bounds \eqref{eq:thm:tf1}.

\paragraph{Proof of \eqref{eq:thm:tf3}}
The proof of \eqref{eq:thm:tf3} is straight-forward. We simply write:
\begin{equation}
r^2\pu(r\phi_1)=r^2T(r\phi_1)-r^2\pv(r\phi_1).
\end{equation}

\paragraph{Proof of \eqref{eq:tf:thm:limit}:}
Finally, we prove \eqref{eq:tf:thm:limit}. As in the proof of Proposition~\ref{prop:limit}, we will first compute the limit $\lim_{u\to-\infty}|u|^2r^2\pv T(r\phi_1)$. 

In view of the approximate conservation law \eqref{eq:l=1cons.law.in.u} and the fact that $r|_{\Gamma_f}$ tends to infinity, we have that
\begin{equation}
\lim_{u\to-\infty}|u|^2r^2\pv T(r\phi_1)(u,v)=\lim_{u\to-\infty}|u|^2r^2\pv T(r\phi_1)(u,v_{\Gamma_f}(u)).
\end{equation}
We estimate $\pv T(r\phi_1)|_{\Gamma_f}$ as follows. Integrating the $T$-commuted \eqref{2}, we obtain that
\begin{equation}
r^2\pv T(r\phi_1)+r^2T \phi_1=\mathcal{O}\left(\frac{r}{|u|^3}+\frac{1}{r|u|^2}\right),
\end{equation}
from which we read off that
\begin{equation}
|u|^2r^2\pv T(r\phi_1)(u,v)=-\cin+\begin{cases} \mathcal{O}\left(\frac{1}{\log |u|}\right),&s=1,\\
												\mathcal{O}\left(\frac{1}{|u|^{s}}+\frac{1}{|u|^{1-s}}+\frac{1}{|u|^{\epsilon}}\right),& s<1,
\end{cases}
\end{equation}
and, in particular, that 
\begin{align}
\mathcal{L}:=\lim_{u\to-\infty}|u|^2r^2\pv T(r\phi_1)(u,v)=-\cin.
\end{align}
Here, we used that 
\begin{align*}
r^2T\phi_1=T(r^2\phi_1)=\te(r^2\phi_1)+\frac{|u|^{-s}}{1+|u|^{-s}}\pv (r^2\phi_1)
\end{align*}
and the fact that, in view of \eqref{2}, the second term above decays faster than the first.

Similarly, we find that
\begin{align}
\lim_{u\to-\infty}|u|^{j+1}r^2\pv T^j(r\phi_1)(u,v)=-j!\cin
\end{align}
for $j\leq 3$.
We can now compute, exactly as in the proof of Proposition~\ref{prop:limit}, the expressions \eqref{proof5.8y}, \eqref{proof5.8yy} and \eqref{proof5.8yyy}, from which it follows, using the identity \eqref{eq:tl:pupu}, that
\begin{equation}
\lim_{u\to-\infty}r^2\pu(r^2\pu(r\phi_1))(u,v)=-3\mathcal{L}=3\cin.
\end{equation}
This concludes the proof.
\end{proof}

\newpage
\part{Generalising to all \texorpdfstring{$\ell\geq 0$}{L}.}\label{part2}
Having understood the case $\ell=1$ in detail in the previous sections \S \ref{sec:nl}--\S \ref{sec:timelikenonconstantr}, we now want to analyse the general case. As explained in \S\ref{sec:intro:gouda}, this second part of the paper can be understood mostly independently of part~\ref{part1}.

As a zeroth step, we need to establish higher $\ell$-analogues of the approximate conservation laws \eqref{eq:l=1cons.law.in.u}, \eqref{eq:l=1cons.law.in.v}. 
This is achieved in \S \ref{sec:generalNP}. 
We then treat the case of timelike boundary data in \S \ref{sec:general:timelike}, restricting the presentation however to cases of hypersurfaces of constant area radius.
 Then, we treat the case of characteristic initial data in \S \ref{sec:general:null}  and~\S \ref{sec:moreg:null}. 

The sections \S\ref{sec:general:timelike} and \S\ref{sec:general:null} are similar in spirit to \S\ref{sec:timelikeconstantR} and \S\ref{sec:nl}, respectively. (The reasons for the reversed order of the sections are solely of presentational, not of mathematical nature.)
On the other hand,  \S\ref{sec:moreg:null} follows a different mathematical structure than \S\ref{sec:nl}: While the methods of \S \ref{sec:nl} and \S\ref{sec:general:null} can only treat data that decay like $r\phi_\ell\lesssim|u|^{-\ell}$,  the approach of \S\ref{sec:moreg:null} allows us to also treat slowly decaying (and even growing) initial data.

\section{The higher-order Newman--Penrose constants}\label{sec:generalNP}
In this section, we derive higher-order conservation laws and define the Newman--Penrose constants associated with them. 
\subsection{Generalising the approximate conservation law (\ref{eq:l=1cons.law.in.u})}
In order to generalise the approximate conservation law in $u$ \eqref{eq:l=1cons.law.in.u},  we first  require a general formula for commutations of the wave equation with $[r^2\pv]^N$:
\begin{prop}\label{prop:gl:dvcommute}
Let $\phi$ be a smooth solution to $\Box_g\phi=0$, and let $N\geq 0$. Then $\phi$ satisfies
\begin{multline}\label{eq:gl:dvcommute}
\pu\pv[r^2\pv]^N(r\phi)=-\frac{2DN}{r}\pv[r^2\pv]^N(r\phi)\\+
\sum_{j=0}^N\frac{D}{r^2}(2M)^j\left(a_j^N+b_j^N\slashed{\Delta}_{\mathbb{S}^2}-c_j^N\cdot\frac{2M}{r}\right)[r^2\pv]^{N-j}(r\phi),
\end{multline}
where the constants $a_j^N, b_j^N$ and $c_j^N$ are given explicitly by
\begin{align}
a_j^N&=(2^j-1)\binom{N}{j}+(2^{j+2}-2)\binom{N+1}{j+2},\\
b_j^N&=\binom{N}{j},\\
c_j^N&=2^j\binom{N}{j}+2^{j+2}\binom{N+1}{j+2},
\end{align}
and we use the convention that $\binom{N}{j}=0$ if $j>N$.
\end{prop}
\begin{proof}
A proof is given in the appendix~\ref{A1}.
\end{proof}
Notice that, in particular, $a_0^N=N(N+1)$ and $b_0^N=1$. Hence, if we restrict to solutions supported on the $\ell=L$-angular frequencies, and consider \eqref{eq:gl:dvcommute} for $N=L$, there will be a cancellation in the highest-order derivatives. One can then iteratively subtract \eqref{eq:gl:dvcommute} for $N<L$, multiplied with a suitable constant, to obtain an approximate conservation law. This is done in
\begin{cor}\label{cor:gl}
Let $\phi=\sum_{|m|\leq L}\phi_{Lm}\cdot Y_{Lm}$ be a smooth solution to $\Box_g\phi=0$ supported on the angular frequencies $\ell=L\geq 0$. In what follows, we shall suppress the $m$-index. Let $N\geq 0$, and define, for $j\leq N$,
\begin{equation}
\tilde{a}_j^{N,L}:=a_j^N-b_j^N\cdot L(L+1),
\end{equation}
and let $\{x_i^{(L)}\}_{0\leq i\leq L}$ be a set of constants with $x_0^{(L)}=1$. 
Then $\phi$ satisfies
\begin{multline}
\pu\pv\left([r^2\pv]^L(r\phi_L)+\sum_{i=1}^L(2M)^ix_i^{(L)}[r^2\pv]^{L-i}(r\phi_L)\right)\\
=-\frac{2LD}{r}\pv[r^2\pv]^L(r\phi_L)-\sum_{i=1}^L(2M)^ix_i^{(L)}\frac{2(L-i)D}{r}\pv[r^2\pv]^{L-i}(r\phi_L)\\
+\sum_{j=0}^L\frac{D}{r^2}(2M)^j[r^2\pv]^{L-j}(r\phi_L)\sum_{i=0}^j\left( x_i^{(L)}\tilde{a}_{j-i}^{L-i,L}-x_i^{(L)} c_{j-i}^{L-i}\cdot \frac{2M}{r} \right).
\end{multline}
\end{cor}
\begin{proof}This is a straight-forward computation.\end{proof}
\begin{defi}[The generalised higher-order future Newman--Penrose constant]
Let $\phi$ be as in Corollary~\ref{cor:gl}, and define the constants $x_i^{(L)}$ for $1\leq i\leq L$ as follows: 
\begin{equation}
x_i^{(L)}:=-\dfrac{1}{\tilde{a}_0^{L-i,L}}\sum_{k=0}^{i-1}\tilde{a}_{i-k}^{L-k,L}x^{(L)}_k.
\end{equation}
This is well-defined since $\tilde{a}_0^{L-i,L}\neq 0$ for $i>0$ and since $x^{(L)}_0=1$.  
We further denote
\begin{equation}\label{eq:NPfuture}
\Phi_{L}:=[r^2\pv]^L(r\phi_{L})+\sum_{i=1}^L(2M)^ix_i^{(L)}[r^2\pv]^{L-i}\phi_{L},
\end{equation}
and define, for any smooth function $f(r)=o(r^3)$, the generalised higher-order Newman--Penrose constant according to
\begin{equation}
\IL{L}[\phi](v):=\lim_{v\to\infty}f\pv\Phi_L(u,v).
\end{equation}
\end{defi}
\begin{cor}
The quantity $\Phi_{L}$ defined above satisfies the following approximate conservation law:
\begin{equation}\label{eq:gl:approx.u}
\pu(r^{-2L}\pv\Phi_L)
=\sum_{j=0}^L\frac{D}{r^{3+2L}}(2M)^{j+1}[r^2\pv]^{L-j}(r\phi_{L})\left(2(j+1)x^{(L)}_{j+1}-\sum_{i=0}^j x_i^{(L)} c_{j-i}^{L-i}\right) .
\end{equation}
Here, we used the notation that $x_i^{(L)}=0$ for $i\geq L$.
\end{cor}
In particular, under suitable assumptions on $\phi$, the generalised higher-order N--P constant defined above is conserved along $\mathcal{I}^+$:
\begin{equation}
 \IL{L}[\phi](v) \equiv \IL{L}[\phi] .
\end{equation}
\begin{rem}
It is helpful to keep in mind that the quantity $\pv\Phi_L$ can always be written as
\begin{equation}\label{eq:NPpolynomial}
\pv\Phi_L=\pv(r^2p_1(r^2p_2\pv(\dots r^2p_L\pv(r\phi_L)\dots))),
\end{equation}
where the $p_i$ are polynomials in $1/r$. Intuitively, this indicates that one should typically be able to recover an estimate for $r\phi_L$ from $\Phi_L$ by simply integrating $L$ times. However, we will never need to make use of the form \eqref{eq:NPpolynomial} in this paper.
\end{rem}
\subsection{Generalising the approximate conservation law (\ref{eq:l=1cons.law.in.v})}
We follow a similar procedure to derive approximate conservation laws in the $\pv$-direction. We have
\begin{prop}\label{prop:gl:ducommute}
Let $\phi$ be a smooth solution to $\Box_g\phi=0$, and let $N\geq 0$. Then $\phi$ satisfies
\begin{multline}\label{eq:gl:ducommute}
\pv\pu[r^2\pu]^N(r\phi)=\frac{2DN}{r}\pu[r^2\pu]^N(r\phi)\\+
\sum_{j=0}^N\frac{D}{r^2}(2M)^j\left(\underline{a}_j^N+\underline{b}_j^N\slashed{\Delta}_{\mathbb{S}^2}-\underline{c}_j^N\cdot\frac{2M}{r}\right)[r^2\pu]^{N-j}(r\phi),
\end{multline}
where $\underline{a}_j^N=(-1)^ja_j^N$, $\underline{b}_j^N=(-1)^jb_j^N$ and $\underline{c}_j^N=(-1)^jc_j^N$. 
\end{prop}
\begin{proof}
The proof follows along the same steps as the one of Proposition~\ref{prop:gl:dvcommute}. See the appendix~\ref{A1} for details.
\end{proof}
\begin{defi}[The generalised higher-order past Newman--Penrose constant]
Let $\phi$ be as in Corollary~\ref{cor:gl}, let $\tilde{\underline{a}}_j^{N,L}:=(-1)^j\tilde{a}_j^{N,L}$, let $x^{(L)}_0=1$, and define, for $1\leq i\leq L$,
\begin{equation}
\underline{x}_i^{(L)}=-\dfrac{1}{\tilde{\underline{a}}_0^{L-i,L}}\sum_{k=0}^{i-1}\tilde{\underline{a}}_{i-k}^{L-k,L}\underline{x}^{(L)}_k.
\end{equation}
We then denote
\begin{equation}\label{eq:NPpast}
\underline{\Phi}_{L}:=[r^2\pu]^L(r\phi_{L})+\sum_{i=1}^L(2M)^i\underline{x}_i^{(L)}[r^2\pu]^{L-i}(r\phi_{L}),
\end{equation}
and, moreover, define for any smooth function $f(r)=o(r^3)$ the generalised higher-order Newman--Penrose constant according to
\begin{equation}
\ILp{L}[\phi](u):=\lim_{u\to-\infty}f\pu\underline{\Phi}_L(u,v).
\end{equation}
\end{defi}
\begin{cor}
The quantity $\underline{\Phi}_{L}$ defined above satisfies the following approximate conservation law:
\begin{equation}\label{eq:gl:approx.v}
\pv(r^{-2L}\pu\underline{\Phi}_L)
=\sum_{j=0}^L\frac{D}{r^{3+2L}}(2M)^{j+1}[r^2\pu]^{L-j}(r\phi_{L})\left(-2(j+1)\underline{x}^{(L)}_{j+1}-\sum_{i=0}^j \underline{x}_i^{(L)} \underline{c}_{j-i}^{L-i}\right) .
\end{equation}
Here, we used the notation that $\underline{x}_i^{(L)}=0$ for $i\geq L$.
\end{cor}
In particular, under suitable assumptions on $\underline{\Phi}$, the generalised higher-order N--P constant defined above is conserved along $\mathcal{I}^-$:
\begin{equation}
 \ILp{L}[\phi](u) \equiv \ILp{L}[\phi] .
\end{equation}

\section{Boundary data on a timelike hypersurface \texorpdfstring{$\Gamma_R$}{Gamma-R}}\label{sec:general:timelike}
Equipped with the approximate conservation laws \eqref{eq:gl:approx.u}, \eqref{eq:gl:approx.v}, we now generalise the results of \S \ref{sec:timelikeconstantR}. More precisely, we construct higher $\ell$-mode solutions to \eqref{waveequation} (and derive estimates for them) that arise from polynomially decaying boundary data on a timelike hypersurface $\Gamma_R$ of constant area radius $r=R$ and the no incoming radiation condition. In particular, the present section contains the proof of Theorem~\ref{thm:intro:gtl}. The generalisation to boundary data on hypersurfaces on which $r$ is allowed to vary then proceeds as in \S\ref{sec:timelikenonconstantr} and is left to the reader.

Throughout the rest of this section, we shall assume that $R>2M$ is a constant and that $\phi$ is a solution to \eqref{waveequation} supported on a single angular frequency $(L,m)$, with $|m|\leq L$ and $L\geq0$. In the usual abuse of notation of \S\ref{sec:decompositionintoYlm}, we omit the $m$-index, that is, we write $\phi=\phi_{Lm}Y_{Lm}=\phi_L Y_{Lm}$.
\subsection{Initial/boundary data assumptions}\label{sec:gtl:ass}
 We prescribe smooth boundary data $\hat{\phi}_L$ on $\Gamma_R=\mathcal{M}_M\cap\{v=v_R(u)\}$ that satisfy, for $u\leq U_0<0$ and $|U_0|$ sufficiently large, the following upper bounds:
\begin{align}\label{eq:gtl:boundarydata1}
\left|T^n(r\hat{\phi}_L)\right|&\leq \frac{n! \cin}{R^L|u|^{n+1}},&&n=0,1,\dots,N+1,\\
\left|T^n(r\hat{\phi}_L-|u|T(r\hat{\phi}_L))\right|&\leq \frac{\ce}{R^L|u|^{n+1+\epsilon}},&& n=0,\dots,N'+1\label{eq:gtl:boundarydata2}
\end{align}
for some positive constants $\cin$, $\ce$, $\epsilon\in(0,1)$ and $N,N'\geq 0$  integers, and which also satisfy the lower bound
\begin{equation}
\left|T(r\hat{\phi}_L)\right|\geq \frac{\cin}{2R^L|u|^{n+1}}>0.\label{eq:gtl:boundarydata3}
\end{equation}
Moreover, we demand, in a limiting sense, that, for all $v$,
\begin{equation}\label{eq:gtl:noincomingradiation}
\lim_{u\to-\infty}\pv^n(r\phi_L)(u,v)=0,\quad n=1,\dots N+1.
\end{equation} 
This latter condition is to be thought of as the no incoming radiation condition.
\subsection{The main theorem (Theorem~\ref{thm:gtl})}
The main result of this section is
\begin{thm}\label{thm:gtl}
Let $R>2M$. Then there exists a unique solution $\phi_L\cdot Y_{Lm}$ to eq.\ \eqref{waveequation} in $\mathcal{D}_{\Gamma_R}=\mathcal{M}\cap\{v\geq v_R(u)\}$ that restricts correctly to $\hat{\phi}_L \cdot Y_{Lm}$ on $\Gamma_R$, $\phi_L|_{\Gamma_R}=\hat{\phi}_L$, and that satisfies \eqref{eq:gtl:noincomingradiation}.

Moreover, if $U_0$ is a sufficiently large negative number, then there exists a constant $C=C(2M/R,\cin,L)$, depending only on data, such that $\phi_L$ obeys the following bounds throughout $\mathcal{D}_{\Gamma_R}\cap\{u\leq U_0\}$:
\begin{align}\label{eq:thm:gtlmain}
\left|[r^2\pv]^{L-j}T^i(r\phi_L)\right|\leq \frac{C}{|u|^{i+1}}\min(r,|u|)^{-\min{(\tilde{j},N-i)}}
\end{align}
for all $j=-1,\dots, L$ and for all $i=0,\dots, N$,  and $\tilde{j}:=\max(j,0)$.

Finally, if  $N-2\geq N'\geq L+1$, then we have along any ingoing null hypersurface $\mathcal{C}_v$:
\begin{align}
[r^2\pu]^{L-j}(r\phi_L)(u,v)&=\mathcal{O}(r^{-1-j}),&&j=0,\dots, L,\label{eq:thm:gtl:notlimit}\\
[r^2\pu]^{L+1}(r\phi_L)(u,v)&=\tilde{C}+\mathcal{O}(r^{-1}+|u|^{-\epsilon})&&\label{eq:thm:gtl:limit}
\end{align}
for some constant $\tilde{C}$ which can be shown to be non-vanishing if $R/2M$ is sufficiently large.
\end{thm}
\begin{rem}
A similar result holds true for more general timelike hypersurfaces $\Gamma_f$ (on which, in particular, $r$ is allowed to tend to infinity) as discussed in \S \ref{sec:timelikenonconstantr}. We leave the proof to the interested reader.
\end{rem}
\subsection{Overview of the proof}\label{sec:tl:overviewofproof}
\newcommand{\pvx}[1]{[r^2\pv]^{#1}}
We shall first give an overview over the proof of Theorem~\ref{thm:gtl}. 
 \renewcommand{\labelenumi}{\Roman{enumi}}
\begin{enumerate}
\item 
In a first step, we construct a sequence of smooth compactly supported data $\hat{\phi}^{(k)}_L$ as in \S \ref{sec:tl:cutoffdata}, which lead to solutions $\phi_L^{(k)}$ in the sense of Prop.~\ref{prop:existence:mixedboundary}. The purpose of this is that we will then be able to use the method of continuity (i.e.\ bootstrap arguments) on these finite solutions $\phi_L^{(k)}$ .
\item 
We then assume (in the form of a bootstrap assumption) that the estimate $\left|r^2\pv(r\phi^{(k)}_L)|_{\Gamma_R}\right|\leq \frac{C_{\mathrm{BS}}}{R^{L-1}|u|}$ holds on $\Gamma_R$. 
An application of an energy estimate will imply that $\left|r^2\pv(r\phi^{(k)}_L)\right|\leq C'(\mathrm{data})\cdot \frac{\sqrt{ C_{\mathrm{BS}}}}{R^L|u|}$ and, thus, improve this assumption.
 From this, we then inductively integrate equation \eqref{eq:gl:dvcommute} to obtain estimates for the boundary terms $\left|[r^2\pv]^{L-j}(r\phi^{(k)}_L)|_{\Gamma_R}\right|$, $j=0,\dots L$. 
 The same estimates hold upon commuting with $T^i$.
\item 
In a third step, we assume decay on $[r^2\pv]^{L}(r\phi^{(k)}_L)$ and integrate the approximate conservation law \eqref{eq:gl:approx.u} in $u$ and in $v$ (the integration in $v$ from $\Gamma_R$ outwards is why we need the estimates on the boundary terms from step II) to improve this decay, exploiting $2M/R$ as a small parameter. 
(We recall from \S \ref{sec:tl:removelargeness} that any smallness assumptions on $2M/R$ can be recovered by replacing the bootstrap argument with a Gr\"onwall argument.)
Integrating this estimate for $[r^2\pv]^{L}(r\phi^{(k)}_L)$ then $j$ times from $\Gamma_R$ and also commuting with $T$ establishes the following estimates:
\begin{align}\label{eq:gtl:overview3}
\left|[r^2\pv]^{L-j}T^i(r\phi^{(k)}_L)\right|\leq \frac{C}{|u|^{i+1}}R^{-\max{(j,0)}}
\end{align}
for $i=0,\dots, N$, $j=-1,0,\dots, L$, and for $C\neq C(k)$ a constant.
\item 
In a fourth step, we adapt the methods of steps II and III as in \S \ref{sec:tl:pvT-T2} to obtain estimates on the boundary terms $[r^2\pv]^{L-j}T^i (r\phi^{(k)}_L-|u| T(r\phi^{(k)}_L))|_{\Gamma_R}$ and, from these, establish the auxiliary estimates (modulo corrections arising from the cut-off terms, cf.\ \eqref{eq:proof:gtl:claim5}):
\begin{align}\label{eq:gtl:overview4}
\left|[r^2\pv]^{L-j}T^i\left(r\phi^{(k)}_L-|u|T(r\phi^{(k)}_L)\right)\right|\leq \frac{C}{|u|^{i+1+\epsilon}}R^{-\max(j,0)}
\end{align}
for $i=0,\dots,N'$ and $j=-1,0, \dots, L$. 
\item 
In a fifth step, we show, as in \S \ref{sec:tl:limit}, that the solutions $\phi_L^{(k)}$  tend uniformly to a limiting solution $\phi_L$, which still satisfies the estimates \eqref{eq:gtl:overview3} and \eqref{eq:gtl:overview4} above.
\item
In a sixth step, we use the estimate \eqref{eq:gtl:overview3}, together with the identities
\begin{multline}\label{eq:gtl:overview6}
\sum_{j=0}^N\left(a_j^N-L(L+1)b_j^N-c_j^N\frac{2M}{r}\right)\pvx{N-j}T^i(r\phi_L)\\
=[r^2\pv]^{N+1} T^{i+1}(r\phi_L)
+\frac{2D(N+1)}{r}\pvx{N+1}T^i(r\phi_L)-\frac{1}{r^2}\pvx{N+2}T^i(r\phi_L)
\end{multline}
implied by \eqref{eq:gl:dvcommute}, to obtain the improved estimates \eqref{eq:thm:gtlmain}, i.e.\ to convert the $R$-weights of \eqref{eq:gtl:overview3} into $r$-weights, using an "upwards-downwards induction".
\item
In a seventh step, we use equation \eqref{eq:gtl:overview6} to obtain a \textit{lower bound} for $\pvx{L}T(r\phi_L)|_{\Gamma}$ on $\Gamma$, provided that the lower bound \eqref{eq:gtl:boundarydata3} for $T(r\phi_L)|_{\Gamma}$ on data is specified. Using the estimate \eqref{eq:gtl:overview3} with $j=-1$, we can then obtain a global lower bound for $\pvx{L}T(r\phi_L)$, provided that  $R/2M$ is sufficiently large. 

Furthermore, and independently of this lower bound, we can use the auxiliary estimate \eqref{eq:gtl:overview4} to show that the following limit exists and is independent of $v$:
\begin{equation}
\lim_{u\to-\infty}|u|^2[r^2\pv]^LT(r\phi_L)(u,v)=:\mathcal{L}.
\end{equation}
By the lower bounds obtained before, this limit is non-vanishing.
\item
Finally, we prove \eqref{eq:thm:gtl:notlimit} and \eqref{eq:thm:gtl:limit} by writing
\begin{equation}
[r^2\pu]^{L+1-j}(r\phi_L)=[r^2T-r^2\pv]^{L+1-j}(r\phi_L),
\end{equation}
and by expressing each term in the expansion of the above expression in terms of $\mathcal{L}$, using the relations \eqref{eq:gtl:overview4} and \eqref{eq:gtl:overview6}.
\end{enumerate}

\subsection{Proof of Theorem~\ref{thm:gtl}}
We now prove Theorem~\ref{thm:gtl}, following the structure outlined above. The proof will be self-contained, with the exceptions of steps IV and V, for which we will refer to \S \ref{sec:timelikeconstantR} for details.

\begin{proof}
Throughout this proof, $C$ shall denote a constant that depends only on $\cin,\ce$, $ 2M/R, M, N, N', L$ (and, in particular, not on $k$) and can be bounded independently of $R$ for sufficiently large $R$. Moreover, $C$ is  allowed to vary from line to line. We will also assume $U_0$ to be sufficiently large, where this largeness, again, only depends on data, i.e.\ on $C$.
\subsubsection*{Step I: Cutting off the data}
We let $(\chi_k(u))_{k\in\mathbb{N}}$ be a sequence of smooth cut-off functions such that
\begin{align*}
\chi_k(u)=\begin{cases}1,&u\geq -k+1,
				\\0,&u\leq -k,
\end{cases}
\end{align*}
and cut-off the the highest-order derivative: $\chi_kT^{N+1}\hat{\phi}_L$. We then define, as in \S \ref{sec:tl:cutoffdata}, $\hat{\phi}_L^{(k)}$ to be the $N+1$-th $T$-integral of $\chi_kT^{N+1}\hat{\phi}_L$ from $-\infty$. Then $\hat{\phi}_L^{(k)}$  satisfies the following bounds:
\begin{align}\label{eq:gtl:boundarydata1cutoff}
\left|T^n(r\hat{\phi}^{(k)}_L)\right|&\leq \frac{n! \cin}{R^L|u|^{n+1}},&&n=0,1,\dots,N+1,\\
\left|T^n\left(r\hat{\phi}^{(j)}_L-|u|T(r\hat{\phi}^{(k)}_L)\right)\right|&\leq \frac{\ce}{R^L|u|^{n+1+\epsilon}}+C\theta_k\cdot\frac{\cin}{Rk^{n+1}},&& n=0,1,\dots,N'+1,\label{eq:gtl:boundarydata2cutoff}
\end{align}
where $\theta_k$ equals 1 if $u\geq -k$, and 0 otherwise. 

The boundary data $\hat{\phi}^{(k)}_L$, combined with the no incoming radiation condition \eqref{eq:gtl:noincomingradiation}, lead to unique solutions $\phi_L^{(k)}$, which vanish identically for $u\leq -k$, and which solve the finite initial/boundary value problem (in the sense of Prop.~\ref{prop:existence:mixedboundary}) where $\hat{\phi}^{(k)}_L$ is specified on $\Gamma_R$, and where $r\phi_L^{(k)}=0$ on $\{u=-k\}$. 

In steps II--IV below, we will show uniform-in-$k$ estimates on these solutions $\phi_L^{(k)}$, temporarily dropping the superscript $(k)$ and denoting them simply by $\phi_L$. We will re-instate this superscript in step V, where we will show that the solutions $\phi_L^{(k)}$ tend to a limiting solution as $k\to-\infty$.

\subsubsection*{Step II: Estimates on the boundary terms}
	\begin{claim}\label{claim1}
	Let $U_0$ be a sufficiently large negative number.	Then there exist constants $B^{(i)}$ such that
	\begin{equation}
	\left|r^2\pv T^i(r\phi_L)\right|(u,v_{\Gamma_R}(u))\leq\frac{B^{(i)}}{R^{L-1}|u|^{i+1}}
		\label{eq:proof:gtl:claim1}
	\end{equation} for $i=0,\dots, N$ and for all $u\leq U_0$.
	\end{claim}
\begin{proof}
We fix $i\leq N$,  and assume \eqref{eq:proof:gtl:claim1} with $B^{(i)}$ sufficiently large as a bootstrap assumption. 

Recall the definition of the energy current \eqref{eq:def:J}. We apply the divergence theorem in the form \eqref{eq:divergencetheorem} to the identity (we abuse notation and omit the $Y_{Lm}$)
\begin{equation}
\div J^T[T^i\phi_L]=0
\end{equation}
as in equation \eqref{random1} in order to obtain
\begin{align}\begin{split}
&\int_{-k}^u r^2 \left(\pu T^i \phi_1\right)^2(u',v)\dd u'\\
 \leq &\int_{\Gamma_R\cap \{-k\leq u' \leq u\}}2r^2 \left(\left|T^{i+1}\phi_L\right|^2+2\left|T^{i+1}\phi_L\cdot \pv T^i\phi_L\right|\right)(u', v_{\Gamma_R}(u'))\dd u'.
\end{split}\end{align}
We estimate, on $\Gamma_R$:
\[\left|r\pv T^i\phi_L\right|=\left|\pv T^i(r\phi_L)-DT^i\phi_L\right|\leq \frac{B^{(i)}}{R^{2+L-1}|u|^{i+1}}+\frac{(i+1)!\cin}{R^{L+1}|u|^{i+1}}.\]
(In the above, we used the bootstrap assumption \eqref{eq:proof:gtl:claim1} and the boundary data assumption \eqref{eq:gtl:boundarydata1cutoff}.)
We thus obtain that
\begin{align*}
&\int_{-k}^u r^2 \left(\pu T^j \phi_1\right)^2(u',v)\dd u'\\
&\leq \int_{\Gamma_R\cap \{-k\leq u' \leq u\}}2R^2
\cdot \left(\left(\frac{(i+2)!\cin}{R^{L+1}|u|^{i+2}}\right)^2\right.+\left.2\frac{(i+2)!\cin\left(B^{(i)}+(i+1)!\cin\right)}{R^{2L+1}|u|^{2i+3}}\right)(u', v_{\Gamma_R}(u'))\dd u'\\
&\leq C\cdot \frac{B^{(i)}}{R^{2L+1}|u|^{2i+2}}
\end{align*}
for some constant $C$ as described in the beginning of the proof. 
We now apply the fundamental theorem of calculus and the Cauchy--Schwarz inequality to obtain
\begin{align*}
T^i\phi_L(u,v)=\int_{-k}^u \pu T^i\phi_{L}(u',v)\dd u'\leq \frac{1}{\sqrt{Dr}}\cdot\frac{\sqrt{CB^{(i)}}}{R^{L+\frac12}|u|^{i+1}}.
\end{align*}
Inserting this bound into \eqref{eq:gl:dvcommute} with $N=0$, and integrating the latter from $u=-k$, we find that 
\begin{align}\label{eq:claim1prr}
\left|\pv T^i(r\phi_L)\right|(u,v_{\Gamma_R}(u))\leq \int_{-k}^u\frac{\sqrt{CB^{(i)}}}{R^{L+\frac12}|u'|^{i+1}}\left(L(L+1)+\frac{2M}{r}\right)\frac{D}{r^{\frac32}}\dd u'\leq \frac{C\sqrt{B^{(i)}}}{R^{L+1}|u|^{i+1}}.
\end{align}
This improves the bootstrap assumption \eqref{eq:proof:gtl:claim1}, provided that $B^{(i)}$ is chosen sufficiently large.
\end{proof}
	\begin{claim}\label{claim2}
	Let $U_0$ be a sufficiently large negative number. Then there exists a constant $C$ such that 
	\begin{equation}
	\left|[r^2\pv]^{L-j}(r\phi_L)\right|(u,v_{\Gamma_R}(u))\leq\frac{C}{R^{j}|u|^{i+1}}\label{eq:proof:gtl:claim2}
	\end{equation}
	for $i=0, \dots, N$ and $j=0,\dots, L-1$, and for all $u\leq U_0$.
	\end{claim}
\begin{proof}
In the proof of the previous claim, we have in fact shown that (cf.\ \eqref{eq:claim1prr})
\begin{equation}
\left|\pv T^i(r\phi_L)\right|(u,v)\leq \frac{C}{R^{L+\frac12}\sqrt{r}|u|^{i+1}}
\end{equation}
for all $v\geq v_{\Gamma_R}(u)$.
Let us assume inductively that 
\begin{equation}\label{induc1}
\left|\pv[r^2\pv]^n T^i(r\phi_L)\right|(u,v)\leq \frac{Cr^{-\frac12+n}}{R^{L+\frac12}|u|^{i+1}}
\end{equation}
for some fixed $n<\max(L-2,1)$ and for all $u\leq U_0, v\geq v_{\Gamma_R}(u)$, noting that we have already established the case $n=0$. 
We then insert this inductive assumption into \eqref{eq:gl:dvcommute} with $N=n+1$ and integrate the latter in $u$ to find
\begin{multline*}
\left|r^{-2(n+1)}\pv[r^2\pv]^{n+1} T^i(r\phi_L)\right|(u,v) \\\leq
\int_{-k}^u \sum_{j=0}^{n+1}\frac{D}{r^{2(n+1)+2}}\cdot\frac{Cr^{\frac32+n}}{R^{L+\frac12}|u'|^{i+1}}\dd u'\leq \frac{1}{r^{2(n+1)}}\cdot\frac{C r^{\frac{1}{2}+n}}{R^{L+\frac{1}{2}}|u|^{i+1}};
\end{multline*}
so \eqref{induc1} holds for $n+1$ as well. Evaluating on $\Gamma_R$ completes the proof.
\end{proof}
\subsubsection*{Step III: The main estimates}
	\begin{claim}\label{claim3}
	Let $U_0$ be a sufficiently large negative number. There exists a constant $C$ such that 
	\begin{align}\label{eq:proof:gtl:claim3}
	\left|[r^2\pv]^{L-j}T^i(r\phi_L)\right|(u,v)\leq \frac{C}{|u|^{i+1}}R^{-\max{(j,0)}}
	\end{align}
	for $i=0,\dots, N$, $j=-1,0,\dots, L$, and for all $u\leq U_0$, $v\geq v_{\Gamma_R}(u)$. 
	\end{claim}
\begin{proof}
In order to simplify the presentation, we will additionally assume that $R/2M$ is sufficiently large. This largeness assumption can be lifted by replacing bootstrap argument below by a  Gr\"onwall argument as in \S \ref{sec:tl:removelargeness}.

Let us fix $i\leq N$. We make the following bootstrap assumption:
\newcommand{\bs}{C_{\mathrm{BS}}^{(i)}}
\begin{equation}
\left|\pvx{L}T^i(r\phi_L)\right|(u,v)\leq \frac{\bs}{|u|^{i+1}},\label{claim4bs}
\end{equation}
where $\bs$ is a constant to be specified later.
Notice that, by integrating this up to $L$ times from $\Gamma$, estimating at each step the boundary term by \eqref{eq:proof:gtl:claim2}, this implies
\begin{equation}
\left|\pvx{L-j}T^i(r\phi_L)\right|(u,v)\leq \frac{C+\bs}{R^j|u|^{i+1}}\label{claim41}
\end{equation}
for $j=0,\dots, L$. In particular, if $R$  and $\bs$ are chosen sufficiently large, then we have
\begin{equation}
\left|\pvx{L}T^i(r\phi_L)\right|(u,v)\leq 2\cdot|\Phi_L|(u,v),
\end{equation}
where we recall the definition \eqref{eq:NPfuture} of $\Phi_L$.
We now plug the bounds \eqref{claim41} into the approximate conservation law \eqref{eq:gl:approx.u} and integrate the latter in $u$ from $u=-k$ to obtain that
\begin{equation}
\left|r^{-2L}\pv T^i\Phi_L\right|(u,v)\leq \frac{C\cdot\bs}{r^{2L+2}|u|^{i+1}}.\label{claim42}
\end{equation}
We then integrate this bound in $v$ from $\Gamma$, estimating the boundary term $T^i\Phi_L|_{\Gamma_R}$ using \eqref{eq:proof:gtl:claim1}, to obtain that
\begin{equation}
\left|T^i\Phi_L\right|(u,v)\leq \frac{C}{|u|^{i+1}}+\frac{C\cdot\bs}{R|u|^{i+1}}.
\end{equation} 
Finally, we can choose $R$ and $\bs$ large enough such that 
\begin{equation}
|\pvx{L}T^i(r\phi_L)|\leq|2T^i\Phi_L |\leq\frac{\bs}{2|u|^{i+1}}.
\end{equation}
This improves the bootstrap assumption \eqref{claim4bs} and thus proves \eqref{eq:proof:gtl:claim3} for $j=0$. The result for $j>0$ then follows in view of the estimates \eqref{claim41}, and the result for $j=-1$ follows from~\eqref{claim42}.
\end{proof}
\subsubsection*{Step IV: The auxiliary estimates}
As in the $\ell=1$-case (cf.\ \S\ref{sec:tl:pvT-T2}), we will need some auxiliary estimates in order to later be able to show that certain quantities attain limits on $\mathcal{I}^-$. These auxiliary estimates will be estimates on the differences $[r^2\pv]^{L}T^i \left(r\phi^{(k)}_L-|u| T(r\phi^{(k)}_L)\right)$. As in step II, we first need estimates on the boundary terms:
	\begin{claim}\label{claim4}
	Let $U_0$ be a sufficiently large negative number, and let $N'\leq N-2$. Then there exists a constant $C$ such that 
	\begin{equation}\label{eq:proof:gtl:claim4}
	\left|[r^2\pv]^{L-j}T^i (r\phi_L-|u| T(r\phi_L))\right|(u,v_{\Gamma_R}(u))\leq \frac{C}{R^{j}|u|^{i+1+		\epsilon}}+\theta_k\cdot\frac{C}{R^{j}k^{i+1}}
	\end{equation}
	 for $i=0,\dots,N'$ and $j=0, \dots, L$, and for all $u\leq U_0$.
	\end{claim}
\begin{proof}
The proof for the case $j=L-1$ is similar to that of Claim~\ref{claim1}:
 We first assume \eqref{eq:proof:gtl:claim4} (with $j=L-1$) as a bootstrap assumption. 
 The main difference to the proof of Claim~\ref{claim1} is that we then apply the divergence theorem to 
\begin{align*}
\div J^T[T^i(\phi_L-|u|T\phi_L)]&=\Box_g(T^i(\phi_L-|u|T\phi_L))\cdot T(T^i(\phi_L-|u|T\phi_L))\\
&=-\frac{1}{Dr}\pv T^{i+1}(r\phi_L)\cdot\frac{1}{r}T^{i+1}(r\phi_L-|u|T(r\phi_L)),
\end{align*}
rather than to $\div J^T[T^i\phi_L]=0$. (Here, we used the expression \eqref{eq:box} for $\Box_g$.) This gives rise to a non-trivial bulk term. However, the estimates we established in Claim~\ref{claim3} provide sufficient bounds for this term.
 
Using the fundamental theorem of calculus and the Cauchy--Schwarz inequality, one then obtains an estimate on (cf.\ \eqref{proof5.6x})
$$\sqrt{r}T^i(\phi_L-|u|T\phi_L)\leq \frac{C}{R^{\frac{2L+1}{2}}}\left(\frac{1}{|u|^{i+1+\epsilon}}+\frac{1}{k^{i+1}}\right).$$ 
In order to translate this into an estimate on $\pv T^i(r\phi_L-|u|T(r\phi_L))$, we consider the wave equation satisfied by $\pu\pv T^i(r\phi_L-|u|T(r\phi_L))$,
\begin{nalign}
\pu\pv T^i(r\phi_L-|u|T(r\phi_L))
=-\frac{D}{r^2}\left(L(L+1)+\frac{2M}{r}\right)T^i(r\phi_L-|u|T(r\phi_L))+\overbrace{\pv T^{i+1}(r\phi_L)}^{\leq Cr^{-L-1}|u|^{-i-2}},
\end{nalign}
where we again note that the error term $\pv T^{i+1}(r\phi_L)$ can be bounded by the previous estimates (Claim~\ref{claim3}), and integrate in $u$. This improves the bootstrap assumption.

The general case $j\leq L-1$ then follows as in the proof of Claim~\ref{claim2}, noting that, when considering the wave equations for $\pu\pv[r^2\pv]^j T^i(r\phi_L-|u|T(r\phi_L))$, the error terms compared to \eqref{eq:gl:dvcommute} will always be given by $\pv[r^2\pv]^jT^{i+1}(r\phi_L)$, which we already control by Claim~\ref{claim3}.
 See also the proof of Proposition~\ref{prop:tl:pvT-T2} for more details. 
\end{proof}
Having obtained estimates on the boundary terms, we can now prove:	
	\begin{claim}\label{claim5}
	Let $U_0$ be a sufficiently large negative number, and let $N'\leq N-2$. Then there exists a constant $C$ such that 
	\begin{align}\label{eq:proof:gtl:claim5}
	\left|[r^2\pv]^{L-j}T^i(r\phi^{(k)}_L-|u|T(r\phi^{(k)}_L))\right|(u,v)\leq \frac{C}{R^{	\max(j,0)}}\left(\frac{1}{|u|^{i+1+\epsilon}}+\frac{1}{k^{i+1}}\right)
	\end{align} for $i=0,\dots,N'$, $j=-1,0, \dots, L$, and for all $u\leq U_0$, $v\geq v_{\Gamma_R}(u)$.
	\end{claim}
\begin{proof}
The proof is similar to that of Claim~\ref{claim3}, with the main modifications being that we now use Claim~\ref{claim4} in order to estimate the boundary terms. Furthermore, instead of the approximate conservation law \eqref{eq:gl:approx.u}, we now consider the equations
\begin{multline}
\pu(r^{-2L}\pv T^i(\Phi_L-|u|T \Phi_L))=\overbrace{r^{-2L}\pv T^{i+1}\Phi_L}^{\leq Cr^{-2L-2}|u|^{-i-2}}\\
+\sum_{j=0}^L\frac{D(2M)^{j+1}}{r^{2L+3}}[r^2\pv]^{L-j}T^i(r\phi_L-|u|T(r\phi_L)) \left(2(j+1)x^{(L)}_{j+1}-\sum_{i=0}^j x_i^{(L)} c_{j-i}^{L-i}\right) ,
\end{multline} in which we again control the error terms $r^{-2L}\pv T^{i+1}\Phi_L$ by Claim~\ref{claim3}. Indeed, they  decay faster near $\Gamma_R$ (that is, they have more $u$-decay).\footnote{Notice, however, that they do \emph{not} decay faster near $\mathcal{I}^-$. It is for this reason that the $R$-weights in \eqref{eq:proof:gtl:claim5} cannot directly be upgraded to $r$-weights.} See the proof of Proposition~\ref{prop:tl:pvT-T2} for more details.
\end{proof}

\subsubsection*{Step V: Taking the limit $k\to\infty$}
So far, we have proved uniform-in-$k$ estimates on the sequence of solutions $\phi_L^{(k)}$ (whose elements vanish on $u\leq -k$) constructed in Step I. We now show that these solutions converge uniformly to another solution $\phi_L$:
\begin{claim}\label{claim6}
The sequence $\{\phi_L^{(k)}\}_{k\in\mathbb{N}}$ tends to a uniform limit $\phi_L$ as $k\to\infty$,
\begin{equation}
\lim_{k\to\infty}||\phi_L^{(k)}-\phi_L 	||_{C^N(\mathcal{D}_{\Gamma_R})}=0.
\end{equation}
In fact, this limiting solution is the unique smooth solution that restricts correctly to the data of \S \ref{sec:gtl:ass}, and it satisfies for all $u\leq U_0$ and $v\geq v_{\Gamma_R}(u)$, for sufficiently large negative values of $U_0$, the following bounds for some constant $C$:
	\begin{align}
	\left|[r^2\pv]^{L-j}T^i(r\phi^{(k)}_L)\right|(u,v)\leq \frac{C}{|u|^{i+1}}R^{-\max{(j,0)}}
	\end{align}
	for $i=0,\dots, N$ and $j=-1,0,\dots, L$. Furthermore, if $N'\leq N-2$, then we also have
	\begin{align}\label{eq:proof:gtl:claim6}
	\left|[r^2\pv]^{L-j}T^i(r\phi^{(k)}_L-|u|T(r\phi^{(k)}_L))\right|(u,v)\leq \frac{C}{|u|^{i+1+\epsilon}}R^{	-\max(j,0)}.
	\end{align} for $i=0,\dots,N'$ and $j=-1,0, \dots, L$.
\end{claim}
\begin{proof}
The proof proceeds, \textit{mutatis mutandis}, as the proof of Proposition~\ref{prop:tl:limit}.
\end{proof}

\subsubsection*{Step VI: Proving sharp decay for $\pvx{L-j} T^i(r\phi_L)$ (Proof of \eqref{eq:thm:gtlmain})}

	\begin{claim}\label{claim7}
	Let $U_0$ be a sufficiently large negative number. There exists a constant $C$ such that the solution $\phi_L$ from Claim~\ref{claim6} satisfies
	\begin{align}\label{eq:proof:gtl:claim7}
	\left|[r^2\pv]^{L-j}T^i(r\phi_L)\right|(u,v)\leq \frac{C}{|u|^{i+1}}\min(r,|u|)^{-\min{(\tilde{j},N-i)}}
	\end{align}
	for all $j=-1,\dots, L$, $i=0,\dots, N$, and for all $u\leq U_0$, $v\geq v_{\Gamma_R}(u)$. Here, $\tilde{j}:=\max(j,0)$.
	\end{claim}
\begin{proof}
We will prove this inductively by showing the following lemma:
\begin{lemma}\label{claim7lem}
Let $n<L$. Then there exists a constant $C$ such that
\begin{equation}
\left|\pvx{n-j}T^i(r\phi_L)\right|(u,v)\leq \frac{C}{|u|^{i+1}}\min(r, |u|)^{-\min(j+1, N-i)}\label{inducclaim7}
\end{equation}
for all $j=0,\dots, n$, for all $i=0,\dots N$, and for all $u\leq U_0$, $v\geq v_{\Gamma_R}(u)$.
\end{lemma}
Indeed, once this lemma is shown for $n=L-1$, then \eqref{eq:proof:gtl:claim7} follows in view of Claim~\ref{claim3} (which already provides the sharp estimates for $n=L, L+1$). \end{proof}
\begin{proof}[Proof of Lemma~\ref{claim7lem}]
We first show \eqref{inducclaim7} for $n=0$. We derive from \eqref{eq:gl:dvcommute} with $N=0$ that
\begin{multline}\label{claim7n0}
-\frac{D}{r^2}T^i(r\phi_L)\left(L(L+1)+\frac{2M}{r}\right)=-\pv^2T^i(r\phi_L)+\pv T^{i+1}(r\phi_L)\\
=-\frac{1}{r^2}\pv(r^2\pv T^i(r\phi_L))+\frac{2D}{r}\pv T^i(r\phi_L)+\pv T^{i+1}(r\phi_L).
\end{multline}
We can assume, without loss of generality, that $i<N$, as \eqref{inducclaim7} follows directly from \eqref{eq:proof:gtl:claim3} if $i=N$. 
If $i<N$, then we can insert the estimates from \eqref{eq:proof:gtl:claim3} into \eqref{claim7n0} to find that
\begin{equation}
|T^i(r\phi_L)|\leq \frac{C}{r^2|u|^{i+1}}+\frac{C}{r|u|^{i+1}}+\frac{C}{|u|^{i+2}}.
\end{equation}
This establishes \eqref{inducclaim7} for $n=0$. 

Let us now assume that \eqref{inducclaim7} holds for some fixed $n<L-1$. We shall show that it then also holds for $n+1$. We derive from \eqref{eq:gl:dvcommute} the following generalisation of \eqref{claim7n0}:
\begin{multline}\label{claim7n}
D\left(a_0^{n+1}-b_0^{n+1}L(L+1)-c_0^{n+1}\frac{2M}{r}\right)[r^2\pv]^{n+1}T^i(r\phi_L)\\
=-\sum_{j=0}^n(2M)^{j+1}D\left(a_{j+1}^{n+1}-b_{j+1}^{n+1}L(L+1)-c_{j+1}^{n+1}\frac{2M}{r}\right)[r^2\pv]^{n-j}T^i(r\phi_L)\\
+[r^2\pv]^{n+2} T^{i+1}(r\phi_L)+\frac{2D(n+2)}{r}[r^2\pv]^{n+2}T^i(r\phi_L)-\frac{1}{r^2}\pvx{n+3}T^i(r\phi_L).
\end{multline}
Notice that, since $n+1<L$, the difference $a_0^{n+1}-b_0^{n+1}L(L+1)$ is non-zero. Therefore, estimating the terms in the second line of \eqref{claim7n} using the induction assumption, and the terms in the third line of \eqref{claim7n} using \eqref{eq:proof:gtl:claim3} (keeping in mind that $n+3\leq L+1$ and assuming as before that $i<N$), we obtain
\begin{multline}
\left|	[r^2\pv]^{n+1}T^i(r\phi_L)	\right| \leq \frac{C}{|u|^{i+1}}\min(r, |u|)^{-\min(1,N-i)}\\
+\frac{C}{|u|^{i+2}}+\frac{C}{r|u|^{i+1}}+\frac{C}{r^2|u|^{i+1}}\leq \frac{C}{|u|^{i+1}}\min(r, |u|)^{-\min(1,N-i)}.
\end{multline}
This establishes \eqref{inducclaim7} for $n+1$, \textit{restricted to} $j=0$. For $j>0 $, we use another induction, this time going down in derivatives:
	\begin{subclaim}\label{subclaim}
	There exists a constant $C$ such that
	\begin{equation}
	\left|[r^2\pv]^{n+1-j}T^i(r\phi_L)\right|(u,v)\leq \frac{C}{|u|^{i+1}}\min(r,|u|)^{-\min(1+j, N-i)}\label{inducclaim72}
	\end{equation}
	for all $j=0,\dots, n+1$, and for all $u\leq U_0$, $v\geq v_{\Gamma_R}(u)$.
	\end{subclaim}
\begin{proof}[Proof of Sublemma~\ref{subclaim}]
We have already established \eqref{inducclaim72} for $j=0$. 
Let us now assume that it holds for $j\leq n$  fixed. 
We shall show that \eqref{inducclaim72} then also holds for $j+1$. 
We can restrict to $i< N-(j+2)$ since the result would be trivial otherwise.
Using \eqref{claim7n}, we obtain the estimate (writing $n+1-(j+1)=n-j$):
\begin{multline}
|[r^2\pv]^{n-j}T^i(r\phi_L)|\leq \sum_{k=0}^{n-j-1}\frac{C}{|u|^{i+1}}\min(r,|u|)^{-\min(k+2+j, N-i)}\\
+\frac{C}{|u|^{i+2}}\min(r,|u|)^{-\min(j+1, N-(i+1))}+\frac{C}{|u|^{i+1}r}\min(r,|u|)^{-\min(j+1, N-i))},
\end{multline}
where we used the induction assumption \eqref{inducclaim7} to estimate the first term on the RHS. This establishes \eqref{inducclaim72} for $j+1$ and, thus, proves the sublemma.
\end{proof}
Sublemma~\ref{subclaim}  proves \eqref{inducclaim7} for $n+1$ and, hence, completes the proof of Lemma~\ref{claim7lem}.
\end{proof}
\subsubsection*{Step VII: The limit $\mathcal{L}=\lim_{u\to-\infty}|u|^2[r^2\pv]^LT(r\phi_L)$}

Throughout the rest of the proof, we shall assume that $N-2\geq N'\geq L+1$, and that $\phi_L$ denotes the solution from Claim~\ref{claim6}.
	\begin{claim}
	The limit $\mathcal{L}(v):=\lim_{u\to-\infty}|u|^2[r^2\pv]^LT(r\phi_L)(u,v)$ exists and is independent of $v$. Moreover, along any ingoing null hypersurfaces of constant $v$,
	\begin{equation}
	|u|^2[r^2\pv]^LT(r\phi_L)(u,v)-\mathcal{L}=\mathcal{O}(r^{-1}+ |u|^{-\epsilon}).
	\end{equation}
	In fact, if the lower bound \eqref{eq:gtl:boundarydata3} is assumed, and if $R/2M$ is chosen large enough, then $\mathcal{L}\neq 0$.
	\end{claim}
\begin{proof}
We show that the limit exists by computing
\begin{multline}
\pu(|u|^2\pvx{L}T(r\phi_L))\\
=\underbrace{-2|u| \pvx{L}T(r\phi_L)+|u|^2\pvx{L}T^2(r\phi_L)}_{\leq C|u|^{-1-\epsilon}}-\underbrace{|u|^2\pv\pvx{L}T(r\phi_L)}_{\leq Cr^{-2}}.
\end{multline}
The first two terms together can be bounded by $|u|^{-1-\epsilon}$ in view of  estimate \eqref{eq:proof:gtl:claim6} from Claim~\ref{claim6}. 
The third term can be bounded by $r^{-2}$ in view of  estimate \eqref{eq:proof:gtl:claim7} from Claim~\ref{claim7}.
 This establishes the existence of the limit $\mathcal L(v)$. 
 The independence of $v$ follows directly from the bound on the third term.

It is left to show that this limit is not zero. 
For this, we will also need to establish a lower bound for $|u|^2[r^2\pv]^LT(r\phi_L)$ on $\Gamma_R$. 
First, we observe that, in view of the identity \eqref{claim7n0}, we have, on $\Gamma_R$, that
\begin{equation}
\left|R^L\cdot T(r\phi_L)(L(L+1))+2R^{L-1}[r^2\pv]T(r\phi_L)\right|\leq \frac{C}{|u|^{3}}+\frac{C}{R|u|^2}.
\end{equation}
Thus, if $R$ and $|U_0|$ are chosen sufficiently large, we have, as a consequence of the lower bound~\eqref{eq:gtl:boundarydata3}:
\begin{equation}
\left|r^2\pv T(r\phi_L)\right|(u,v_{\Gamma_R}(u))\geq \frac{\cin}{4R^{L-1}|u|^2}
\end{equation}
Similarly, we can now show inductively, using \eqref{claim7n} instead of \eqref{claim7n0}, that, say, 
\begin{equation}
\left|[r^2\pv]^{L-j} T(r\phi_L)\right|(u,v_{\Gamma_R}(u))\geq \frac{\cin}{2^{L-j+1}R^{j}|u|^2}.
\end{equation}
for $j=0,\dots, L-1$, provided that $R$ is chosen sufficiently large.\footnote{Notice that this leads to an extremely wasteful lower bound on $R$. One can improve this using a different approach, or just not show the lower bound. We do not analyse this issue any further.}  

Using the lower bound above for $j=0$ and then integrating the bound \eqref{eq:proof:gtl:claim7} for $j=-1$ from $\Gamma_R$, one obtains, provided that $R$ is chosen sufficiently large,
\begin{equation}
\left|[r^2\pv]^{L} T(r\phi_L)\right|(u,v)\geq \frac{\cin}{2^{L+2}|u|^2}
\end{equation} for all $v\geq v_{\Gamma_R}(u), u\leq U_0$.
This shows that $\mathcal L\neq 0$ and thus completes the proof. (Notice that this approach only works for the highest-order $v$-derivative (i.e.\ for $j=0$).)
\end{proof}
	We now show that various different limits can be computed from this limit.	
	\begin{lemma}\label{claim9}
	The following limits exist and satisfy the relations
	\begin{equation}\label{induclimits}
	\lim_{u\to-\infty}|u|^{1+i}r^j[r^2\pv]^{L-j}T^i(r\phi_L)=\mathcal{L}^{(i,j)}
	\end{equation}
	for $i=0,\dots, N'+1$ and $j=0,\dots, L$, provided that $i+j\leq N'+1$, where $\mathcal{L}^{(i,j)}$ are rational multiples of $\mathcal{L}$. 
	\end{lemma}
\begin{proof}
We will first show \eqref{induclimits} for $j=0$. Indeed, we have $\mathcal L^{(0,0)}=\mathcal{L}^{(1,0)}=\mathcal L$, and, in view of the estimates \eqref{eq:proof:gtl:claim6} from Claim 6, we have the relations
\begin{align}
\lim_{u\to-\infty} |u|^2[r^2\pv]^L T(r\phi_L)&=\frac12\lim_{u\to-\infty} |u|^3[r^2\pv]^L T^2(r\phi_L)\\
=\frac16 \lim_{u\to-\infty}|u|^4[r^2\pv]^L T^3(r\phi_L)
&=\dots=\frac{1}{i!}\lim_{u\to-\infty} |u|^{1+i}[r^2\pv]^L T^i(r\phi_L)\nonumber,
\end{align}
provided that $i\leq N'+1$. We thus have established that, for all $i\leq N'+1$:
\begin{equation}
\mathcal{L}^{(i,0)}=i!\mathcal{L}.
\end{equation}

We now prove \eqref{induclimits} inductively in $j$. Without loss of generality, we may assume that $L>0$.
We assume that we have already established \eqref{induclimits} for $(i,j)$, for some fixed $j<L$ and for all $i\leq N'+1-j$. 
We then show that \eqref{induclimits} also holds for $(i,j+1)$, provided that $i+j+1\leq N'+1$: 
Indeed, if $i+j\leq N'$, we can appeal to equation \eqref{claim7n} (with $n+1$ replaced by $L-j-1$) to obtain
\begin{multline}
(a_0^{L-j-1}-b_0^{L-j-1}L(L+1))\lim_{u\to-\infty}|u|^{2+i}r^{j+1}[r^2\pv]^{L-j-1}T^i(r\phi_L)\\
=\lim_{u\to-\infty}|u|^{2+i+1}r^j[r^2\pv]^{L-j}T^{i+1}(r\phi_L)+2(L-j)\lim_{u\to-\infty}|u|^{2+i}r^j[r^2\pv]^{L-j}T^i(r\phi_L),
\end{multline}
where we used \eqref{eq:proof:gtl:claim7} to estimate the lower-order derivatives $[r^2\pv]^{L-j-2-k}$ for $k\geq 0$. We have thus established that
\begin{equation}
\mathcal{L}^{(i,j+1)}=\frac{1}{a_0^{L-j-1}-b_0^{L-j-1}L(L+1)}\cdot\left(\mathcal{L}^{(i+1,j)}+2(L-j)\mathcal{L}^{(i,j)}\right)
\end{equation}
for all $0\leq i\leq N'+1-(j+1)$, $0\leq j<L$. (Notice that, in this range of indices,  $a_0^{L-j-1}-b_0^{L-j-1}L(L+1)\neq 0$.) 
We leave it to the reader to derive explicit expressions for $\mathcal{L}^{(i,j)}$ from the recurrence relations above. 
\end{proof}

\subsubsection*{Step VIII: The limit $\lim_{u\to-\infty}[r^2\pu]^{L+1}(r\phi_L)$ (Proof of \eqref{eq:thm:gtl:notlimit} and \eqref{eq:thm:gtl:limit})}
\begin{claim}\label{calim9}
The limit $\lim_{u\to-\infty}[r^2\pu]^{L+1}(r\phi_L)=\tilde C\neq 0$ exists and can be computed explicitly in terms of the constants $\mathcal{L}^{(i,j)}$. Moreover, we have along any ingoing null hypersurface of constant $v$:
\begin{align}
[r^2\pu]^{L-j}(r\phi_L)(u,v)=\mathcal{O}(\min(r,|u|)^{-1-
j}),&&j=0,\dots, L,\label{eq:thm:gtl:notlimitclaim}\\
[r^2\pu]^{L+1}(r\phi_L)(u,v)=\tilde{C}+\mathcal{O}(r^{-1}+|u|^{-\epsilon}).&&\label{eq:thm:gtl:limitclaim}
\end{align}
\end{claim}
\begin{proof}
We will prove this by writing $\pu=T-\pv$ and using the following lemma:
\begin{lemma}\label{lemma:pu=t-pv}
Let $f$ be a smooth function, and let $n\in\mathbb{N}$. Then
\begin{nalign}\label{eq:lemmapp}
(r^2\pu)^n f=(r^2T-r^2\pv)^nf
=\sum_{k=0}^n (-1)^{n-k}\binom{n}{k}\left(\sum_{i=0}^{k-1}\alpha^{(n,k)}_i(r+\mathcal{O}(1))^i(r^2 T)^{k-i}\right)[r^2\pv]^{n-k}f
\end{nalign}
for some constants $\alpha^{(n,k)}_i$.
\end{lemma}
\begin{proof}
A proof is provided in the appendix~\ref{A2}.
\end{proof}
Indeed, applying this lemma with $n=L-j$ for $j\geq0$ immediately proves \eqref{eq:thm:gtl:notlimitclaim} upon inserting the bounds \eqref{eq:proof:gtl:claim7} from Claim 7.
On the other hand, applying the lemma with $n=L+1$ and recalling Lemma~\ref{claim9}, we find
\begin{equation}\label{combination}
\lim_{u\to-\infty}[r^2\pu]^{L+1}(r\phi_L)=\sum_{k=0}^{L+1}(-1)^{L+1-k}\binom{L+1}{k}\sum_{i=0}^{k-1}\alpha^{(L+1,k)}_i \mathcal{L}^{(k-i,L+1-k)}=\tilde C
\end{equation}
since we assumed that $N'\geq L+1$. Equation \eqref{eq:thm:gtl:limitclaim} then follows similarly.

Finally, we need to show that $\tilde C\neq 0$. 
Instead of explicitly computing the sum above, we proceed by contradiction: 
Suppose that \eqref{eq:thm:gtl:limitclaim} holds with $\tilde C= 0$. 
Then, in view of \eqref{eq:thm:gtl:notlimitclaim}, we have that, say, on $v=1$, $r\phi_L(u,1)=\mathcal O(|u|^{-L-1-\epsilon})$. 
Inductively inserting this estimate into \eqref{eq:gl:approx.u} and integrating in $u$ (cf.\ Proposition~\ref{prop:prop.10.1}), this implies that $[r^2\pv]^L(r\phi_L)(u,v)=\mathcal O(|u|^{-1-\epsilon})$, which would imply that $\mathcal L=0$, a contradiction.
\end{proof}

This concludes the proof of Theorem~\ref{thm:gtl}.
\end{proof}

\newpage
\section{Data on an ingoing null hypersurface \texorpdfstring{$\mathcal C_{v=1}$}{C(v=1)} I}\label{sec:general:null}
Having obtained an understanding of solutions arising from timelike boundary data in the previous section, we now aim to understand solutions arising from polynomially decaying initial data on an ingoing null hypersurface. 
We will, in the present section, focus on initial data which decay as predicted by Theorem~\ref{thm:gtl} of the previous section. 
While the present section completely generalises the methods of \S \ref{sec:nl} and contains a proof of Theorem~\ref{thm:intro:gnl}, it requires fast initial decay on the data (depending on $\ell$).
The case of more slowly decaying data will thus need to be treated differently and is discussed in \S \ref{sec:moreg:null}.

Throughout this section, we shall again assume that $\phi$ is a solution to \eqref{waveequation}, supported on a single angular frequency $(L,m)$ with $|m|\leq L$. In the usual abuse of notation, we omit the $m$-index, that is, we write $\phi=\phi_{Lm}\cdot Y_{Lm}=\phi_L\cdot  Y_{Lm}$.
\subsection{Initial data assumptions}\label{sec.9.1}
\newcommand{\cl}[1]{C_{\mathrm{in}}^{(L,#1)}}
Prescribe smooth characteristic/scattering data for \eqref{waveequation}, restricted to the angular frequency $(L,m)$,  that satisfy on $\mathcal C_{v=1}$
\begin{align}\label{eq:gnl:ass1}
	\lim_{u\to-\infty}[r^2\pu]^{L+1-i}(r\phi_L)(u,1)&=\cl{i},&& i=1,\dots,L,\\
	[r^2\pu]^{L+1}(r\phi_L)(u,1)&=\cl{0}+\mathcal{O}(r^{-\epsilon})\label{eq:gnl:ass2}
\end{align}
for some $\epsilon\in(0,1)$, where the $\cl{i}$ are constants, and which moreover satisfy for all $v\geq 1$:
\begin{equation}\label{eq:gnl:assNoIncoming}
\lim_{u\to-\infty}\pv^n(r\phi_L)(u,v)=0
\end{equation}
for $n=0,\dots, L+1$. We interpret this latter assumption as the no incoming radiation condition.

\subsection{The main theorem (Theorem~\ref{thm:gnl})}
Motivated by the previous Theorem~\ref{thm:gtl}, we will only consider the case where $\cl{i}=0$ for $i>0$ for now.  
The other cases, and further generalisations that do not require conformal regularity on the initial data, will be treated in \S \ref{sec:moreg:null}, as they have to be dealt with in a different way. 
Let us mention, however, that the proof of the present section still works if additionally $\cl{1}\neq 0$.

\begin{thm}\label{thm:gnl}
By standard scattering theory~\cite{DRSR18}, there exists a unique smooth scattering solution $\phi_L\cdot Y_{Lm}$ in $\mathcal{M}\cap\{v\geq 1\}$ attaining the data of \S\ref{sec.9.1}.

Let $U_0$ be a sufficiently large negative number. Assume moreover that $\cl{i}=0$ for all $i=1,\dots, L$ and that $\cl0\neq0$. Then, for all $u\leq U_0$, the limit of the radiation field on future null infinity is given by
\begin{align}\label{eq:thm:gnl limphi}
\lim_{v\to\infty}(r\phi_L)(u,v)=\frac{L!\cl{0}}{(2L+1)!|u|^{L+1}}+\mathcal{O}(|u|^{-L-1-\epsilon}),
\end{align}
and, throughout $\mathcal{D}=(-\infty,U_0]\times[1,\infty)$, the outgoing derivative of $r\phi_L$ satisfies, for fixed values of $u$, the following asymptotic expansion as $\mathcal{I}^+$ is approached:
\begin{align}\label{eq:thm:gnl asydvphi}
\pv(r\phi_L)(u,v)=\sum_{i=0}^L \frac{f_i^{(L)}(u)}{r^{2+i}}+\frac{ (-1)^LM\cdot B^*}{(L+1)!}\frac{\log r-\log|u|}{r^{3+L}}+\mathcal{O}(r^{-3-L}).
\end{align}
Here, the $f^{(L)}_i(u)$ are smooth functions of $u$ which satisfy $f^{(L)}_i(u)=\dfrac{(-1)^i\beta_{L-1-i}^{(L)}}{i!|u|^{L-i}}+\mathcal{O}(|u|^{-L+i-\epsilon})$ for $i<L$ and $f^{(L)}_i(u)=\dfrac{2M(-1)^L\tilde a^{L,L}_1\beta_1^{(L)}}{L!|u|}+\mathcal{O}(|u|^{-1-\epsilon})$ for $i=L$, and $
B^*=2(2x_1^{(L)}-c_0^L)\beta_0^{(L)}
$.   
The constants $\beta_i^{(L)}$ are given explicitly by the formulae\footnote{For the readers convenience, we recall that $\tilde a_1^{L,L}=-L^2$, that $2x_1^{(L)}=-L$, and that $c_0^L=1+2L(L+1)$. 
Finally, we recall that $a_0^k-b_0^kL(L+1)=k(k+1)-L(L+1)$. 
All these constants have been defined in \S \ref{sec:generalNP}.}
\begin{align}
\beta_i^{(L)}=\frac{L!\cl{0}}{(2L+1)!}\cdot\frac{i!}{L!}\prod_{k=0}^{L-i-1} ( a_0^k -b_0^kL(L+1))=(-1)^{L-i}\frac{(2L-i)!\cl{0}}{(2L+1)!}.
\end{align}

Moreover, the quantity $\Phi_L$ defined in eq.\ \eqref{eq:NPfuture} has the expansion
\begin{equation}\label{eq:thm:gnl asyPhi}
\pv\Phi_L=\frac{ M\cdot B^*(\log r-\log|u|)}{r^{3}}+\mathcal{O}(r^{-3}),
\end{equation}
and, in particular, the logarithmically modified Newman--Penrose constant is finite and conserved:
\begin{equation}\label{eq:thm:gnl NP}
I_{\ell=L}^{\mathrm{future},\frac{\log r}{r^3}}[\phi](u):=\lim_{v\to\infty}\frac{r^3}{\log r}\pv\Phi_L(u,v)=M\cdot B^*\neq 0.
\end{equation}

\end{thm}
\begin{rem}\label{rem:thm:gnl}
The first two statements of the above theorem, \eqref{eq:thm:gnl limphi} and \eqref{eq:thm:gnl asydvphi}, still apply if one lifts the restriction $\cl{1}=0$, albeit with different constants and with different $f_i^{(L)}(u)$. See also Remark~\ref{rem:9}. In particular,  we again have a cancellation if $r\phi_L\sim 1/|u|^L$ initially: The initial $|u|^{-L}$-decay translates into $|u|^{-L-1}$-decay on $\mathcal I^+$.
Equations \eqref{eq:thm:gnl asyPhi} and \eqref{eq:thm:gnl NP}, on the other hand change, change: The leading-order decay behaviour of $\pv\Phi_L$ is now given by $\sim\frac{1}{r^2}$, and, in particular, the usual Newman--Penrose constant
\begin{equation}
I_{\ell=L}^{\mathrm{future}}[\phi]:=\lim_{v\to\infty}r^2\pv\Phi_L(u,v)
\end{equation}
will be finite, generically non-vanishing, and conserved along future null infinity. If one also allows $\cl{i}\neq 0$ for $i>1$, then the modifications to Theorem~\ref{thm:gnl} will be more severe, see already Theorem~\ref{thm:moreg}.
\end{rem}
\subsubsection*{Overview of the proof}
We will prove the theorem in two steps. First, we will obtain an asymptotic estimate for $r\phi_L$ (which will, in particular, imply \eqref{eq:thm:gnl limphi}) by integrating \eqref{eq:gl:ducommute} in $v$ from data and then integrating the result $L+1$ times in $u$ from $\mathcal I^-$. 

Then, we will use this estimate to get  the leading-order decay of $\pv(r\phi_L)$ by integrating equation \eqref{eq:gl:dvcommute} with $N=0$. 
Once this is achieved, we will inductively obtain leading-order asymptotics for $(r^2\pv)^n(r\phi_L)$ using the corresponding \eqref{eq:gl:dvcommute}, from which we can,  in turn, deduce higher-order asymptotics for $\pv(r\phi_L)$. 
This will prove \eqref{eq:thm:gnl asydvphi}. 
Equation \eqref{eq:thm:gnl asyPhi} then follows in a similar fashion from the approximate conservation law \eqref{eq:gl:approx.u}.
\subsection{Asymptotics for \texorpdfstring{$r\phi_L$}{r phiL}}
\begin{prop}\label{propp}
There exists a constant $C$ depending only on data such that $r\phi_L$ satisfies the following asymptotic expansion throughout $\mathcal{D}$:
\begin{equation}\label{eq:gnl:asy.phi.r}
\left|r\phi_L(u,v)-\frac{L!\cl{0}}{(2L+1)!|u|^{L+1}}\right|\leq\frac{C}{|u|^{L+1+\epsilon}}+\frac{C}{r|u|^L}.
\end{equation}
In particular, we have
\begin{equation}\label{eq:gnl:asy.phi.u}
\lim_{v\to\infty}r\phi_L(u,v)=\frac{L!}{(2L+1)!|u|^{L+1}}+\mathcal{O}(|u|^{-L-1-\epsilon}).
\end{equation}
\end{prop}
\begin{proof}
By applying the weighted energy estimate of Proposition~\ref{prop:l:energy} (whose proof still works for higher $\ell$-modes), we obtain the decay estimate (cf.\ Corollary~\ref{cor:almostsharp}):
\begin{equation}\label{eq:cli:1}
|r\phi_L(u,v)|\leq C|u|^{-L-1}
\end{equation}
for some constant $C$ depending only on initial data.

By inserting this estimate into eq.\ \eqref{eq:gl:ducommute} with $N=0$ and integrating the latter from $v=1$, we then obtain that $|\pu(r\phi_L)|\leq C|u|^{-L-2}$.
Similarly, by inductively integrating eq.\ \eqref{eq:gl:ducommute} for higher $N\leq L$ from $v=1$, we find that 
\begin{equation}\label{eq:cli:2}
|(r^2\pu)^n(r\phi)|\leq C\frac{r^{2n}}{|u|^{L+1+n}}.
\end{equation}
We can plug these estimates into \eqref{eq:gl:ducommute} with $N=L$, 
\begin{multline}\label{proofnl1}
	r^{2L}\pv(r^{-2L}\pu[r^2\pu]^L(r\phi_L))=-\underline{c}_0^L\frac{2MD}{r^3}\overbrace{[r^2\pu]^L(r\phi_L)}^{\lesssim r^{2L}|u|^{-2L-1}}\\
	+\sum_{j=1}^N\frac{D}{r^2}(2M)^j\underbrace{[r^2\pu]^{L-j}(r\phi_L)}_{\lesssim r^{2L-2}|u|^{-2L}}\cdot \left(	\underline{a}_j^L-\underline{b}_j^L-\underline{c}_j^L\frac{2M}{r}	\right), 
\end{multline}
and integrate in $v$ from $v=1$ to find that 
\begin{equation}\label{proofnl:Phipast}
\left|\pu[r^2\pu]^{L}(r\phi_L)(u,v)-\cl{0}\frac{r^{2L}}{|u|^{2L+2}}\right|\leq C\frac{r^{2L}}{|u|^{2L+2+\epsilon}}.
\end{equation}

Essentially, we can now integrate \eqref{proofnl:Phipast}  $L+1$ times from $\mathcal I^-$ to improve the bootstrap assumption. For this, we invoke the following
\begin{lemma}\label{lemma9.1}
Let $N, N'\in \mathbb{N}$ with $N>N'+1$. Then
\begin{equation}\label{eq:lemma91}
(N-1)\int_{-\infty}^u\frac{r^{N'}}{|u'|^{N}}\dd u'=\frac{r^{N'}}{|u|^{N-1}}
+\sum_{k=1}^{N'}\frac{r^{N'-k}}{|u|^{N-1-k}}\prod_{j=1}^k\frac{N'+1-j}{N-1-j}\left(1+\mathcal{O}(r^{-1})\right).
\end{equation}
\end{lemma}
\begin{proof}
The proof is straight-forward, but nevertheless provided in the appendix~\ref{A3}.
\end{proof}
We now apply Lemma~\ref{lemma9.1} with $N'=2L$, $N=2L+2$ to \eqref{proofnl:Phipast} (and divide by $r^2$) to obtain that
\begin{equation}\label{proofnl:Phipast-1}
\left|\pu[r^2\pu]^{L-1}(r\phi_L)(u,v)-\frac{\cl{0}}{2L+1}\left(\frac{r^{2(L-1)}}{|u|^{2L+1}}+
\sum_{k=1}^{2L}\frac{r^{2(L-1)-k}}{|u|^{2L+1-k}}\right)\right|\leq C\frac{r^{2(L-1)}}{|u|^{2L+1+\epsilon}}.
\end{equation}
Here, we used that the boundary term vanishes,\footnote{In the general case, $\cl{i}\neq0$, these boundary terms would of course not vanish.} $$\lim_{u\to-\infty}[r^2\pu]^L(r\phi_L)(u,v)=0.$$
Notice that the terms inside the sum $\sum_{k=1}^{2L}$ decay faster near $\mathcal{I}^+$ than the $\frac{r^{2(L-1)-k}}{|u|^{2L+1-k}}$-term inside \eqref{proofnl:Phipast-1}. 
Therefore, inductively applying the above procedure $L$ more times\footnote{Notice that the $\sum_{k=1}^{2L}$-sum in \eqref{proofnl:Phipast-1} also contains the terms $\frac{1}{|u|^3}$, $\frac{1}{r|u|^2}$ and $\frac{1}{r^2|u|}$. The latter two are not of the form of Lemma~\ref{lemma9.1}, but can simply  be estimated against the former.}, one obtains
\begin{equation}
\left|r\phi_L(u,v)-\frac{1}{(2L+1)\cdots (L+1)}\frac{\cl{0}}{|u|^{L+1}}\right|\leq \frac{C}{|u|^{L+1+\epsilon}}+\frac{C}{r|u|^L}.
\end{equation}
This proves the proposition.
\end{proof}
\begin{rem}\label{rem:9}
We stated in Remark~\ref{rem:thm:gnl} that parts of Theorem~\ref{thm:gnl} still apply if one assumes that also $\cl{1}=\lim_{u\to-\infty}[r^2\pu]^L(r\phi_L)\neq 0$. 
Let us explain the modifications to the proof above needed to see this: 
First, one needs to replace the $|u|^{-L-1}$-decay in \eqref{eq:cli:1} with $|u|^{-L}$-decay. 
This then leads to the RHS's of \eqref{eq:cli:2} and \eqref{proofnl1} to decay one power in $u$ slower. 
However, this still produces the same leading-order decay of $\pu[r^2\pu]^L(r\phi_L)$ in \eqref{proofnl:Phipast}. 
Upon integrating this from $\mathcal I^-$, one picks up the limit $\cl{1}=\lim_{u\to-\infty}[r^2\pu]^L(r\phi_L)$ and obtains an asymptotic estimate for $[r^2\pu]^L(r\phi_L)$ near $\mathcal I^-$. 
One can then, as was done in \S \ref{sec:nl}, insert this asymptotic estimate for $[r^2\pu]^L(r\phi_L)$ back into \eqref{proofnl1} and proceed with the rest of the proof as above.
In view of the boundary term coming from $\lim [r^2\pu]^L(r\phi_L)$, one then obtains, instead of \eqref{eq:gnl:asy.phi.r},
\begin{equation}\label{eq:gnl:asy.phi.rrrrr}
\left|r\phi_L(u,v)-\frac{L!\cl{0}+C'(\cl{1})}{(2L+1)!|u|^{L+1}}-\frac{C''(\cl{1})}{r^L}\right|\leq\frac{C}{|u|^{L+1+\epsilon}}
\end{equation}
for some $C'(\cl{1})$, $C''(\cl{1})$ which depend only on $M$, $L$ and $\cl{1}$. 

The proof of the asymptotics for $\pv(r\phi_L)$, presented in the next section, remains largely unchanged. See also \S \ref{sec:nl} for details.
\end{rem}
\subsection{Asymptotics for \texorpdfstring{$\pv(r\phi_L)$}{d/dv(r phiL)} and proof of Theorem~\ref{thm:gnl}}
\begin{proof}[Proof of Theorem~\ref{thm:gnl}]
Having obtained the asymptotics for $r\phi_L$ along $\mathcal{I}^+$, we can now compute the asymptotics of $\pv(r\phi_L)$. For the sake of notational simplicity, we restrict to $L\neq 0$ for now, the case $L=0$ is recovered in \eqref{tul2}.

We first compute the leading-order asymptotics by integrating the wave equation \eqref{eq:gl:dvcommute} with $N=0$, 
\begin{equation}\label{eq:gl:wave}
\pu\pv(r\phi_L)=-\frac{D}{r^2}\left(L(L+1)+\frac{2M}{r}\right)r\phi_L
\end{equation}
from past null infinity (where $\pv(r\phi_L)$ vanishes by assumption \eqref{eq:gnl:assNoIncoming}) and by plugging in the estimate \eqref{eq:gnl:asy.phi.r}. This yields, after also commuting \eqref{eq:gl:wave} with $r^2$,
\begin{equation}
\left|	r^2\pv(r\phi_L)(u,v)+\frac{(L+1)!}{(2L+1)!}\frac{\cl{0}}{|u|^L}	\right|\leq \frac{C}{|u|^{L+\epsilon}}+\frac{C}{r|u|^{L-1}},
\end{equation}
where, from now on, $C$ will be a constant which depends only on data and which is allowed to vary from line to line.
More precisely, by writing $r\phi_L(u,v)=r\phi_L(u,\infty)-\int_{v}^\infty\pv(r\phi)(u,v)\dd v'$ in \eqref{eq:gl:wave}, we can write
\begin{equation}
\left|	r^2\pv(r\phi_L)(u,v)+L(L+1)\int_{-\infty}^u \lim_{v\to\infty}(r\phi)(u',v)\dd u'	\right|\leq \frac{C}{r|u|^{L-1}}.
\end{equation}

Let us now make the following induction assumption. Let $n\geq 1$. Then we assume that for  all $L\geq i\geq n$:
\begin{equation}\label{induc:nl}
\left|[r^2\pv]^{L-i}(r\phi_L)-\frac{\beta_i^{(L)}}{|u|^{i+1}}\right|\leq \frac{C}{|u|^{i+1+\epsilon}}+\frac{C}{r|u|^i}
\end{equation}
for some non-vanishing constants $\beta_i^{(L)}\in\mathbb{Q}$.
Since we have already established that this holds true for $n=L$ with $\beta_{L}^{(L)}=\frac{L!\cl0}{(2L+1)!}$, it suffices to show that \eqref{induc:nl} also holds for $n-1\geq 1$, provided that it holds for $n$. Therefore, we now consider \eqref{eq:gl:dvcommute} with $N=L-n\geq 0$:
\begin{align}\label{induc:nlstep}
&\pu\left(r^{-2(L-n)}\pv[r^2\pv]^{L-n}(r\phi_L)\right)\\
=&\frac{1}{r^{2(L-n)}}\sum_{i=n}^L\frac{D(2M)^{i-n}}{r^2}[r^2\pv]^{L-i}(r\phi_L)\left(a_{i-n}^{L-n}-b_{i-n}^{L-n}L(L+1)-c_{i-n}^{L-n}\frac{2M}{r}\right)\nonumber.
\end{align}
Plugging in the induction assumption \eqref{induc:nl} for the terms on the RHS and then integrating \eqref{induc:nlstep} gives that $\pv[r^2\pv]^{L-n}(r\phi_L)$ is of order $\mathcal O(r^{-2}|u|^{-n})$. Moreover, commuting now \eqref{induc:nlstep} with $r^{2(L-(n-1))}$, we obtain that
\begin{multline}\label{induc:nlstep2}
\pu(r^2\pv[r^2\pv]^{L-n}(r\phi_L))=\overbrace{-\frac{2(L-(n-1))D}{r}\cdot r^2\pv[r^2\pv]^{L-n}(r\phi_L)}^{\lesssim r^{-1}|u|^{-n}}\\
+\sum_{i= n}^L D(2M)^{i-n}\underbrace{[r^2\pv]^{L-i}(r\phi_L)}_{\lesssim |u|^{-n-1}}\left(a_{i-n}^{L-n}-b_{i-n}^{L-n}L(L+1)-c_{i-n}^{L-n}\frac{2M}{r}\right),
\end{multline}
from which, in turn, we recover, by again integrating from $\mathcal{I}^-$, that
\begin{equation}\label{induc:nl+1}
\left|[r^2\pv]^{L-(n-1)}(r\phi_L)-\frac{\beta_{n-1}^{(L)}}{|u|^{n}}\right|\leq \frac{C}{|u|^{n+\epsilon}}+\frac{C}{r|u|^{n-1}},
\end{equation}
with $\beta_{n-1}^{(L)}$ given by
\begin{equation}\label{pindi}
n\beta_{n-1}^{(L)}={\beta_n^{(L)}}\left(a_0^{L-n}-b_0^{L-n}L(L+1)\right)\neq0.
\end{equation}
This proves \eqref{induc:nl} for all $n\geq 1$ and, thus, that \eqref{induc:nl+1} holds for all $n\geq 2$. 
In fact, it is easy to see that \eqref{induc:nl+1} also holds for $n=1$, with the $r^{-1}|u|^{-n+1}$-term on the RHS replaced by $\log(1-v/u)/v$ (cf.\ \eqref{eq:nl:stupidintegral1.5}).

In order to get a similar  estimate to \eqref{induc:nl+1} for $n=0$, we recall the crucial cancellation in \eqref{eq:gl:approx.u} for $N=L$ (namely $a_0^L-b_0^LL(L+1)=0$). We are thus led to consider, in a very similar fashion to the above, the equation
\begin{multline}\label{qlneli}
\pu(r^{-2L}\pv[r^2\pv]^{L}(r\phi_L))=-\frac{1}{r^{2L}}\frac{D}{r^2}\cdot c_0^L\frac{2M}{r}[r^2\pv]^L(r\phi_L)\\
+\frac{1}{r^{2L}}\sum_{i=1}^L\frac{D(2M)^i}{r^2}[r^2\pv]^{L-i}(r\phi_L)\left(a_i^{L}-b_i^{L}L(L+1)-c_i^{L}\frac{2M}{r}\right).
\end{multline}
The first term on the RHS is bounded by $Cr^{-2L-3}|u|^{-1}$, whereas the other terms in the sum are bounded by $C{r^{-2L-2}|u|^{-2}}$. More precisely, we have
\begin{multline}
\pu(r^{-2L}\pv[r^2\pv]^{L}(r\phi_L))=-\frac{2MD}{r^{2L+3}|u|} c_0^L\beta_0^{(L)}\\
+\frac{2MD}{r^{2L+2}|u|^2}\left(a_1^L-b_1^L L(L+1)\right)\beta_1^{(L)}+\mathcal{O}\left(r^{-2L-2}|u|^{-2-\epsilon}+r^{2L-3}\frac{\log 1-\frac vu}{v}\right).
\end{multline}
Integrating this from $\mathcal{I}^-$ then yields that
\begin{equation}\label{tul}
[r^2\pv]^{L+1}(r\phi_L)(u,v)=\frac{2M}{|u|}\left(a_1^L-b_1^L L(L+1)\right)\beta_1^{(L)}
+\mathcal{O}\left(|u|^{-1-\epsilon}+\frac 1 r +\frac{\log 1-\frac{v}{u}}{v}\right),
\end{equation}
where, in the two asymptotic equalities above, we made use of the integral estimates \eqref{eq:nl:stupidintegral1.5} and \eqref{eq:nl:stupidintegral2}.
In order to find the logarithmic next-to leading order asymptotics, we insert the estimates above into the approximate conservation law \eqref{eq:gl:approx.u}:
\begin{nalign}
\pu(r^2\pv\Phi_L)=(2L+2)Dr \pv\Phi_L
+\sum_{j=0}^L\frac{D}{r}(2M)^{j+1}[r^2\pv]^{L-j}(r\phi_{L})\left(2(j+1)x^{(L)}_{j+1}-\sum_{i=0}^j x_i^{(L)} c_{j-i}^{L-i}\right) .
\end{nalign}
From this, we then obtain, in a similar way to how we proved the estimates above, that
\begin{equation}\label{tul2}
r^2\pv\Phi_L(u,v)=2M(2x_1^{(L)}-c_0^L)\beta_0^{(L)}\frac{\log(v-u)-\log|u|}{v}+\mathcal{O}(r^{-1}).
\end{equation}
Notice that the difference $2x_1^{(L)}-c_0^L$ is non-vanishing for all $L(\geq0)$ since $x_1^{(L)}=-\frac L2$ and $c_0^L=1+2L(L+1)$.

Comparing equations \eqref{tul} and \eqref{tul2} then gives us the next-to-leading-order asymptotics for $[r^2\pv]^{L+1}(r\phi_L)$ since, for $i>0$, the terms $[r^2\pv]^{L-i}(r\phi_L)$ contained in $\Phi_L$ do not contain logarithmic terms at next-to-leading order, which can be seen by integrating \eqref{tul} $i$ times from $\mathcal I^+$.

Finally, the statement \eqref{eq:thm:gnl asydvphi} follows by simply integrating these asymptotics $L$ times from $\mathcal{I}^+$ and using \eqref{induc:nl} for the arising boundary terms on $\mathcal{I}^+$.

This concludes the proof of Theorem~\ref{thm:gnl}.\end{proof}
\subsection{Comments}\label{sec:gnl:comments}
\subsubsection{A logarithmically modified Price's law at all orders}
We expect Theorem~\ref{thm:gnl} (and, in particular, eq.\ \eqref{eq:thm:gnl NP}) to imply a logarithmically modified Price's law for the $\ell=L$-mode (see also the remarks in \S \ref{sec:nl:comments}). 
However, Theorem~\ref{thm:gnl} only applies if $\cl{i}=0$ for all $i>0$. 
This assumption, in turn, is motivated by the results of the previous \S\ref{sec:general:timelike} (eq.\ \eqref{eq:thm:gtl:notlimit} from Theorem~\ref{thm:gtl}).
Therefore, although the general situation ($\cl{i}\neq0$) might (and will) be quite different, we can expect that the data considered in \S \ref{sec:general:timelike}, i.e.\ data on a timelike boundary data which decay like $r\phi_\ell\sim |t|^{-1}$ near $i^-$ and which are smoothly extended to $\mathcal H^+$, lead to a logarithmically modified Price's law, for each $\ell$.
To be concrete, the expected decay rates would then be
\begin{equation}
r\phi_L|_{\mathcal{I}^+}\sim u^{-L-2}\log u,\quad \phi_L|_{\mathcal{H}^+}\sim v^{-2L-3}\log v
\end{equation}
near $i^+$.  Moreover, the leading-order asymptotics should be independent of the extension of the data towards $\mathcal H^+$.
 As has been discussed in \S \ref{sec:intro:NP}, the proof of the above expectation should follow by combining the results of this paper with those of~\cite{AAG21}, similarly to how~\cite{II} combined the results of~\cite{I} with those of~\cite{AAG18a,AAG18b}, so long as sufficient regularity (depending on $L$) is assumed. The fixed-regularity problem, on the other hand, seems much more difficult, cf.\ Conjecture~\ref{conj1}.

\subsubsection{The case \texorpdfstring{$\cl{i}\neq 0$}{Cin(L,i) neq 0}}
Notably, the proof presented in this section cannot be directly applied  to the case $\cl{i}\neq 0$ for $i>1$, since one would encounter several difficulties related to the quantities $(r^2\pu)^i(r\phi_L)$. (Notice already that the limits
$
\lim_{u\to-\infty}[r^2\pu]^{i}(r\phi_L)
$
grow like $v^{i-1}$ for $i=1,\cdots, L$, and like $v^{L-1}$ for $i=L+1$.) Furthermore, working with the quantity $[r^2\pu]^{L+1}(r\phi_L)$ requires strong conformal regularity assumptions.
 In the next section, we shall therefore obtain asymptotics for much more general data by working only with the quantities $[r^2\pv]^i(r\phi_L)$ and not using the approximate conservation law in $v$ \eqref{eq:gl:approx.v} at all.

\newpage
\section{Data on an ingoing null hypersurface \texorpdfstring{$\mathcal C_{v=1}$}{C(v=1)} II}\label{sec:moreg:null}
In this final section, we present a different approach towards obtaining the early-time asymptotics of $\pv(r\phi_L)$ of solutions $\phi_L$ arising from polynomially decaying initial data on a null hypersurface $\mathcal C_{v=1}$, without requiring any conformal regularity and/or fast decay on initial data.
 In particular, this section contains the proof of Theorem~\ref{thm:intro:moreg} from the introduction and can also treat the cases $\cl{i}\neq 0$ from the previous section.

Throughout this section, we shall again assume that $\phi$ is a solution to \eqref{waveequation}, supported on a single angular frequency $(L,m)$ with $|m|\leq L$. 
In the usual abuse of notation, we omit the $m$-index, that is, we write $\phi=\phi_{Lm}\cdot Y_{Lm}=\phi_L\cdot  Y_{Lm}$.

\subsection{Initial data assumptions}\label{sec:mg:ass}
\newcommand{\ci}{C_{\mathrm{in}}}
Prescribe smooth characteristic/scattering data for \eqref{waveequation}, restricted to the angular frequency $(L,m)$,  that satisfy on $\mathcal C_{v=1}$
\begin{equation}\label{eq:mg:ass1}
\left|r\phi_L(u,1)-\frac{\ci}{r^p}\right|=\mathcal O_1( r^{-p-\epsilon})
\end{equation}
for some $\epsilon\in(0,1]$, a constant $\ci>0$ and for some $p\in\mathbb N_0$, and which moreover satisfy 
\begin{equation}\label{eq:mg:assNoIncoming}
\lim_{u\to-\infty}\pv^n(r\phi_L)(u,v)=0
\end{equation}
for all $v\geq 1$ and for all $n=1,\dots, L+1$.

Notice that if $p=1$ in \eqref{eq:mg:ass1}, then this includes the cases $\cl{i}\neq 0$ from \S\ref{sec:general:null}.

\subsection{The main theorem (Theorem~\ref{thm:moreg})}
\begin{thm}\label{thm:moreg}
By standard scattering theory~\cite{DRSR18}, there exists a unique smooth scattering solution $\phi_L\cdot Y_{Lm}$ in $\mathcal{M}\cap\{v\geq 1\}$ attaining the data of \S\ref{sec:mg:ass}.

Let $U_0$ be a sufficiently large negative number, let $\mathcal{D}=(-\infty,U_0]\times[1,\infty)$, and let $r_0:=r(u,1)=|u|-2M\log|u|+\mathcal{O}(1)$. Then the following statements hold for all $(u,v)\in\mathcal D$:

a) We have that:
\begin{equation}\label{eq:mg:thma1}
\lim_{v\to\infty}r\phi_L(u,v)=F(u)=\begin{cases}
												\mathcal O(r_0^{-p-\epsilon}),&\text{if }p\leq L\text{ and } p\neq 0\\
												{C_0\cdot}\ci r_0^{-p}+\mathcal O(r_0^{-p-\epsilon}), &\text{if }p>L \text{ or } p=0,
									\end{cases}
\end{equation}
for some smooth function $F(u)$ and some explicit, non-vanishing constant $C_0=C_0(L,p)$.

b) Moreover, the outgoing derivative of the radiation field $\pv(r\phi_L)$ satisfies the following asymptotic expansion \underline{if $p<L$}:
\begin{equation}\label{eq:mg:thmb1}
r^2\pv(r\phi_L)(u,v)=\sum_{i=0}^{p-1}\frac{f^{(L,p)}_i(u)}{r^i}+\frac{M\cdot C_1(\log r-\log|u|)+C_2 r_0}{r^p}+\mathcal{O}\left(\frac{|u|^{1-\epsilon}}{r^{p}}\right),
\end{equation}
where the $f^{(L,p)}_i$ are smooth functions of order $f^{(L,p)}_i=\mathcal O(r_0^{-p+i+1-\epsilon})$ if $i<p-1$, or of order $f^{(L,p)}_{i}=C_3^{L,p,i}+\mathcal O(r_0^{-\epsilon})$ if $i=p-1$.

On the other hand, \underline{if $p\geq L$}, then 
\begin{equation}\label{eq:mg:thmb2}
r^2\pv(r\phi_L)(u,v)=\sum_{i=0}^{L}\frac{f^{(L,p)}_i(u)}{r^i}+\mathcal O\left(\frac{\log r}{r^{L+1}}\right),
\end{equation}
where the $f^{(L,p)}_i$ are smooth functions of order $f^{(L,p)}_i=\mathcal O(r_0^{-p+i+1-\epsilon})$ if $p=L$ and $i<L-1$, and which are given by $f^{(L,p)}_i=C_3^{L,p,i}r_0^{-p+i+1}+\mathcal O(r_0^{-p+i-\epsilon})$ otherwise (i.e.\ if $p=L=i$, if $p=L=i+1,$ or if $p>L$).

In each case, we have explicit, non-vanishing expressions for the constants $C_1,C_2,C_3^{L,p,i}$ that depend only  on $L,p,i,\ci$ (and not on $M$).

c) Finally, \underline{if $p\leq L$}, then the following limit exists, is non-vanishing, and is independent of $u$:
\begin{equation}\label{eq:mg:thmc1}
\lim_{v\to\infty}r^{2+p-L}\pv\Phi_L(u,v)=I^{\mathrm{future},r^{2+p-L}}_{\ell=L}[\phi]\neq 0.
\end{equation} 

\underline{If $p=L+1$}, then the following limit exists, is non-vanishing, and is independent of $u$:
\begin{equation}\label{eq:mg:thmc2}
\lim_{v\to\infty}\frac{r^3}{\log r}\pv\Phi_L(u,v)=I^{\mathrm{future},\frac{\log r}{r^3}}_{\ell=L}[\phi]\neq 0.
\end{equation} 

\underline{If $p>L+1$}, then $\pv\Phi_L=\mathcal O(r^{-3})$, and all modified Newman--Penrose constants vanish on $\mathcal I^+$.

In each case, we have explicit expressions for the constants $I^{\mathrm{future},r^{2+L-p}}_{\ell=L}[\phi],I^{\mathrm{future},\frac{\log r}{r^3}}_{\ell=L}[\phi]$. These depend only on $L,p,M,\ci$.

All explicit expressions for constants are listed in the proof of this theorem on page~\pageref{auxref}.
\end{thm}
 \begin{rem}
 Notice that we often expressed $u$-decay in terms of $r_0$ rather than $u$ in order to compactly express logarithmic contributions: For instance, if $\epsilon=1$, and if we express decay directly in terms of $u$, then we  have an additional $\mathcal O(|u|^{-p-1}\log |u|)$-contribution in the second line of \eqref{eq:mg:thma1}. 
 \end{rem}
 
\begin{rem}
The faster decay in \eqref{eq:mg:thma1} for $p\leq L$ can be traced back to certain cancellations. These  already happen for $M=0$. In fact, we will, in the course of the proof, derive effective expressions for exact solutions to the wave equation on Minkowski which have data $r\phi_L(u,1)=\ci/r^p$ and satisfy the no incoming radiation condition \eqref{eq:mg:assNoIncoming}. (See Proposition~\ref{prop:mg:prop4}.) 
\end{rem}

\begin{rem}
If also the next-to-leading-order behaviour on initial data is specified in \eqref{eq:mg:ass1}, we can upgrade the $\mathcal O$-symbols in \eqref{eq:mg:thma1} etc.\ to precise asymptotics, see Corollary~\ref{cor10}.
\end{rem}
\begin{rem}\label{rem104}
With a bit more effort, one can extend the analysis of the proof (using for instance time integrals as was done in~\cite{I}) to show that \eqref{eq:mg:thmb2} can be improved to 
$$
r^2\pv(r\phi_L)(u,v)=\sum_{i=0}^{\max{(L,p-1)}}\frac{f^{(L,p)}_i(u)}{r^i}+\mathcal O\left(\frac{\log r}{r^{\max(L+1,p)}}\right).
$$

Notice that "the first logarithmic term" of the expansion of $r^2\pv(r\phi_L)$ never appears at order $r^{-L}\log r$: It either appears at order $r^{-L-i}\log r$ or at order $r^{-L+i}\log r$, with $i>0$. In this sense, there is a cancellation happening at $L=p$. In particular, if $p=1$, then the expansion of $\pv(r\phi_L)$ contains a logarithmic term at order $r^{-3}\log r$ for all $L\neq 1$ (including $L=0$), whereas the first logarithmic term for $L=1$ only appears at order $r^{-4}\log r$.
\end{rem}
\begin{rem}
Using the methods of the proof, one can show a very similar result if one assumes more generally that $0\leq p\in\mathbb R$. (In fact, one should also be able to consider a certain range of positive $p$!) In this case, however, the cancellation \eqref{eq:mg:thma1} in general no longer appears; it seems to be a special property of $p\in\{1,\dots, L\}$. 
\end{rem}

\subsection{Overview of the proof}\label{sec:mg:overview}
In contrast to the proof of \S \ref{sec:general:null}, we will, in this section, only use the approximate conservation law \eqref{eq:gl:approx.u} and obtain an asymptotic estimate for $[r^2\pv]^i(r\phi_L)$ directly from data. 
For this, we will first need to compute $[r^2\pv]^N(r\phi_L)$ for $N\leq L$ on data, i.e.\ on $v=1$, by inductively integrating the relevant equation~\eqref{eq:gl:dvcommute}. This is done in Proposition~\ref{prop:prop.10.1} in \S\ref{sec:mg:data}.  
We then make a bootstrap assumption on the decay of $[r^2\pv]^L(r\phi_L)$ and improve it using \eqref{eq:gl:approx.u}.
 Once we have obtained a sharp estimate on $[r^2\pv]^L(r\phi_L)$ in this way (Proposition~\ref{prop:prop.10.2} in \S\ref{sec:mg:x}), we can then inductively integrate from $v=1$ to obtain a sharp estimate for $[r^2\pv]^{L-i}(r\phi_L)$ (Proposition~\ref{prop:prop.10.3} in \S\ref{sec:mg:y}).
  In doing so, we pick up an "initial data term" with each integration. 
  These data terms will all be of the same order, so there might be cancellations. 
We will understand these cancellations  in Proposition~\ref{prop:mg:prop4} in~\S\ref{sec:mg:datacanc}. 
  The results are then summarised in Corollary~\ref{cor10}. 
  Finally, the proof of Theorem~\ref{thm:moreg} is given in~\S \ref{sec:mg:proofofTHM}.

The disadvantage of this more direct approach to the asymptotics of $\pv$-derivatives of $r\phi_L$ is that we gain no direct information on the asymptotics of $\pu$-derivatives. On the other hand, this should also be seen as an advantage since this approach requires no assumption on the conformal regularity of the initial data on $v=1$. 

\subsection{Computing transversal derivatives on data}\label{sec:mg:data}
Inserting the initial data bound \eqref{eq:mg:ass1} into the wave equation \eqref{eq:gl:dvcommute} with $N=0$, and integrating from $u=-\infty$, where $\pv(r\phi_L)$ vanishes by the no incoming radiation condition \eqref{eq:mg:assNoIncoming}, we obtain that on $v=1$
\begin{equation}
\left|\pv(r\phi_L)+\frac{L(L+1)\ci}{(2+p-1)r^{1+p}}\right|\leq C r^{-1-p-\epsilon}
\end{equation}
for some constant $C$.
In turn, inserting this estimate into \eqref{eq:gl:dvcommute} with $N=1$, one obtains an estimate for $[r^2\pv]^2(r\phi_L)$. Proceeding inductively, one obtains the following
\newcommand{\cv}[1]{C_{[r^2\pv]^{#1}}}
\begin{prop}\label{prop:prop.10.1}  Let $\phi_L$ be as in Theorem~\ref{thm:moreg}. Then 
we have on $v=1$ that
\begin{align}\label{eq:mg:databounds}
[r^2\pv]^{N+1}(r\phi_L)&=\cv{N+1}r^{N+1-p}+\mathcal{O}(r^{N+1-p-\epsilon}),&&\text{for }N=0,\dots, L-1,\\
[r^2\pv]^{L+1}(r\phi_L)&=\cv{L+1}r^{L-p}+\mathcal{O}(r^{L-p-\epsilon}).&&
\end{align}
Here, we defined the constants
\begin{align}\label{eq:mg:cv}
\cv{N+1}:=& \frac{p!\ci}{(N+1+p)!}\cdot \prod_{i=0}^N(a_0^i-b_0^iL(L+1)),\quad N=0,\dots, L-1,\\
\cv{L+1}:=&\frac{2M}{L+2+p}\left(-c_0^L \cv{L}+(a_1^L-b_1^LL(L+1)\cv{L-1}\right).
\end{align}
\end{prop}
\begin{proof} Inductively integrate equation \eqref{eq:gl:dvcommute}.
\end{proof}
\subsection{Precise leading-order behaviour of \texorpdfstring{$[r^2\pv]^L(r\phi_L)$}{(r2 d/dv)L(r phi-L)}}\label{sec:mg:x}Equipped with the initial data estimates \eqref{eq:mg:databounds}, we now prove
\begin{prop}\label{prop:prop.10.2}
Let $\phi_L$ be as in Theorem~\ref{thm:moreg}. Then we have
\begin{nalign}\label{eq:mg:prop2}
&[r^2\pv]^L(r\phi_L)(u,v)=[r^2\pv]^L(r\phi_L)(u,1)\\
+&2M\cv{L}(2x_1^{(L)}-c_0^L)\frac{(L-\min(p,L))!(L+1+p)!}{(2L+2)!}\cdot \begin{cases} \frac{r^{L-1-p}}{(L-1-p)} (1+\mathcal O(\frac{|u|}{r}+\frac{1}{r^\epsilon})),&L>p+1,\\
\log r-\log |u| +\mathcal O(1),& L=p+1,\\
\mathcal O(|u|^{L-1-p}) ,&L<p+1.\end{cases}
\end{nalign}
If $L=p$, then we can write more precisely:
\begin{nalign}
[r^2\pv]^L(r\phi_L)(u,v)=[r^2\pv]^L(r\phi_L)(u,1)+
2M\cv{L}\left(\frac{(2x_1^{(L)}-c_0^L)}{2L+2}-x_1^{(L)}\right)r_0^{-1}+\mathcal{O}(|u|^{-1-\epsilon}).
\end{nalign}
\end{prop}
Notice that if $L<p+1$, then the second line of \eqref{eq:mg:prop2} decays faster than the first line, whereas if $L\geq p+1$, the second line determines the leading-order $r$-behaviour. 

In principle, we can also compute the case $L<p$ more precisely, but since it has already been dealt with in \S \ref{sec:general:null}, we choose not to.
 Suffice it to say that if $L<p$, then there will also be a logarithmic term at order $(\log r-\log|u)/r^{p+1-L}$, and if $L=p$, there will be a logarithmic term at order $(\log r-\log|u)/r^{p+2-L}$. Cf.~Remark~\ref{rem104}.
\begin{proof}
We will prove the proposition by first making a bootstrap assumption to obtain a preliminary estimate on $[r^2\pv]^L(r\phi_L)$ (see \eqref{eq:mgx}), and then using this preliminary estimate to obtain the sharp leading-order decay.
\subsubsection{A preliminary estimate}
We can deduce from the energy estimate from Proposition~\ref{prop:l:energy} that $|r\phi_L|\leq C |u|^{-p}$ for some constant $C$ solely determined by initial data. 
Cf.\ Corollary~\ref{cor:almostsharp}. 
(This is the reason why we also assumed decay on the first derivative in \eqref{eq:mg:ass1}.)
By repeating the calculations done in \S \ref{sec:mg:data}, we can then derive from $|r\phi_L|\leq C |u|^{-p}$ the estimates
\begin{align}\label{eq:mg:BS2}
[r^2\pv]^{N+1}(r\phi_L)(u,v)\leq C r^{N+1}|u|^{-p}+\mathcal{O}(r^{N+1-p-\epsilon})
\end{align}
 for $N=0,\dots, L-1$ and another constant $C$.

Consider now the set $X$ of all $V> 1$ such that the bootstrap assumption
\begin{align}\label{eq:mg:BS}
\left|[r^2\pv]^L(r\phi_L)\right|(u,v)\leq C_{\mathrm{BS}} \max(r^{L-p},|u|^{L-p})
\end{align}
holds for all $1\leq v\leq V$ and
for some suitably chosen constant $C_{\mathrm{BS}}$. 
The $\max$ above distinguishes between the cases of growth ($L-p>0$) and decay ($L-p<0$). 
\textit{For easier readability, we will suppose for the next few lines that $L\geq p$. This assumption will be removed in \eqref{eq:mgx}.}
 In view of the estimate \eqref{eq:mg:BS2}, this set is non-empty provided that $C_{\mathrm{BS}}$ is sufficiently large, and it suffices to improve the assumption \eqref{eq:mg:BS} within $X$ to deduce that $X=(1,\infty)$. 
 Indeed, if we assume that $\sup_X v$ is finite and improve estimate \eqref{eq:mg:BS} within $X$, then \eqref{eq:mg:BS2} shows that $\sup_X v+\delta$ would still be in $X$ for sufficiently small $\delta$ by continuity.
 (Here, we used that the RHS of \eqref{eq:mg:BS2} can be written as $C(V)\cdot \max(r^{N+1-p},|u|^{N+1-p})$ for some continuous function $C(V)$.)

Let us therefore improve the bound \eqref{eq:mg:BS} inside $X$:
First, note that \eqref{eq:mg:BS} implies that
\begin{align}\label{eq:mg:BS3}
\left|[r^2\pv]^{N}(r\phi_L)\right|(u,v)\lesssim C_{\mathrm{BS}} \max(r^{N-p},|u|^{N-p}).
\end{align}
for $N=0,\dots,L$, where we also used \eqref{eq:mg:databounds}.
Recall now equation \eqref{eq:gl:approx.u}:
\begin{equation}
\pu(r^{-2L}\pv\Phi_L)
=\sum_{j=0}^L\frac{D}{r^{3+2L}}(2M)^{j+1}[r^2\pv]^{L-j}(r\phi_{L})\left(2(j+1)x^{(L)}_{j+1}-\sum_{i=0}^j x_i^{(L)} c_{j-i}^{L-i}\right) .
\end{equation}
As a consequence of \eqref{eq:mg:BS3} and the bootstrap assumption \eqref{eq:mg:BS}, the RHS is bounded by $r^{-L-3-p}$:
\begin{equation}
\left|\pu(r^{-2L}\pv\Phi_L)\right|\lesssim \frac{2MD}{r^{L+3+p}}\cdot C_{\mathrm{BS}}+\mathcal O(r^{-L-4-p})\lesssim \frac{C_{\mathrm{BS}}}{r^{L+3+p}}.
\end{equation}
We integrate this bound in $u$ from $\mathcal I^-$, where $r^{-2L}\pv\Phi_L=0$ by the no incoming radiation condition \eqref{eq:mg:assNoIncoming}. This yields
\begin{equation}
\left|r^{-2L}\pv\Phi_L\right|\lesssim \frac{C_{\mathrm{BS}}}{r^{L+2+p}}.
\end{equation}
We now recall the definition of $\Phi_L$ from \eqref{eq:NPfuture} and estimate the difference $\pv\Phi_L-\pv[r^2\pv]^L(r\phi_L)$ using once more the bootstrap assumption, resulting in the bound
\begin{equation}
\left|\pv[r^2\pv]^L(r\phi_L)\right|\lesssim C_{\mathrm{BS}}r^{L-2-p}.
\end{equation}
Finally, we integrate the bound above from $v=1$, 
\begin{equation}\label{eq:mgx}
\left|[r^2\pv]^L(r\phi_L)(u,v)-[r^2\pv]^L(r\phi_L)(u,1)\right|\lesssim\begin{cases} C_{\mathrm{BS}}r^{L-1-p},& L> 1+p,\\
C_{\mathrm{BS}}(\log r-\log |u|), & L=1+p,\\
C_{\mathrm{BS}}|u|^{L-1-p}, & L<1+p,\end{cases}
\end{equation}
which, combined with the initial data bound \eqref{eq:mg:databounds}, improves the bootstrap assumption. 
(The third case in \eqref{eq:mgx} follows from considerations similar to the above.)
However, we can already read off from \eqref{eq:mgx} that, unless $L<1+p$, it is actually the RHS of \eqref{eq:mgx} that determines the leading-order $r$-behaviour of $[r^2\pv]^L(r\phi_L)$, whereas
the data term on the LHS only determines the leading-order $u$-behaviour. 
We will understand the precise behaviour of the RHS in the next section.
\subsubsection{Precise leading-order behaviour of \texorpdfstring{$[r^2\pv]^L(r\phi_L)$}{(r2 d/dv)L(r phi-L)}}
We again restrict to $p\leq L$ for simpler notation, the only major difference if $p>L$ is explained in Remark~\ref{rem10}. 
We will also assume for simplicity that $\epsilon<1$, leaving the case $\epsilon=1$ to the reader.

In order to find the precise leading-order behaviour of $[r^2\pv]^L(r\phi_L)$, we repeat the previous steps, with the difference that we now  use the improved estimate\footnote{\label{footnote}The $|u|^{L-p-\epsilon}$ term in the RHS above needs to be replaced by $|u|^{L-1-p}\log |u|$ if $\epsilon= 1$ because $r\sim |u|-2M\log|u| $ on $v=1$. This can be fixed by replacing $|u|^{L-p}$ with $r_0^{L-p}$.}
\begin{align}
\left|	[r^2\pv]^L(r\phi_L)-\cv{L}|u|^{L-p}	\right|\leq C r^{L-1-p}(\delta_{L,p+1}(\log r-\log |u|)+1)+C|u|^{L-p-\epsilon},
\end{align}
implied by \eqref{eq:mgx} and \eqref{eq:mg:databounds}, instead of the preliminary estimate \eqref{eq:mg:BS}. 
Similarly, we improve the estimate \eqref{eq:mg:BS3} to
\begin{align}
\left| [r^2\pv]^N(r\phi_L)  \right|\leq C |u|^{N-p}+Cr^{N-1-p}(\log r-\log |u|+1)
\end{align} for all $N=0,\dots, L-1$.
Inserting these two bounds into \eqref{eq:gl:approx.u}, we obtain
\begin{multline}
\pu(r^{-2L}\pv\Phi_L)
=\frac{2MD}{r^{2L+3}}(2x_1^{(L)}-x_0^{(L)}c_0^L)\cdot \cv{L}|u|^{L-p}\\
+\mathcal O\left(r^{-L-4-p}(\log r-\log|u|+1)+r^{-2L-3}|u|^{L-p-\epsilon}\right).
\end{multline}
Integrating the above in $u$ gives
\begin{align}\label{eq:mgxx}
r^{-2L}\pv\Phi_L=\int_{-\infty}^u \frac{2MD}{r^{2L+3}}(2x_1^{(L)}-x_0^{(L)}c_0^L)\cdot \cv{L}|u'|^{L-p}\dd u'+\mathcal O(r^{L-2-p-\epsilon}).
\end{align}
On the LHS of \eqref{eq:mgxx}, we have
\begin{align}\label{eq:mg:LHS}
r^{-2L}\pv\Phi_L=r^{-2L}\pv[r^2\pv]^L(r\phi_L)
+2Mx_1^{(L)}\frac{\cv{L}|u|^{L-p}}{r^{2L+2}}+\mathcal O\left(\frac{r^{L-1-p}+\log r +|u|^{L-p-\epsilon}}{r^{2L+2}}\right).
\end{align}
In order to estimate the RHS of \eqref{eq:mgxx}, we recall that $x_0^{(L)}=1$, and compute the integral using the following
\begin{lemma}\label{lemma10.1}
Let $N, N'\in \mathbb{N}$ with $N>N'+1$. Then
\begin{equation}\label{eq:lemma101}
(N-1)\int_{-\infty}^u\frac{|u'|^{N'}}{|r|^{N}}\dd u'=
\sum_{k=0}^{N'}\frac{|u|^{N'-k}}{|r|^{N-1-k}}\prod_{j=1}^k\frac{N'+1-j}{N-1-j}+\mathcal O(r^{N'-N}).
\end{equation}
\end{lemma}
\begin{proof}
The proof proceeds almost identically to the proof of Lemma~\ref{lemma9.1}. 
Alternatively, one can also compute the integral directly by writing $|u'|=r+v+\mathcal O(\log r)$. 
This latter approach is also useful for $N'\notin \mathbb N$.
\end{proof}
\begin{rem}\label{rem10}
When considering the case $p>L$, then the lemma above slightly changes (i.e.\ for $N'<0$). 
While it is trivial to obtain the $|u|^{L-1-p}$-decay claimed in \eqref{eq:mg:prop2}, one can also obtain a more precise statement:
 In fact, if $N'<0$, then the above integral is precisely the one that gave rise to the logarithmic terms in the previous sections \S\ref{sec:nl} and \S\ref{sec:general:null} (see, for instance, eq.\ \eqref{eq:nl:stupidintegral1.5}). 
 In particular, if $N>0>N'$, the integrals of \eqref{eq:lemma101} will lead to logarithmic terms at order $(\log r-\log|u|)/r^{N-N'-1}$.
\end{rem}
Applying Lemma~\ref{lemma10.1} (with $N'=L-p$ and $N=2L-3$) to \eqref{eq:mgxx}, we obtain 
\begin{nalign}\label{eq:mg:yy}
r^{-2L}\pv\Phi_L
=2M\cv{L}(2x_1^{(L)}-c_0^L)\frac{(L-p)!}{(2L+2)\cdots (L+2+p)}\cdot\frac{1}{r^{L+2+p}}\left(1+\mathcal{O}\left(\frac{|u|}{r}+\frac{1}{r^{\epsilon}}\right)\right)
\end{nalign}
where we used that $|u|^q/r^q\leq 1$ for any $q>0$.
Finally, using \eqref{eq:mg:LHS} to write $\pv\Phi_L\sim \pv[r^2\pv]^L(r\phi_L)$, and integrating from $v=1$, we obtain, \underline{if $L>1+p$},
\begin{multline}\label{eq:mg:pvLbound1}
[r^2\pv]^L(r\phi_L)(u,v)=[r^2\pv]^L(r\phi_L)(u,1)\\
+2M\cv{L}(2x_1^{(L)}-c_0^L)\frac{(L-p)!(L+1+p)!}{(L-1-p)(2L+2)!}\cdot r^{L-1-p} \left(1+\mathcal{O}\left(\frac{|u|}{r}+\frac{1}{r^{\epsilon}}\right)\right).
\end{multline}

On the other hand, \underline{if $L=1+p$}, then we obtain
\begin{multline}\label{eq:mg:pvLbound2}
[r^2\pv]^L(r\phi_L)(u,v)=[r^2\pv]^L(r\phi_L)(u,1)\\
+2M\cv{L}(2x_1^{(L)}-c_0^L)\frac{(L-p)!(L+1+p)!}{(2L+2)!}\cdot (\log r-\log|u|) +\mathcal O(1)
\end{multline} 
where we used that $\log r(u,1)=\log |u|+\mathcal O(|u|^{-1})$.

The case \underline{$L<1+p$} follows in much the same way.
 The only difference is that one now also needs to take the second term on the RHS of \eqref{eq:mg:LHS} into account since it will give a contribution of the same order as the  $r^{-2L}\pv\Phi_L$-term. 
 For the latter, one can derive an estimate similar to \eqref{eq:mg:yy}. 
 This concludes the proof of Proposition~\ref{prop:prop.10.2}.
\end{proof}
\subsection{Precise leading-order behaviour of \texorpdfstring{$[r^2\pv]^{L-i}(r\phi_L)$}{(r2 d/dv)L-i(r phi-L)}}\label{sec:mg:y}

\begin{prop}\label{prop:prop.10.3}
Let $\phi_L$ be as in Theorem~\ref{thm:moreg}, and let $0\leq j\leq L$. Then we have \underline{for $j<L-1-p$:}
\begin{multline}\label{eq:mg:prop31}
[r^2\pv]^{L-j}(r\phi_L)(u,v)=\mathrm{data}_{L-j}\\
+2M\cv{L}(2x_1^{(L)}-c_0^L)\frac{(L-2-p-j)!(L-p)(L+1+p)!}{(2L+2)!}\cdot r^{L-1-p-j}\left(1+\mathcal{O}\left(\frac{|u|}{r}+\frac{1}{r^{\epsilon}}\right)\right) 
\end{multline}
On the other hand, \underline{if $j=L-1-p$}, we have
\begin{multline}\label{eq:mg:prop31.5}
[r^2\pv]^{p+1}(r\phi_L)(u,v)=\mathrm{data}_{p+1}\\
+2M\cv{L}(2x_1^{(L)}-c_0^L)\frac{(L-p)(L+1+p)!}{(2L+2)!}\cdot (\log r-\log|u|) +\mathcal O(1).
\end{multline}
Finally, \underline{if $j>L-1-p$}, we have
\begin{equation}\label{eq:mg:prop32}
[r^2\pv]^{L-j}(r\phi_L)(u,v)=\mathrm{data}_{L-j}+\mathcal O(|u|^{-j+(L-1-p)}).
\end{equation} Moreover, if also $p\leq L-1$, then $[r^2\pv]^{L-j}(r\phi_L)$ possesses an asymptotic expansion in powers of $1/r$ up to $r^{-j+L-p}$, with a logarithmic term appearing at order
\begin{equation}\label{eq:mg:prop33}
2M\cv{L}(2x_1^{(L)}-c_0^L)\frac{(L-p)!(L+p+1)!}{(L-p-1)!(2L+2)!}\frac{(-1)^{j-(L-1-p)}}{(j-(L-1-p))!}\cdot\frac{\log r-\log|u|}{r^{j-(L-1-p)}}.
\end{equation}
In the above, the expression $\mathrm{data}_{L-j}$ is shorthand for
\begin{multline}\label{eq:mg:prop34}
\mathrm{data}_{L-j}:=[r^2\pv]^{L-j}(r\phi_L)(u,1)\\
+\sum_{i=1}^j 
\underbrace{\int_{r(u,1)}^{r(u,v)}\frac{1}{Dr_{(i)}^2}\dots\int_{r(u,1)}^{r_{(2)}}\frac{1}{Dr_{(1)}^2}\dd r_{(1)}\cdots\dd r_{(i)}}_{i \text{ integrals}}\cdot[r^2\pv]^{L-j+i}(r\phi_L)(u,1).
\end{multline}
\end{prop}
\begin{rem}
Notice that \eqref{eq:mg:prop33} only holds for $p\leq L-1$. 
Indeed, we already know from the results of \S \ref{sec:general:null} (or \S \ref{sec:nl} for $L=1$) that if $p=L$, then the first logarithmic term in the expansion of $[r^2\pv]^{L-j}(r\phi_L)$ will appear at order $\dfrac{\log r-\log|u|}{r^{j-(L-2-p)}}$.
 For $p>L$, in contrast, one can show the first logarithmic term will appear at $\dfrac{\log r-\log|u|}{r^{j-(L-1-p)}}$, although we won't show this here. 
 (Again, we have in fact already shown this for $p=L+1$ in \S \ref{sec:general:null}.) 
 In this sense, there is a cancellation happening at $p=L$.
\end{rem}
\begin{proof}
We simply need to integrate \eqref{eq:mg:pvLbound1} (or, in general, the bound \eqref{eq:mg:prop2}) $j$ times from $v=1$, using at each step the initial data bounds \eqref{eq:mg:databounds}.
If $j<L-1-p$, one obtains inductively that
\begin{multline}\label{eq:mgxxx}
[r^2\pv]^{L-j}(r\phi_L)(u,v)=[r^2\pv]^{L-j}(r\phi_L)(u,1)\\
+\sum_{i=1}^j 
\int_{r(u,1)}^{r(u,v)}\frac{1}{Dr_{(i)}^2}\int_{r(u,1)}^{r_{(i)}}\frac{1}{Dr_{(i-1)}^2}\dots\int_{r(u,1)}^{r_{(2)}}\frac{1}{Dr_{(1)}^2}\dd r_{(1)}\cdots\dd r_{i-1}\dd r_{(i)}[r^2\pv]^{L-j+i}(r\phi_L)(u,1)\\
+2M\cv{L}(2x_1^{(L)}-c_0^L)\frac{(L-p)!(L+1+p)!\cdot r^{L-1-p-j}}{(L-1-p)\cdots (L-1-p-j)\cdot(2L+2)!} 
\cdot \left(1+\mathcal O\left(\frac{|u|}{r}+\frac{1}{r^{\epsilon}}\right)\right).
\end{multline}
Here, the terms in the second line come from integrating the initial data contributions (divided by $r^2$) in the estimates for $[r^2\pv]^{L-j+i}(r\phi_L)$. 
Notice that all these terms have the same leading-order $u$-decay, so there might be cancellations between them. 
We will return to this in \S\ref{sec:mg:datacanc}.
 For now, we  simply leave them as they are and write them as $\mathrm{data}_{L-j}$.

From \eqref{eq:mgxxx}, one deduces that if $j=L-1-p$:
\begin{multline}\label{eq:mgxxxx}
[r^2\pv]^{p+1}(r\phi_L)(u,v)=\mathrm{data}_{p+1}\\
+2M\cv{L}(2x_1^{(L)}-c_0^L)\frac{(L-p)!(L+1+p)!}{(L-1-p)!(2L+2)!}\cdot (\log r-\log|u|) +\mathcal O(1).
\end{multline}
Assuming that $L-1-p\geq 0$, we finally integrate \eqref{eq:mgxxxx} again from $v=1$ (and write $j':=j-(L-1-p)$) to obtain that,
for all $L\geq j\geq L-1-p$:
\begin{nalign}\label{eq:mgxxxxx}
[r^2\pv]^{L-j}(r&\phi_L)(u,v)=\mathrm{data}_{L-j}\\
+&2M\cv{L}(2x_1^{(L)}-c_0^L)\frac{(L-p)!(L+p+1)!}{(L-p-1)!(2L+2)!}\frac{(-1)^{j'}}{j'!}\cdot\left(\frac{\log r-\log|u|}{r^{j'}}\right) +\mathcal O(|u|^{-j'}),
\end{nalign}
where we inductively used that, for any $q>0$,
\begin{align*}
(q-1)\int_{r(u,1)}^{r(u,v)}\frac{\log r-\log |u|}{r^q}\dd r=-\frac{\log r-\log|u|}{r^{q-1}}+\frac{1}{|u|^{q-1}}-\frac{1}{r^{q-1}}.
\end{align*}
Notice that, for $j>L-1-p$, in contrast to \eqref{eq:mgxxx} and \eqref{eq:mgxxxx}, the leading-order $r$-decay of $[r^2\pv]^{p+1-j}(r\phi_L)$ is no longer determined by the second line of \eqref{eq:mgxxxxx}, but by the first line, namely the initial data terms. (If $j=L-p$,  the second line still provides the next-to-leading-order behaviour in $r$.) To nevertheless prove the fourth claim \eqref{eq:mg:prop33} of the proposition, one can simply obtain an analogue of \eqref{eq:mgxxxxx} by integrating the estimate \eqref{eq:mgxxxx} $j$ times from \textit{future null infinity}, rather than from $v=1$. This concludes the proof.
\end{proof}

Setting $j=L$ in the above, we get
\begin{multline}\label{eq:mgz}
r\phi_L(u,v)=r\phi_L(u,1)+\sum_{i=1}^L 
\int_{r(u,1)}^{r(u,v)}\frac{1}{Dr_{(i)}^2}\dots\int_{r(u,1)}^{r_{(2)}}\frac{1}{Dr_{(1)}^2}\dd r_{(1)}\cdots\dd r_{(i)}[r^2\pv]^{i}(r\phi_L)(u,1)\\
+\mathcal O\left(\frac{\log r-\log |u|}{r^{p+1}}\right)+\mathcal{O}\left(\frac{1}{|u|^{p+1}}\right).
\end{multline}
In view of \eqref{eq:mg:databounds}, the above estimate shows that $r\phi_L=C|u|^{-p}+\mathcal O(r^{-1}|u|^{-p+1})$ for some constant $C$, however, this constant $C$ might potentially be zero. Indeed, we already know that this is what happens in the case $L=1=p$ (discussed in \S \ref{sec:nl}), to which, in fact,  \eqref{eq:mgz} applies. (Recall that we showed that $r\phi_1\sim 1/r+1/|u|^2$ if $r\phi_1\sim 1/|u|$ initially.) We discuss these potential cancellations now.

\subsection{Cancellations in the initial data contributions}\label{sec:mg:datacanc}
We now analyse the $v=1$-contributions $\mathrm{data}_{L-j}$ in the first (and second) line(s)  of \eqref{eq:mgxxx}--\eqref{eq:mgxxxxx} in more detail. 
Define $r(u,1)=r_0(u)$.
We will prove the following 
\begin{prop}\label{prop:mg:prop4}
Let $0\leq j\leq L$. 
The expression $\mathrm{data}_{L-j}$ defined in \eqref{eq:mg:prop34} evaluates to\footnote{Note that since we express $u$-decay  in terms of $r_0$, \eqref{eq:mg:prop4} holds for all $\epsilon \leq 1$. Cf.\ footnote~\ref{footnote}.}
\begin{equation}\label{eq:mg:prop4}
\mathrm{data}_{L-j}=\ci r_0^{L-p-j}\sum_{n=0}^j S_{L,p,j,n}\left(\frac{r(u,1)}{r(u,v)}\right)^n+\mathcal O(r_0^{L-p-j-\epsilon}),
\end{equation}
where the $S_{L,p,j,n}$ are constants that are computed explicitly in eq.\ \eqref{eq:mg:SLPj-n}. They never vanish if $p> L$. However, if $p\leq L$, then they vanish if and only if $L-p+n+1\leq j\leq L$. 
\end{prop}
\begin{rem}
The computations required for the proof of the above are completely Minkowskian. 
This is to be understood in the sense that the $M$-dependence of \eqref{eq:mg:prop4} is entirely  contained in the $\mathcal O(r_0^{L-p-j-\epsilon})$ term. 
In fact, the above proposition provides us with exact solutions to the linear wave equation on Minkowski that arise from initial data $r\phi_L(u,1)=\ci r_0^{-p}$ and the no incoming radiation condition \eqref{eq:mg:assNoIncoming}.
\end{rem}
\begin{proof}
We first require an expression for the integrals in $\mathrm{data}_{L-j}$. For this, we prove
\begin{lemma}\label{lem10.2}Let $k\in\mathbb N$. Then
\begin{align}
\underbrace{\int_{r(u,1)}^{r(u,v)}\frac{1}{r_{(k)}^2}\dots\int_{r(u,1)}^{r_{(2)}}\frac{1}{r_{(1)}^2}\dd r_{(1)}\cdots\dd r_{(k)}}_{k\text{ integrals}}=\frac{1}{k!}\left(\frac{1}{r(u,1)}-\frac{1}{r(u,v)}\right)^k.
\end{align}
\end{lemma}
\begin{proof}
The proof is deferred to  the appendix~\ref{A4}.
\end{proof}
Equipped with this lemma, we can write the data contributions in the estimates of Proposition~\ref{prop:prop.10.3}, namely
\begin{equation}
\mathrm{data}_{L-j}=\sum_{i=0}^j 
\int_{r(u,1)}^{r(u,v)}\frac{1}{Dr_{(i)}^2}\dots\int_{r(u,1)}^{r_{(2)}}\frac{1}{Dr_{(1)}^2}\dd r_{(1)}\cdots\dd r_{(i)}[r^2\pv]^{L-j+i}(r\phi_L)(u,1),
\end{equation}
as follows, writing from now on $r=r(u,v)$ and $r_0=r(u,1)$ and estimating the $D^{-1}$-terms against $1+\mathcal O(r_0^{-1})$:
\begin{equation}
\mathrm{data}_{L-j}=\sum_{i=0}^j \frac{1}{i!}\left(\frac{1}{r_0}-\frac{1}{r}\right)^i	 [r^2\pv]^{L-j+i}(r\phi_L)(u,1)\cdot\left(1+\mathcal{O}(r_0^{-1})\right).
\end{equation}

We now insert the estimates \eqref{eq:mg:databounds} to write this as
\begin{equation}
\mathrm{data}_{L-j}=\sum_{i=0}^j  \frac{p!\prod_{k=0}^{L-j+i-1}(a_0^k-L(L+1))}{i!(L+p-j+i)!}\left(\frac{1}{r_0}-\frac{1}{r}\right)^i\cdot	\ci r_0^{L-p-j+i}\left(1+\mathcal{O}(r_0^{-\epsilon})\right),
\end{equation}
where we used that $b_0^k=1$. By noting that
\begin{equation*}
a_0^k-L(L+1)=k(k+1)-L(L+1)=-(L+k+1)(L-k), 
\end{equation*}
we can further express the product as
\begin{equation}\label{eq:mg:product}
\prod_{k=0}^{L-j+i-1}(a_0^k-L(L+1))=(-1)^{L-j+i}\frac{(2L-j+i)!L!}{L!(j-i)!}=(-1)^{L-j+i}\frac{(2L-j+i)!}{(j-i)!}.
\end{equation}
This yields
\begin{multline}\label{eq:mg:datalj}
\mathrm{data}_{L-j}=\sum_{i=0}^j \frac{ (-1)^{L-j+i}p!(2L-j+i)!}{i!(L+p-j+i)!(j-i)!}\left(\frac{1}{r_0}-\frac{1}{r}\right)^i	\cdot	\ci r_0^{L-p-j+i}\left(1+\mathcal{O}(r_0^{-\epsilon})\right).
\end{multline}
A cancellation at leading-order, i.e.\ at order $r_0^{L-p-j}$, takes place if the sum
\begin{equation}
\sum_{i=0}^j (-1)^{L-j+i} \frac{p!(2L-j+i)!}{i!(L+p-j+i)!(j-i)!}=:(-1)^{L-j}p!\cdot \mathrm{sum}(L,p,j)
\end{equation}
vanishes. To understand when this happens, we prove the following
\begin{lemma}\label{prop10}
Let $0\leq j\leq L$.
If $p>L$ {or $p=0$}, then $\mathrm{sum}(L,p,j)$ never vanishes.
If $0<p\leq L$, then $\mathrm{sum}(L,p,j)$ vanishes if and only if $j\in\{L-p+1,\dots,L\}$.

More precisely, if $p>L$, then
\begin{equation}\label{eq:prop101}
\mathrm{sum}(L,p,j)=\overbrace{\int_0^1\int_0^{x_{p-L}}\cdots\int_{0}^{x_{2}}}^{p-L \text{ integrals}}x_{1}^{2L-j}\frac{(1-x_{1})^j}{j!}\dd x_{1}\cdots\, dx_{p-L-1}\dd x_{p-L},
\end{equation}
which is manifestly positive. On the other hand, if $p\leq L$, then
\begin{equation}\label{eq:prop102}
\mathrm{sum}(L,p,j)=(-1)^j\binom{L-p}{L-p-j}\frac{(2L-j)!}{(L+p)!}, 
\end{equation}
where we use the convention that $\binom{n}{k}=0$ if $k<0\leq n$. 

In fact, equation \eqref{eq:prop102} also applies to $p>L$ if we define in the standard way  $\binom{L-p}{L-p-j}:=(-1)^{j} \binom{p-L+j-1}{j}$.
\end{lemma}
\begin{proof}
The proof is deferred to the appendix~\ref{A5}. Notice, however, that one can make certain soft statements without having to do any computations. For instance, if we consider the case $p=1$ and suppose there are no cancellations for $j=L$, then we would obtain from \eqref{eq:mgz} an estimate of the form $r\phi_L=C/|u|+\mathcal O(1/r+1/|u|^2)$. Inserting this into the wave equation \eqref{eq:gl:dvcommute} with $N=0$ would then give that $\pv(r\phi_L)\sim \log r/r^2$, a contradiction to the estimate \eqref{eq:mgxxxxx} for $j=L-1$. Thus, there has to be a cancellation at $j=1$; in other words, $\mathrm{sum}(L,1,1)=0$. However, we here choose to calculate the sums explicitly.
\end{proof}
Lemma~\ref{prop10} provides us with an understanding of cancellations at leading-order, i.e.\ at order $r_0^{L-p-j}$. Similarly, we can understand cancellations at higher order in \eqref{eq:mg:datalj}, say at order $r_0^{L-p-j+n}r^{-n}$, by considering the corresponding sum
\begin{align}
\sum_{i=0}^j(-1)^n \binom{i}{n} (-1)^{L-j+i} \frac{p!(2L-j+i)!}{i!(L+p-j+i)!(j-i)!}=: (-1)^{L-j}p!\cdot\mathrm{sum}(L,p,j,n).
\end{align}
Understanding this sum is straight-forward: We have
\begin{align*}
\mathrm{sum}(L,p,j,n)=&\sum_{i=0}^j(-1)^n\binom{i}{n} (-1)^{i} \frac{(2L-j+i)!}{i!(L+p-j+i)!(j-i)!}\\
=\frac{1}{n!} &\sum_{i=n}^j (-1)^{i-n} \frac{(2L-j+i)!}{(i-n)!(L+p-j+i)!(j-i)!}\\
=\frac{1}{n!} &\sum_{i=0}^{j-n} (-1)^{i} \frac{(2L-j+i+n)!}{i!(L+p-j+i+n)!(j-i-n)!},
\end{align*}
and thus
\begin{equation}
\mathrm{sum}(L,p,j,n)=\frac{1}{n!}\mathrm{sum}(L,p,j-n),
\end{equation}
where we set $\mathrm{sum}(L,p,j)=0$ if $j<0$.
In particular, in view of Lemma~\ref{prop10} above, we obtain that if $j\geq n$ and $p> L$, no cancellations occur. On the other hand, if $j\geq n$ and $p\leq L$, then cancellations occur if and only if $L-p+n+1\leq j\leq L$ and $j-n\geq 1$.

This concludes the proof of Proposition~\ref{prop:mg:prop4}, with the constants $ S_{L,p,j,n}$ being given by
\begin{equation}\label{eq:mg:SLPj-n}
 S_{L,p,j,n}=(-1)^{L-j}\frac{p!}{n!}\mathrm{sum}(L,p,j-n),
\end{equation}
where $\mathrm{sum}(L,p,j-n)$ is computed explicitly in Lemma~\ref{prop10}.
\end{proof}

\subsection{Summary and proof of Theorem~\ref{thm:moreg}}\label{sec:mg:summary}
We can roughly (and schematically) summarise the results obtained so far as
\begin{align}\label{cases}
[r^2\pv]^{L-j}(r\phi_L)(u,v)\sim
	\begin{cases}
			|u|^{L-p-j},& \text{ if } L<p \text{ or } j=L-p\geq 0,\\
			|u|^{L-p-j}(|u|^{-\epsilon}+|u|r^{-1})	,& \text{ if } L\geq p \text{ and } j\geq L-p+1,\\
			\log r-\log|u| +|u|,		& \text{ if } L> p \text{ and } j=L-p-1,\\
			r^{L-1-p-j}+|u|^{L-p-j},		& \text{ if } L> p \text{ and } 0\leq j<L-p-1.
	\end{cases}
\end{align}
The first case follows from estimate \eqref{eq:mg:prop32} from Proposition~\ref{prop:prop.10.3} and the fact that \textit{there are no} cancellations in the data term $\mathrm{data}_{L-j}$ in view of Proposition~\ref{prop:mg:prop4}. The leading-order behaviour is thus entirely determined by the data.

In contrast, the second case follows from \eqref{eq:mg:prop32} from Proposition~\ref{prop:prop.10.3} and the fact that \textit{there are} cancellations in the data term $\mathrm{data}_{L-j}$ in view of Proposition~\ref{prop:mg:prop4}. 
Notice moreover that if $\epsilon<1$, then the leading-order behaviour will only have contributions from the data. (To see this, one needs to repeat the calculations of Proposition~\ref{prop:prop.10.1}, taking into account also the subleading terms.) 
 If $\epsilon=1$, then there will, in addition, be contributions from the $\mathcal O$-terms in \eqref{eq:mg:prop32}. Note that, if desired, all of these contributions can be computed explicitly by following the steps above but without discarding the subleading terms.

The third case follows from \eqref{eq:mg:prop31.5} from Proposition~\ref{prop:prop.10.3}, with the $|u|$-term coming again from the data contribution $\mathrm{data}_{L-j}$, which contains no cancellations in view of Proposition~\ref{prop:mg:prop4}.

The fourth case follows in the same way from \eqref{eq:mg:prop31} from Proposition~\ref{prop:prop.10.3}, with the $|u|^{L-p-j}$-term coming again from the data contribution $\mathrm{data}_{L-j}$, which contains no cancellations in view of Proposition~\ref{prop:mg:prop4}.

More precisely, we have the following
\begin{cor}\label{cor10}
Let $\phi_L$ and $\mathcal D$ be as in Theorem \ref{thm:moreg}, and recall that $r_0:=r(u,1)=|u|-2M\log |u|+\mathcal O(1)$, as well as the constants $S_{L,p,j,n}$ defined in \eqref{eq:mg:SLPj-n}.

1.) If \underline{$L<p$ and $j\geq 0$}, or if \underline{$j=L-p\geq 0$}, then we have throughout $\mathcal D$:
\begin{equation}\label{eq:mg:cor1}
[r^2\pv]^{L-j}(r\phi_L)=\ci  S_{L,p,j,0}\cdot  r_0^{L-p-j}(1+\mathcal O(|u|r^{-1}+|u|^{-\epsilon})).
\end{equation}

2.) If \underline{$L\geq p$ and $j\geq L-p+1$}, then we have throughout $\mathcal D$:
\begin{equation}\label{eq:mg:cor2-1}
[r^2\pv]^{L-j}(r\phi_L)=\mathcal O (r_0^{L-p-j-\epsilon})+\mathcal O(r^{-1}|u|^{L-p-j+1}).
\end{equation}
Indeed, if we suppose instead of \eqref{eq:mg:ass1} that
\begin{equation}
|r\phi-\ci r^{-p}+\ci{}_{,\epsilon} r^{-p-\epsilon}|\leq C r^{-p-\epsilon'}
\end{equation}
for some constants $C, \ci{}_{,\epsilon}$ and for some $0<\epsilon\leq 1<\epsilon'$, then we have
\begin{equation}\label{eq:mg:cor2}
[r^2\pv]^{L-j}(r\phi_L)=\tilde C \cdot  r_0^{L-p-j-\epsilon} +\mathcal O(|u|^{L-p-j-1})+\mathcal O(r^{-1}|u|^{L-p-j+1})
\end{equation}
for some constant $\tilde C=\tilde C(L,p,j,\epsilon,M,\ci,\ci{}_{,\epsilon})$ which we can compute explicitly.
 
 3.) If \underline{$L\geq p$ and $j=L-p-1$}, then we have throughout $\mathcal D$:
\begin{equation}\label{eq:mg:cor31}
[r^2\pv]^{p+1}(r\phi_L)=C'_1 \cdot (\log r-\log |u|)+\ci  S_{L,p,L-p-1,0}\cdot  r_0 +\mathcal O(|u|^{1-\epsilon}), 
\end{equation}
with the constant $C'_1$ being given by
\begin{flalign}
C'_1=(-1)^L\cdot 2M(2x_1^{(L)}-c_0^L)\cdot \frac{p!(L-p)(L+1+p)}{(2L+2)(2L+1)}\cdot \ci.
\end{flalign}

4.) If \underline{$L>p$ and $0\leq j\leq L-p-2$}, then we have throughout $\mathcal D$:
 \begin{nalign}
[r^2\pv]^{L-j}(r\phi_L)=C'_2  \cdot r^{L-1-p-j}+\ci  S_{L,p,j,0} \cdot r_0^{L-p-j} +\mathcal O(|u|^{L-p-j-\epsilon})+\mathcal O(r^{L-2-p-j}|u|), &&
\end{nalign}
with the constant $C'_2$ being given by
$C'_2=(L-2-p-j)!\cdot C'_1$.
\end{cor}

\begin{proof}
The proof is obtained by combining the results of Propositions~\ref{prop:prop.10.1}--\ref{prop:mg:prop4} in the manner described above (below \eqref{cases}). Notice that we expressed the constant $\cv{L}$ appearing in \eqref{eq:mg:prop31} and \eqref{eq:mg:prop31.5} (and defined in \eqref{eq:mg:cv}) as
$\cv{L}=(-1)^L \frac{p!(2L)!}{(L+p)!}\cdot C_{\mathrm{in}}$, which follows from \eqref{eq:mg:product}.
\end{proof}
\subsubsection*{Proof of Theorem~\ref{thm:moreg}}\label{sec:mg:proofofTHM}
\begin{proof}[Proof of Theorem~\ref{thm:moreg}]
Part {a)} of the theorem follows directly from Corollary~\ref{cor10}, with $C_0$ given by $C_0=S_{L,p,L,0}$.

The first part of b) follows by dividing \eqref{eq:mg:cor31} by $r^2$, integrating from $\mathcal I^+$, and repeating the procedure $p-1$ terms. The boundary terms $\lim_{v\to\infty}[r^2\pv]^{L-p-i}(u,v)$ we pick up with each integration are estimated via \eqref{eq:mg:cor1} or \eqref{eq:mg:cor2-1}, thus giving rise to the functions $f_i^{(L,p)}$ and their leading-order behaviour. This proves \eqref{eq:mg:thmb1}, with the constants $C_1, C_2$ given by
\begin{align}\label{auxref}
MC_1=\frac{(-1)^p}{p!}C_1', && C_2=\frac{(-1)^p}{p!} S_{L,p,L-p-1,0}\cdot\ci.
\end{align}

The second part of {b)} follows similarly: We take \eqref{eq:mg:cor1} with $j=0$ and integrate $L-1$ times from $\mathcal I^+$, using at each step either \eqref{eq:mg:cor1} or \eqref{eq:mg:cor2-1} to estimate the boundary terms on $\mathcal I^+$. This alone only gives an expansion of $r^2\pv(r\phi_L)$ up to $r^{1-L}$. The higher-order behaviour can be obtained by  also taking into account the estimate for $\pv[r^2\pv]^L(r\phi_L)$ implied by the estimates \eqref{eq:mg:LHS} and \eqref{eq:mg:yy}. One can obtain expressions for the constants $C_3^{L,p,i}$ in much the same way as for $C_1,C_2$, using also \eqref{eq:mg:LHS} and \eqref{eq:mg:yy} for $i=L$.

Finally, part c) of Theorem~\ref{thm:moreg} follows by (in the case $p=L+1$ a slightly modified version of) \eqref{eq:mg:yy}. We have for $p\leq L$ (recall $\cv{L}=(-1)^L \frac{p!(2L)!}{(L+p)!}\cdot \ci$):
\begin{equation}
I_{\ell=L}^{\mathrm{future}, r^{2+p-L}}[\phi]=2M\cv{L}(2x_1^{(L)}-c_0^L)\frac{(L-p)!}{(2L+2)\cdots (L+2+p)},
\end{equation}
and for $p=L+1$:
\begin{equation}
I_{\ell=L}^{\mathrm{future}, \frac{\log r}{r^3}}[\phi]=2M\cv{L}(2x_1^{(L)}-c_0^L).
\end{equation}
This concludes the proof of Theorem~\ref{thm:moreg}.
\end{proof}
\subsection{Comments: More severe modifications to Price's law}\label{sec:10comments}
We have already discussed in detail in \S \ref{sec:gnl:comments} that we expect to obtain a \textit{logarithmically modified Price's law} for each $\ell$ provided that one smoothly extends the data to the event horizon and that $p=L+1$ in \eqref{eq:mg:ass1}, which, in turn, is the decay predicted by the results of \S \ref{sec:general:timelike}.
On the other hand, in view of equation \eqref{eq:mg:thmc1}, one can expect that the modification to Price's law is much more severe for $p\leq L>0$: 
Indeed, we expect that if $p\leq L$, then one obtains asymptotics near $i^+$ that are $L-p+1$ powers worse than in the case of smooth compactly supported data, i.e., we expect that
\begin{equation}
r\phi_L|_{\mathcal{I}^+}\sim u^{-L-2+(L-p+1)}=u^{-1-p},\quad \phi_L|_{\mathcal{H}^+}\sim v^{-2L-3+(L-p+1)}=v^{-L-p-2}
\end{equation}
near $i^+$.
The reader should compare this to the behaviour of the $\ell=0$-mode for $p=1$, 
\begin{equation}
r\phi_0|_{\mathcal{I}^+}\sim u^{-2}\log u,\quad \phi_0|_{\mathcal{H}^+}=v^{-3}\log v,
\end{equation}
which was proved in~\cite{II}.
 
We again refer the reader to \S \ref{sec:intro:NP} and Conjecture~\ref{conj2} therein for a more detailed discussion.

\section*{Acknowledgements}\addcontentsline{toc}{section}{Acknowledgements}
The author would like to thank Dejan Gajic, Mihalis Dafermos and Hamed Masaood for many helpful and stimulating discussions and comments on drafts of the present paper.
\newpage

\appendix
\section{Appendix}\label{sec:app}
This appendix contains various proofs which have been omitted in the main body of the paper.
\subsection{Proofs of Propositions~\ref{prop:gl:dvcommute} and~\ref{prop:gl:ducommute}}\label{A1}
\begin{proof}[Proof of Proposition~\ref{prop:gl:dvcommute}]
We prove \eqref{eq:gl:dvcommute} by induction, noting that it is true for $N=0$ with $a_0^0=0$, $b_0^0=1=c_0^0$. Assume now that \eqref{eq:gl:dvcommute} holds for a fixed $N$. 
We have 
\begin{multline}\label{appeq1}
\pv\pu[r^2\pv]^{N+1}(r\phi)=\pv\left(r^2\pu\pv[r^2\pv]^N(r\phi)\right)+\pv \left(-2r^2\cdot\frac{D}{r}\pv[r^2\pv]^N(r\phi)\right)\\
=\pv\left(r^2\pu\pv[r^2\pv]^N(r\phi)\right)-\frac{2D}{r}\pv[r^2\pv]^{N+1}(r\phi)+\left(1-\frac{4M}{r}\right)\frac{2D}{r^2}[r^2\pv]^{N+1}(r\phi).
\end{multline}
Using the induction hypothesis, we compute the first term on the RHS according to
\begin{align}
\begin{split}
&\pv\left(r^2\pu\pv[r^2\pv]^N(r\phi)\right)\\
=&-2N\pv\left(r^2\frac{D}{r}\pv[r^2\pv]^N(r\phi)\right)
+\pv\left(\sum_{j=0}^ND(2M)^j\left((a_j^n+b_j^N\slashed{\Delta}_{\mathbb{S}^2}-c_j^N\cdot\frac{2M}{r}\right)[r^2\pv]^{N-j}(r\phi)\right)\\
=&\frac{-2ND}{r}\pv[r^2\pv]^{N+1}(r\phi)+\frac{2ND}{r^2}\left(1-\frac{4M}{r}\right)\cdot[r^2\pv]^{N+1}(r\phi)\\
		&+\sum_{j=0}^N\frac{D(2M)^j}{r^2}\left(a_j^n+b_j^N\slashed{\Delta}_{\mathbb{S}^2}-c_j^N\cdot\frac{2M}{r}\right)[r^2\pv]^{N-j+1}(r\phi)\\
		&+\sum_{j=0}^N
\frac{D(2M)^{j+1}}{r^2}\left(a_j^n+b_j^N\slashed{\Delta}_{\mathbb{S}^2}-c_j^N\cdot\frac{2M}{r}+c_j^N-\frac{2M}{r}c_j^N\right)[r^2\pv]^{N-j}(r\phi)\\
=&\frac{-2ND}{r}\pv[r^2\pv]^{N+1}(r\phi)\\
	&+\sum_{j=1}^{N+1}\frac{D(2M)^j}{r^2}\left((a_{j-1}^N\!+\!c_{j-1}^N\!+\!a_j^{N})+(b_{j-1}^N\!+b_j^N)\slashed{\Delta}_{\mathbb{S}^2}-(c_{j}^N\!+\!2c_{j-1}^N)\cdot\frac{2M}{r}\right)[r^2\pv]^{N+1-j}(r\phi)\\
	&+\frac{D}{r^2}\left((a_0^N+2N)-b_0^N\slashed{\Delta}_{\mathbb{S}^2}-(c_0^N+4N)\cdot\frac{2M}{r}\right)[r^2\pv]^{N+1}(r\phi),
\end{split}
\end{align}
where we defined $a_j^N, b_j^N, c_j^N:=0$ for $j>N$.

Plugging the above equation back into \eqref{appeq1}, we thus obtain
\begin{multline}
\pu\pv[r^2\pv]^{N+1}(r\phi)=-\frac{2D(N+1)}{r}\pv[r^2\pv]^{N+1}(r\phi)\\+
\sum_{j=0}^{N+1}\frac{D}{r^2}(2M)^j\left(a_j^{N+1}+b_j^{N+1}\slashed{\Delta}_{\mathbb{S}^2}-c_j^{N+1}\cdot\frac{2M}{r}\right)[r^2\pv]^{N+1-j}(r\phi),
\end{multline}
with 
\begin{align}\label{appeq4}
a_0^{N+1}=a_0^N+2(N+1),&&b_0^{N+1}=b_0^N,&&c_0^{N+1}=c_0^N+4(N+1),
\end{align}
and the additional relations
\begin{align}
a_j^{N+1}&=a_j^N+a_{j-1}^N+c_{j-1}^N\label{app:a},\\
b_j^{N+1}&=b_{j}^N+b_{j-1}^N,\label{app:b}\\
c_{j}^{N+1}&=c_j^N+2c_{j-1}^N.\label{app:c}
\end{align}
This proves equation \eqref{eq:gl:dvcommute}. 

To find the explicit expressions for $a_j^N, b_j^N$ and $c_j^N$, we need to solve the recurrence relations above.
From \eqref{appeq4}, we read off that
\begin{align*}
a_0^{N}=N(N+1),&&b_0^N=1, &&c_0^N=1+2N(N+1)
\end{align*}
for all $N\geq 0$.
From this and \eqref{app:b}, we can then read off that $b_j^N=\binom{N}{j}$ for all $N \geq 0$ and $0\leq j\leq N$. Similarly, by writing $\tilde{c}_j^N=2^{-j}c_j^N$, we find from \eqref{app:c} that 
\[\tilde{c}_j^N=\binom{N}{j}+4\binom{N+1}{j+2}.\]
Finally, plugging in the expression for $c_j^N$ into \eqref{app:a}, one finds
\[a_j^N=(2^j-1)\binom{N}{j}+(2^{j+2}-2)\binom{N+1}{j+2}.\]
This completes the proof of Proposition~\ref{prop:gl:dvcommute}.
\end{proof}
\begin{proof}[Proof of Proposition~\ref{prop:gl:ducommute}]
The proof follows along the same lines as the previous one, with the difference that one now obtains the linear system
\begin{align}
\underline{a}_0^{N+1}=\underline a_0^N+2(N+1),&&\underline b_0^{N+1}=\underline b_0^N,&&\underline c_0^{N+1}=\underline c_0^N+4(N+1),
\end{align}
and the additional relations
\begin{align}
\underline a_j^{N+1}=\underline a_j^N-\underline a_{j-1}^N-\underline c_{j-1}^N,&&
\underline b_j^{N+1}=\underline b_{j}^N-\underline b_{j-1}^N,&&
\underline c_{j}^{N+1}=\underline c_j^N-2\underline c_{j-1}^N.
\end{align}
One can relate this to the previous system \eqref{appeq4}--\eqref{app:c} by writing $\underline{ \tilde b}_j^N=(-1)^{j}\underline{b}_j^N$ etc.
\end{proof}

\subsection{Proof of Lemma~\ref{lemma:pu=t-pv}}\label{A2}
\begin{proof}
Equation \eqref{eq:lemmapp} clearly holds for $n=0$. Assume that it holds for a fixed $n$. We shall show that it also holds for $n+1$. Letting $[\cdot,\cdot]$ denote the usual commutator, we have
\begin{nalign}\label{applemmaproof}
&(r^2\pu)^{n+1} f=(r^2T-r^2\pv)(r^2T-r^2\pv)^nf\\
=&\sum_{k=0}^n (-1)^{n-k}\binom{n}{k}\left(\sum_{i=0}^{k-1}\alpha^{(n,k)}_i(r+\mathcal{O}(1))^i(r^2 T)^{k+1-i}\right)[r^2\pv]^{n-k}f\\
+&\sum_{k=0}^n (-1)^{n+1-k}\binom{n}{k}\left(\sum_{i=0}^{k-1}\alpha^{(n,k)}_i(r+\mathcal{O}(1))^i(r^2 T)^{k-i}\right)[r^2\pv]^{n+1-k}f\\
-&\sum_{k=0}^n (-1)^{n+1-k}\binom{n}{k}\left(\sum_{i=0}^{k-1}\alpha^{(n,k)}_i\left[r^2\pv,(r+\mathcal{O}(1))^i(r^2 T)^{k-i}\right]\right)[r^2\pv]^{n-k}f
\end{nalign}
by commuting $r^2\pv$ past the other terms.
We now write
\begin{align}
\begin{split}
&\left[r^2\pv,(r+\mathcal{O}(1))^i(r^2 T)^{k-i}\right]f\\
=&\left[r^2\pv,(r+\mathcal{O}(1))^i\right](r^2 T)^{k-i}f+(r+\mathcal{O}(1))^i\left[r^2\pv,(r^2 T)^{k-i}\right]f,
\end{split}
\end{align}
and compute
\begin{equation}
\left[r^2\pv,(r+\mathcal{O}(1))^i\right]f=iD r^{i+1}f+\mathcal{O}(r^i)f,
\end{equation}
as well as
\begin{align}
\begin{split}
\left[r^2\pv,(r^2 T)^{k-i}\right]f=(k-i)\left[r^2\pv,r^2 T\right](r^2T)^{k-i-1}f
=(k-i)2Dr (r^2T)^{k-i} f,
\end{split}
\end{align}
where we used that $[V_1,V_2^n]=n[V_1,V_2]V_2^{n-1}$, which holds true if $[[V_1,V_2],V_2]=0$. Plugging the above identities back into \eqref{applemmaproof}, one then recovers \eqref{eq:lemmapp} for $n+1$ in a standard way. This gives rise to recurrence relations for the $\alpha_i^{(n,k)}$,  which can be solved explicitly. We leave this to the interested reader.

Alternatively, we could have used the non-commutative binomial theorem to write
\begin{equation}
(r^2T-r^2\pv)^n f=\sum_{k=0}^n\binom{n}{k}\left(r^2T+[r^2\pv, \cdot]\right)^k(-r^2\pv)^{n-k} f,
\end{equation}
and computed the commutators directly.
\end{proof}
\subsection{Proof of Lemma~\ref{lemma9.1}}\label{A3}
\begin{proof}
The result clearly holds $N=2$. Let us now assume that \eqref{eq:lemma91} holds for a fixed $N$ and for all $N'<N-1$. Then it also holds for $N+1$ and for any $N'<N$. Indeed,
\begin{align}
N\int_{-\infty}^u\frac{r^{N'}}{|u'|^{N+1}}\dd u'=\int_{-\infty}^u \pu\left(	\frac{r^{N'}}{|u'|^{N}}	\right) +\frac{N'r^{N'-1}}{|u'|^{N}}\left(1-\frac{2M}{r}\right)\dd u',
\end{align}
and we can apply the induction assumption to the second term on the RHS to obtain the result (after some standard  shifting of indices).
\end{proof}
\subsection{Proof of Lemma~\ref{lem10.2}}\label{A4}
\begin{proof}
Let us inductively assume that
\begin{align}
{\int_{r_0}^{r}\frac{1}{r_{(k)}^2}\dots\int_{r_0}^{r_{(2)}}\frac{1}{r_{(1)}^2}\dd r_{(1)}\cdots\dd r_{(k)}}=\frac{1}{k!}\left(\frac{1}{r_0}-\frac{1}{r}\right)^k
\end{align}
for any $r\geq r_0>0$.
This holds true for $k=1$. Going from $k$ to $k+1$, we have
\begin{align}
{\int_{r_0}^{r}\frac{1}{r'^2}\int_{r_0}^{r'}\frac{1}{r_{(k)}^2}\dots\int_{r_0}^{r_{(2)}}\frac{1}{r_{(1)}^2}\dd r_{(1)}\cdots\dd r_{(k)}\dd r'}=\int_{r_0}^r \frac{1}{r'^2}\frac{1}{k!}\left(\frac{1}{r_0}-\frac{1}{r'}\right)^k\dd r'.
\end{align}
We write $x=1/r$ and $x_0=1/r_0$ in the integral above to obtain
\begin{align}
\int_{r_0}^r \frac{1}{r'^2} \frac{1}{k!}\left(\frac{1}{r_0}-\frac{1}{r'}\right)^k\dd r'=(-1)^{k+1}\int_{x_0}^x \frac{1}{k!}(x'-x_0)^k\dd x'=\frac{1}{(k+1)!}(x_0-x)^{k+1}
\end{align}
for any $r\geq r_0>0$.
This concludes the proof.
\end{proof}
\subsection{Proof of Lemma~\ref{prop10}}\label{A5}
\begin{proof}
For convenience, we recall the definition 
\begin{equation}
\mathrm{sum}(L,p,j):=\sum_{i=0}^j (-1)^{i} \frac{(2L-j+i)!}{i!(L+p-j+i)!(j-i)!}.
\end{equation}
We begin with the crucial observation that if $p=L$ and $j\geq 1$, 
then
\begin{align}
\mathrm{sum}(L,L,j)=\sum_{i=0}^j (-1)^{i} \frac{1}{i!(j-i)!}
=\sum_{i=0}^j\frac{1}{j!}\binom{j}{i}(-1)^i\cdot1^j=\frac{1}{j!}(1-1)^j=0.
\end{align}
The idea is to interpret the above sum as a function of some variable $x$ evaluated at $x=1$: Consider, for instance, the case where $p=L-1$. Then
\begin{align}
\begin{split}
\mathrm{sum}(L,L-1,j)=\sum_{i=0}^j (-1)^i\frac{2L+i-j}{i!(j-i)!}
=\left.\frac{\dd}{\dd x}\right|_{x=1}&\sum_{i=0}^j  (-1)^i\frac{1}{i!(j-i)!}x^{2L+i-j}\\
=(-1)^j\left.\frac{\dd}{\dd x}\right|_{x=1}&\left(x^{2L-j}\sum_{i=0}^j  (-1)^{j-i}x^i\frac{1}{i!(j-i)!}\right)\\
=(-1)^j\left.\frac{\dd}{\dd x}\right|_{x=1}&\left(x^{2L-j}\frac{(x-1)^j}{j!}\right).
\end{split}
\end{align}
Similarly, if $p=L-k$ for some $k\geq 0$, then we obtain inductively that
\begin{equation}
\mathrm{sum}(L,L-k,j)=(-1)^j\left.\frac{\dd^k}{\dd x^k}\right|_{x=1}\left(x^{2L-j}\frac{(x-1)^j}{j!}\right).
\end{equation}
We then compute the $k$-th derivative above using the Leibniz rule:
\begin{align}\label{A.23}
\frac{\dd^k}{\dd x^k}\left(x^{2L-j}\frac{(x-1)^j}{j!}\right)=\sum_{n=0}^k \binom{k}{n}\frac{(2L-j)!x^{2L-j-n}}{(2L-j-n)!}\cdot \frac{(x-1)^{j-(k-n)}}{(j-(k-n))!},
\end{align}
where we use the convention that $\frac{1}{(-n)!}=0$ for all $n\in\mathbb N^+$. In particular, upon evaluating the expression \eqref{A.23} at $x=1$, we get
\begin{align}
\left.\frac{\dd^k}{\dd x^k}\right|_{x=1}\left(x^{2L-j}\frac{(x-1)^j}{j!}\right)=\sum_{n=0}^k \binom{k}{n}\frac{(2L-j)!}{(2L-j-n)!}\cdot \delta_{k-n,j}=\binom{k}{k-j}\frac{(2L-j)!}{(2L-k)!},
\end{align}
 which proves the second formula \eqref{eq:prop102} of the proposition.

On the other hand, if $p=L+1$, we \textit{integrate} the corresponding sum (instead of differentiating):
\begin{align}
\begin{split}
\mathrm{sum}(L,L+1,j)&=\sum_{i=0}^j (-1)^i\frac{1}{i!(j-i)!(2L+1-j+i)}\\
&=\int_0^1\sum_{i=0}^j  (-1)^i\frac{1}{i!(j-i)!}x^{2L+i-j}\dd x =\int_0^1 x^{2L-j}\frac{(1-x)^j}{j!}\dd x,
\end{split}
\end{align}
which is manifestly positive. Equation \eqref{eq:prop101} then follows inductively.
\end{proof}

\small{\bibliographystyle{alpha} 
\bibliography{references_all}}

\end{document}